\documentclass[3p]{elsarticle}
\biboptions{sort&compress}
\usepackage{amsmath}
\usepackage{amsfonts}
\usepackage{amsthm}
\usepackage{graphicx}
\usepackage{graphics}
\usepackage{color}
\usepackage{bbold}
\newcommand{\idmat}{\mbox{$1 \hspace{-1.3 mm} 1$}}

\newcommand{\be}{\begin{equation}}
\newcommand{\ee}{\end{equation}}
\newcommand{\bea}{\begin{eqnarray}}
\newcommand{\eea}{\end{eqnarray}}
\DeclareMathOperator{\Tr}{\mathrm{Tr}}

\newtheorem{trm}{Theorem}
\newtheorem{lem}[trm]{Lemma}
\newtheorem{prop}[trm]{Proposition}
\newtheorem{cor}[trm]{Corollary}
\newtheorem{definition}{Definition}

\title{Single-cone real-space finite difference schemes \\for the  Dirac von Neumann equation}
\author{Magdalena Schreilechner}
\author{Walter P\"{o}tz\corref{cor1}}
\ead{walter.poetz@uni-graz.at}
\address{Institut f\"{u}r Physik, Karl-Franzens-Universit\"{a}t Graz, Universit\"{a}tsplatz 5, 8010 Graz, Austria}
\cortext[cor1]{Corresponding author}

\begin{document}

\begin{abstract}
Two finite difference schemes  for the numerical treatment of the von Neumann equation for the (2+1)D  Dirac Hamiltonian are presented.  Both utilize a single-cone staggered space-time grid which ensures a single-cone energy dispersion  to formulate  a numerical treatment of the mixed--state dynamics within the von Neumann equation.  The first scheme executes the time-derivative according to the product rule for "bra" and "ket" indices of the density operator.   It therefore directly inherits all the favorable properties of the difference scheme for the pure--state Dirac equation and conserves positivity.  The second scheme proposed here performs the time-derivative in one sweep.  This direct scheme is investigated regarding stability and convergence.  Both schemes are tested numerically for elementary simulations  using parameters which pertain to topological insulator surface states.  Application of the schemes to a Dirac Lindblad equation and real--space--time Green's function formulations are discussed.
\end{abstract}
\begin{keyword}
Dirac equation \sep  vonNeumann equation \sep finite-difference scheme  \sep staggered grid \sep fermion doubling \sep topological insulator 
\end{keyword} 

\maketitle

\section{Introduction - preliminaries and definitions}

\subsection{The Dirac equation and numerical schemes}

In context with graphene and,  more recently, topological insulator surface states, the Dirac equation has received renewed interest within the physics community.  
Introduced by P. A. M. Dirac in 1928, the (3+1)D Dirac equation traditionally has been a model in high-energy physics to describe relativistic spin-1/2 particles and, as a field theory,  is a key ingredient  to the standard model of elementary particle physics.\cite{dirac,sakurai,ryder,itzykson,srednicki}   In condensed matter and atomic  physics its  importance lies in the non-relativistc limit providing the spin-orbit interaction, which is instrumental to an understanding of atomic spectra and represents the foundation for the field of spintronics.\cite{spintr}  In electronic structure calculations the full Dirac equation has been used to describe inner shell electrons.\cite{wien1,wien2}   While the (1+1)D  and (2+1)D Dirac equation, allowing a two-dimensional representation of the Clifford algebra,  in the early days of quantum physics was used for simplicity in formal model studies, condensed matter physics has recently identified physical realizations of the (2+1)D Dirac equation in form of low energy electronic excitations in graphene and topological insulators (TI).\cite{neto,qi,hsieh,hasan,xia,analytis}

With the interest in graphene and TIs as components in electronic and spintronic devices efficient schemes for the numerical  solution of the (2+1)D Dirac equation have become desirable.  Numerical approaches for solving the Dirac equation have taken several approaches.  For investigations of relativistic electrons in atomic physics, a  (2+1)D FFT-split-operator code was used\cite{piazza,mocken}.  In such an approach, fast Fourier transformation between the space and momentum representation is used to compute the Dirac propagator. The computational effort scales like $\mathcal{O}(N \ln N)$ where $N$ is the number of grid-points. An efficient code using operator splitting in real space was introduced for the (3+1)D case.\cite{gourdeau}  It leads to the highly efficient operations count of $\mathcal{O}(N)$.   Alternatively, approaches using finite-difference approximations to the first-order space-time derivatives have been used and have foremostly been applied in lattice quantum chromo dynamics (QCD).\cite{QCD} Traditionally they were plagued by the fermion doubling problem associated with centered difference approximations to first-order space- and time-derivatives.\cite{Stacey} We note that such a 2d finite-difference scheme for the Dirac--Poisson system was presented recently.\cite{brinkman}

Recently we have identified a class of staggered grids on which fermion doubling can be avoided in arbitrary space dimensions, with explicit schemes given for the (1+1)D, (2+1)D, and (3+1)D case.\cite{hammer1d,hammer3d}  This class of schemes scales linearly in the number of grid points and allows a gauge-invariant formulation of electromagnetic fields  on the lattice.  Moreover, this grid can be used in any space-time formulation involving the Dirac operator, such as the density matrix and Green's function approach, while yielding a single Dirac cone.  

In this work we develop two numerical schemes for the mixed-state time evolution under the Dirac Hamiltonian $H$ in (2+1)D , which we term the Dirac-von-Neumann equation
\be
i\hbar {\dot \rho}=H\rho -  \rho H= \left[ H, \rho\right] ~ .
\label{vN}
\ee
It describes the coherent mixed state dynamics of Dirac fermions which is useful  for simulations of quantum transport on TI surfaces near the Dirac point.  Both schemes to be presented in the following sections utilize the staggered grid of Ref. \cite{hammer3d}  to define centered differences over a single lattice spacing, thereby eliminating the very source for fermion doubling from the start.    The first scheme ("bra-ket") scheme, introduced in Sect. \ref{braket},  treats the time derivative of the density operator $\rho=\sum_i \gamma_i|\phi_i\rangle\langle \phi_i|$ within the  product rule as indicated by the commutator on the r.h.s. of Eq. \ref{vN}:  apply $H$  from the left and therewith propagate the ket in time, apply $H$ from the right and therewith propagate the bra in time, and then form the difference.   The second ("direct") scheme treats the time-derivative of $\rho$ in one step and is introduced in Sect. \ref{directs}, with further details given in the Appendix.   Numerical examples, one for each scheme, are given in Sect. \ref{NumSim}.  Finally, we briefly discuss the extension of this approach to the Lindblad equation and space-time  formulations of the Green's function method for the Dirac Hamiltonian in Sect. \ref{EXT} and close with Summary and Outlook.  

\subsection{The continuum formulation of the von Neumann equation for the (2+1)D Dirac Hamiltonian}

We consider the effective model Hamiltonian which accounts for the energy spectrum of TI surface states near the Dirac point\cite{qi}\footnote{Note that the TI  representation differs from the standard representation used in Ref. \cite{hammer3d}, however, shifting  $P_x \rightarrow P_y$ and $P_y \rightarrow -P_x$ converts the latter to the former.}
\begin{equation}
H=v\left({\bf \sigma} \times {\bf P}\right)\mid_z + m(X,Y,t)\sigma_z +V(X,Y,t)\idmat_2
\label{H}
\end{equation}
$\sigma_i$ denote the Pauli matrices. 
Using the abbreviation\footnote{For now,  the electromagnetic vector potential is set equal to zero.}
$$
\partial_\pm=v\left(P_y\pm i P_x\right) 
$$
and (omitting space--time arguments)
$$
m_\pm=V \pm m
$$
this Hamilton operator takes the form of a  $2\times 2$ matrix
$$
H= \left( \begin{array}{cc} m_+ & \partial_+ \\
 \partial_- & m_- \end{array}\right) ~.
 $$
 Inserted into the von Neumann equation \eqref{vN} 
one obtains a set of first-order partial differential equation in space and time for the density operator elements $\rho_{i j}$, conveniently written as the  $2\times 2$ matrix identity
$$
i\hbar \frac{\partial }{\partial t} \left( \begin{array}{cc} \rho_{11} &  \rho_{12}\\  \rho_{21} & \rho_{22} \end{array}\right)= \left( \begin{array}{cc}  m_+ \rho_{11} -\rho_{11}m_+   + \partial_+\rho_{21}    - \rho_{12} \partial_-   &   m_+\rho_{12} - \rho_{12}m_- +  \partial_+\rho_{22} - \rho_{11}\partial_+ \\
m_-\rho_{21}   - \rho_{21}m_+   +  \partial_- \rho_{11} - \rho_{22} \partial_-   &  m_-\rho_{22}   -  \rho_{22}m_-   +  \partial_- \rho_{12}   -  \rho_{21}\partial_+ \end{array}\right)~~.
$$ 
Note that,  as in the original von Neumann equation, Hamilton matrix operator elements to the right of the density operator matrix elements act to the left (and vice versa).  The two-component nature of the spin-$1/2$ Dirac fermion suggests this $2\times 2$ form.  Choosing a continuous space representation for the orbital degrees of freedom one arrives at 
\begin{eqnarray}
i\hbar {\dot \rho}({\bf r},{\bf r'}) &= &    ~~~~~~~~~~~~~~~~~~~~~~~~~~~~~~~~~~~~~~~~~~~~~~~~~~~~~~~~~~~~~~~~~~~~~~~~~~~~~~~~~~~~~~~~~~~~~~~~~~~~~~~~~~~~~~~~~~~~~~ \label{DVN}
\end{eqnarray}
\be
\left( \begin{array}{cc}  (m_+({\bf r})-m_+({\bf r'}) )\rho_{11}({\bf r},{\bf r'})+  \partial_+\rho_{21}({\bf r},{\bf r'})    + \partial_-'\rho_{12}({\bf r},{\bf r'})   &   (m_+({\bf r})-m_- ({\bf r'}))\rho_{12}({\bf r},{\bf r'}) +  \partial_+\rho_{22}({\bf r},{\bf r'}) + \partial_+' \rho_{11}({\bf r},{\bf r'})\\ 
(m_-({\bf r})-m_+({\bf r'}))\rho_{21}({\bf r},{\bf r'})    +  \partial_- \rho_{11}({\bf r},{\bf r'}) +  \partial_-' \rho_{22}({\bf r},{\bf r'})  &  (m_-({\bf r})-m_-({\bf r'}))\rho_{22}({\bf r},{\bf r'})    +  \partial_- \rho_{12}({\bf r},{\bf r'})   + \partial_+'\rho_{21}({\bf r},{\bf r'}) \end{array}\right)~~ . \label{vNrsp}
\ee
Here we omit the time variable $t$, common to all terms, for brevity, and 
use the real-space versions of the abbreviations defined above, {\it i.e.}, 
$$
\partial_\pm=\frac{\hbar v}{i}\left(\partial_y\pm i \partial_x\right)~ , \partial_\pm'=\frac{\hbar v}{i}\left(\partial_{y'}\pm i \partial_{x'}\right)
$$
and (including arguments)
$$
m_\pm(x,y,t)=V(x,y,t) \pm m(x,y,t)~.
$$
Note that in the real-space representation we have
$$
 \langle \Psi_j \mid \left(P_y\pm i P_x\right) \mid x,y\rangle= \langle   \left(P_y\mp i P_x\right) \Psi_j \mid  x,y\rangle = 
\langle    x,y\mid \left(P_y\mp i P_x\right) \Psi_j \rangle^*=  - \partial_\pm  \Psi_j ( x,y)^* ~.
$$

\subsection{Placement onto a space-time grid}\label{placement}

The task at hand now is to develop a space-time finite-difference approximation to Eq. \ref{vNrsp} which avoids fermion doubling.  
This is achieved by using the staggered grid introduced in an earlier paper to accommodate the real-space density matrix in which upper and lower spinor component(s) are placed on two adjacent time-sheets.\cite{hammer3d}
The proper implementation is  found by inspection of the density operator for a pure state spinor.  In (2+1)D one has  a pure-state  ket $\mid\Psi \rangle=\left(\begin{array}{c}\psi_1\\ \psi_2\end{array}\right)$ and a ket-bra projector  for the density operator 
\begin{equation}
\rho= \mid\Psi \rangle \langle \Psi\mid = \left(\begin{array}{c}\psi_1\\ \psi_2\end{array}\right) \left(\begin{array}{cc}\psi_1^*& \psi_2^*\end{array}\right) = 
\left(\begin{array}{cc} \psi_1 \psi_1^* & \psi_1 \psi_2^* \\ \psi_2 \psi_1^* & \psi_2 \psi_2^*\end{array}\right) ~ . \label{pures}
\end{equation}
This shows that the Pauli indices  of $\rho$ $, 1$ and  $2$ respectively, take the face-centered rectangular $\psi_1$-grid ($ {\cal G}_1$) and $\psi_2$-grid ($ {\cal G}_2$)  position.  For a single-cone representation the latter are\cite{hammer3d}

\begin{align}
\psi_1(j) ~\mbox{with}~ j &~\in~ {\cal G}_1 =\left\{(j_x,j_y,j_o-1/2),~(j_x+1/2,j_y+1/2,j_o-1/2)\mid  j_\nu  \in\mathbb{Z}, \nu=x,y,o\right\}~,  \nonumber \\
\psi_2(j)  ~\mbox{with}~ j &~\in~  {\cal G}_2 =\left\{(j_x+1/2,j_y,j_o),~(j_x,j_y+1/2,j_o)\mid   j_\nu \in\mathbb{Z}, \nu=x,y,o\right\}. \nonumber 
\end{align}
Here the time index is labeled $o$.  Note that, for given time, $\psi_1$ and $\psi_2$ are placed onto adjacent time sheets, respectively, $j_0-1/2$ and $j_0$.  Each time sheet contains a rectangular face-centered spatial grid, symmetrically staggered relative to the ones on the two adjacent time sheets, such that symmetric difference quotients replace the respective partial derivatives on the grid.  The grid spacings are denoted by $\Delta_x$, $\Delta_y$,  and $\Delta_o = v\Delta_t$.   
We also introduce the  associated grids
\begin{align}
{\cal \bar{G}}_1 =\left\{(j_x,j_y,j_o),~(j_x+1/2,j_y+1/2,j_o)\mid  j_\nu  \in\mathbb{Z}, \nu=x,y,o\right\}~,  \nonumber \\
{\cal \bar{G}}_2 =\left\{(j_x+1/2,j_y,j_o+1/2),~(j_x,j_y+1/2,j_o+1/2)\mid   j_\nu \in\mathbb{Z}, \nu=x,y,o\right\}. \nonumber 
\end{align}
These grids are obtained from their non-barred counterparts by a forward time-shift by $\Delta_t/2$ and thus share the spatial positions with the former.  These will be  the grids on which we place mass, scalar potential, and all partial derivatives.  Common to all four grids is their symmetry in space-time, as required in a covariant theory.  Since most dynamic problems are formulated as an initial--value problem in time, however, it is useful to view each of the two  grids as a set of time sheets.  Each time sheet in turn  consists of a rectangular-,  for $\Delta_x\neq\Delta_y$,  or cubic-face-centered spatial grid,   for $\Delta=\Delta_x=\Delta_y$, which in units $(\Delta_x,\Delta_y)$ is characterized by  the set of relative positions 
$(j_x,j_y)\in \mathbb{Z}^2 \cup (\mathbb{Z}+\frac{1}{2})^2$.  

The matrix elements $\rho_{ij}$ (and the vector potential to be introduced later) live on the grids ${\cal G}_i$
$$
\rho_{ij}({\bf r},t; {\bf r'},t')\rightarrow \rho_{ij}({\cal  G}_i; {\cal  G}_j)  ~,~  i,j=1,2.
$$
Note that $t=t'$ leads to the need for two adjacent time sheets (e.g., $j_o-1/2$  and $ j_o$)  for the representation of all matrix elements of $\rho(t)$ on the lattice:  while $\rho_{11}$ and  $\rho_{22}$  each are placed on a single time sheet, 
$\rho_{12}$ and  $\rho_{21}$ require the use of two adjacent time sheets.  This leads one to the definition of the trace of $\rho$ using two adjacent time sheets of the grid:

\begin{definition} \label{def1}
The trace of $\rho$  at a given time ($j_o-1/2, j_o$) on the grid, 
\bea
\Tr\{ \rho\} &=& \Tr\{ \rho_{11}({\cal G}^{(j_o-1/2)}_1;{\cal G}^{(j_o-1/2)}_1) + \rho_{22}({\cal G}_2^{(j_o)};{\cal G}^{(j_o)}_2)\} \nonumber \\
&= &\sum_{(j_x,j_y)}\left[ \rho_{11}(j_x,j_y,j_o-\frac{1}{2};j_x,j_y,j_o-\frac{1}{2}) + \rho_{22}(j_x+\frac{1}{2},j_y,j_o;j_x+\frac{1}{2},j_y,j_o)\right] ~,\label{gridtrace}
\eea
is defined  as the sum of the trace of  $\rho_{11}$ with respect to the  lattice sites ${\cal G}^{(j_o-1/2)}_1$ and the trace of  $\rho_{22}$ with respect to the lattice sites ${\cal G}^{(j_o)}_2$.  
The summation runs over all lattice sites of the respective time sheet, $(j_x,j_y) \in \mathbb{Z}^2 \cup (\mathbb{Z}+\frac{1}{2})^2$.  

We extend this definition to $(j_x,j_y)$ lattice sums for non-diagonal density matrix elements
\bea
\Tr\{ \rho_{ij} ({\cal G};{\cal G}')\}&=& 
\sum_{(j_x,j_y)}\left[ \rho_{ij}(j_x,j_y,j_t; j_x',j_y',j_t')\right]\mid_{ j_x'=j_x+\Delta j_x,j_y'=j_y+\Delta j_y} ~,\label{gridtrace1}
\eea
with  constant  relative indices $\Delta j_x$ and $\Delta j_y$.
\end{definition}

\begin{definition} \label{def2}
Centered spatial difference operators   are defined on the barred grids as follows
$$
D_x f({\cal {\bar G}}_i): D_x f\mid^{j_t}_{j_x,j_y} = \frac{1}{\Delta_x} ( f_{j_x+1/2,j_y}^{j_t}-f_{j_x-1/2,j_y}^{j_t}),  ~~(j_x,j_y,j_t) \in {\cal {\bar G}}_i~,
$$
$$
D_y f({\cal {\bar G}}_i): ~D_y f\mid^{j_t}_{j_x,j_y} =\frac{1}{\Delta_y} ( f_{j_x,j_y+1/2}^{j_t}-f_{j_x,j_y-1/2}^{j_t}) ,  ~~(j_x,j_y,j_t) \in {\cal {\bar G}}_i~,.
$$ 
$$
D_\pm f({\cal {\bar G}}_i):D_\pm f\mid^{j_t}_{j_x,j_y} = D_y f\mid^{j_t}_{j_x,j_y}  \pm iD_x f\mid^{j_t}_{j_x,j_y}, ,  ~~(j_x,j_y,j_t) \in {\cal {\bar G}}_i ~,
$$ 
and
$$
D_o f({\cal {\bar G}}_i):~D_o f\mid^{j_t}_{j_x,j_y} = \frac{1}{\Delta_o} ( f_{j_x,j_y}^{j_t+1/2}-f_{j_x,j_y}^{j_t-1/2}),  ~~(j_x,j_y,j_t) \in {\cal {\bar G}}_i~.
$$
Furthermore, we will use 
 the time-average operation 
$$
{\cal T} f({\cal {\bar G}}_i):     {\cal T} f\mid^{j_t}_{j_x,j_y} = \frac{1}{2} ( f_{j_x,j_y}^{j_t+1/2}+f_{j_x,j_y}^{j_t-1/2}),  ~~(j_x,j_y,j_t) \in {\cal {\bar G}}_i~.
$$ 
\end{definition}

Formally these operators are defined on the barred grids ${\cal {\bar G}}_i$, however,  the execution involves density matrix elements on the unbarred grids.  
Note that when applied to the density matrix, these operators may act on the first (bra) or second set of indices (ket).

Partial derivatives, scalar potential and mass terms, respectively, are  implementing into Eq. (\ref{DVN}) by the substitutions
$$
\partial_\pm \rightarrow \frac{v\hbar}{i} D_\pm~ ,
$$
$$
m_\pm(x,y,t)\rightarrow \left[V({\cal {\bar G}}_i) \pm m({\cal  {\bar G}}_i)\right]{\cal T}=2iv\hbar M_\pm({\cal  {\bar G}}_i){\cal T}, ~M_\pm=\frac{m_\pm}{2i\hbar v}~.
$$

Within a sum over all lattice sites  $(j_x,j_y) \in \mathbb{Z}^2 \cup (\mathbb{Z}+\frac{1}{2})^2$ a spatial difference operation ($D'_k$, k=x,y) performed  on the second set of indices $i'$  of the density matrix $\rho(i;i')$ is  equivalent to minus the derivative ($-D_k$) performed on the first set $i$. 

\begin{cor}\label{cor1}
Integration by parts on the grid: 
Let $(j_x,j_y) \in \mathbb{Z}^2 \cup (\mathbb{Z}+\frac{1}{2})^2$. Let $D_k$ and  $D_k', ~k=x,y$ as in Def. \ref{def2}, respectively,  denote spatial difference operators with respect to the first and second set of indices of $\rho_{ij}$.  
Then, under periodic or zero spatial boundary conditions for $\rho$ and with Def. \ref{def1}
\be
\Tr\{D_k'\rho_{ij}( {\cal {G}}; {\cal {\bar G}}')\} =\sum_{(j_x,j_y)}D_k'\rho_{ij}( {\cal {G}}; {\cal {\bar G}}') = -\sum_{(j_x,j_y)}D_k\rho_{ij}({\cal {\bar G}}; {\cal {G}}')= -\Tr\{D_k\rho_{ij}({\cal {\bar G}}; {\cal {G}}') \}~, \mbox{ for } k\in\{x,y\}~.
\ee
\end{cor}

\begin{proof}
We give the proof for  $ i=1, j=2, k=y$ 

\begin{eqnarray}
\sum_{(j_x,j_y)}D_y'\rho_{12}( {\cal {G}}^{(j_o-1/2)}_1;{\cal {\bar G}}^{(j_o)}_{1} ) &=& \sum_{(j_x,j_y)}\frac{{\rho_{12}}_{j_x,j_y ;j_{x'},j_{y'}+ \frac{1}{2}}^{j_o-\frac{1}{2} ; j_{o} } -{\rho_{12}}_{j_x,j_y ;j_{x'},j_{y'}- \frac{1}{2}}^{j_o- \frac{1}{2} ; j_{o}} }{\Delta_{y'}}\nonumber \\
&=& \frac{1}{\Delta_{y'}}\sum_{(j_x,j_y)}{\rho_{12}}_{j_x,j_y ;j_{x'},j_{y'}+ \frac{1}{2}}^{j_o-\frac{1}{2} ; j_{o} }-\frac{1}{\Delta_{y'}}\sum_{(j_x,j_y)}{\rho_{12}}_{j_x,j_y ;j_{x'},j_{y'}- \frac{1}{2}}^{j_o- \frac{1}{2} ; j_{o}}\nonumber \\
&=& \frac{1}{\Delta_{y}}\sum_{{(j_x->j_x+1/2,},{j_y->j_y+1/2)}}{\rho_{12}}_{j_x-1/2,j_y-1/2 ;j_{x'}-1/2,j_{y'}}^{j_o-\frac{1}{2} ; j_{o} }\nonumber \\
&-&\frac{1}{\Delta_{y}}\sum_{(j_x->j_x+1/2, j_y->j_y-1/2)}{\rho_{12}}_{j_x-1/2,j_y+1/2 ;j_{x'}-1/2,j_{y'}}^{j_o- \frac{1}{2} ; j_{o}}\nonumber \\
&=&-\sum_{(j_x,j_y)}D_y\rho_{12}( {\cal {\bar G}}^{(j_o-1/2)}_2;{\cal {G}}^{(j_o)}_{2} )~. 
\label{proofD'}
\end{eqnarray}
We have used that  $\Delta_{y'}=\Delta_y$ and that,  under the trace $(j_{x'}-j_x)$ and $(j_{y'}-j_y)$ are constant.   
Note that a combined shift (carried out in the third step) of $j_x$ and $j_y$ is required to stay on the proper sub-grid ${\cal G}_i$.  
This mixes corner and face-center positions of a given  time sheet.  However, since the sum is to be taken over all grid points this "integration by parts" rule holds under the trace.  Note also how the barred grids used for the definition of center position of the spatial derivative are related to the index of the density matrix elements.  Other cases for $i$, $j$, and $k$ can be shown in the same fashion.  
\end{proof}

\section{The bra-ket scheme}\label{braket}

The bra-ket time-propagation scheme is obtained directly by adapting the scheme  for the pure-state Dirac equation, applied to the bra- and the ket-side of the density operator.  This corresponds to the interpretation 
\be
{\dot \rho} = \sum_k \gamma_k \left[| {\dot \Psi_k} \rangle\langle  \Psi_k |+|  \Psi_k  \rangle\langle {\dot \Psi_k}  |\right] = \frac{1}{i\hbar} \sum_k  \gamma_k \left[ (H |  \Psi_k \rangle)\langle  \Psi_k |-|  \Psi_k 
\rangle\langle (\Psi_k |H)  \right] ~.
\ee

\begin{definition}\label{purest}
Within the representation Eq. \eqref{H} the  pure state time--evolution of  Ref. \cite{hammer3d} from initial time $"-": j_0-1/2, j_0$ to final time $"+":  j_0+1/2, j_0+1$  may  be written  as
\be
\psi_1^+=\alpha \psi_1^- + \beta \psi_2^- \label{p1}
\ee
\be
\psi_2^+=\gamma \psi_2^- + \delta \psi_1^+~, \label{p2}
\ee
with the coefficients
$$
\alpha=\frac{\frac{1}{\Delta_o}+M_+}{\frac{1}{\Delta_o}-M_+}~, ~ \alpha^*=1/\alpha~,
$$
$$
\beta=-\frac{1}{\frac{1}{\Delta_o}-M_+}D_+~,~~\beta^*=-\frac{1}{\frac{1}{\Delta_o}+M_+}D_-~,
$$
 living on integer  time sheet  $j_o$ of grid ${\cal \bar{G}}_1$, and 
$$
\gamma=\frac{\frac{1}{\Delta_o}+M_-}{\frac{1}{\Delta_o}-M_-}~, ~~\gamma^*=1/\gamma~,
$$
$$
\delta=-\frac{1}{\frac{1}{\Delta_o}-M_-}D_- ~,~~ \delta^*=-\frac{1}{\frac{1}{\Delta_o}+M_-}D_+ ~,
$$
living on half-integer time sheet $j_o+1/2$ of grid ${\cal \bar{G}}_2$.  We have used $(M_\pm)^*=-M_\pm$ which holds for real-valued mass and scalar potential $V$.  
\end{definition}

Using the short-hand notation $\rho_{11'}=\langle x,y\mid \rho(t)\mid x',y'\rangle$, the progression by one full time step $\Delta_o$ is executed as follows (initial ($-$) and final ($+$) time is indicated by the respective superscript):

$$
\rho_{1^-{1'}^-} \rightarrow \left\{\begin{array}{l} \rho_{1^+{1'}^-} = \alpha \rho_{1^-{1'}^-} + \beta \rho_{2^-{1'}^-} \\ \rho_{1^-{1'}^+} ={\alpha'}^* \rho_{1^-{1'}^-} + {\beta'}^* \rho_{1^-{2'}^-} \end{array} \right.~,
$$
$$
\rho_{1^-{2'}^-} \rightarrow \left\{\begin{array}{l} \rho_{1^+{2'}^-} = \alpha \rho_{1^-{2'}^-} + \beta \rho_{2^-{2'}^-} \\ \rho_{1^-{2'}^+} ={\gamma'}^* \rho_{1^-{2'}^-} +{\delta'}^* \rho_{1^-{1'}^+} \end{array} \right.~,
$$
$$
\rho_{2^-{1'}^-} \rightarrow \left\{\begin{array}{l} \rho_{2^+{1'}^-} = \gamma \rho_{2^-{1'}^-} + \delta \rho_{1^+{1'}^-} \\ \rho_{2^-{1'}^+} ={\alpha'}^* \rho_{2^-{1'}^-} +{\beta'}^* \rho_{2^-{2'}^-} \end{array} \right.~,
$$
$$
\rho_{2^-{2'}^-} \rightarrow \left\{\begin{array}{l} \rho_{2^+{2'}^-} = \gamma \rho_{2^-{2'}^-} + \delta \rho_{1^+{2'}^-} \\ \rho_{2^-{2'}^+} ={\gamma'}^* \rho_{2^-{2'}^-} +{\delta'}^* \rho_{2^-{1'}^+} \end{array} \right. ~.
$$
This set of operations is followed by 
$$
\rho_{1^+{1'}^+} = \alpha \rho_{1^-{1'}^+} + \beta \rho_{2^-{1'}^+} ={\alpha'}^* \rho_{1^+{1'}^-} + {\beta'}^* \rho_{1^+{2'}^-} 
$$
$$
\rho_{1^+{2'}^+} = \alpha \rho_{1^-{2'}^+} + \beta \rho_{2^-{2'}^+}  ={\gamma'}^* \rho_{1^+{2'}^-} +{\delta'}^* \rho_{1^+{1'}^+}  
$$
$$
 \rho_{2^+{1'}^+} = \gamma \rho_{2^-{1'}^+} + \delta \rho_{1^+{1'}^+}  ={\alpha'}^* \rho_{2^+{1'}^-} +{\beta'}^* \rho_{2^+{2'}^-}  
$$
$$
\rho_{2^+{2'}^+} = \gamma \rho_{2^-{2'}^+} + \delta \rho_{1^+{2'}^+} ={\gamma'}^* \rho_{2^+{2'}^-} +{\delta'}^* \rho_{2^+{1'}^+}  
$$
Executed in this order the scheme is explicit and follows exactly the single-cone time-propagation of Ref. \cite{hammer3d} applied to both sets of indices  of the density matrix, with unprimed and primed operators acting on, respectively, bra and ket.  As such all properties of the scheme for the pure-state propagation, such as  stability, convergence, and spectral properties, are carried over to the scheme for the density matrix.

\begin{definition}\label{bra-ket}
With the abbreviations defined above, the single time-step progression under the bra--ket scheme is defined as
\be
\rho_{1^+{1'}^+} =  \alpha {\alpha'}^* \rho_{1^-{1'}^-} + \alpha{\beta'}^* \rho_{1^-{2'}^-} + \beta {\alpha'}^* \rho_{2^-{1'}^-} +\beta{\beta'}^* \rho_{2^-{2'}^-} , \label{BK1}
\ee
\be
\rho_{1^+{2'}^+} = \alpha {\gamma'}^* \rho_{1^-{2'}^-} +\beta{\gamma'}^* \rho_{2^-{2'}^-}+ {\delta'}^* \rho_{1^+{1'}^+} ~,
\ee
\be
 \rho_{2^+{1'}^+} = \gamma {\alpha'}^* \rho_{2^-{1'}^-} +\gamma{\beta'}^* \rho_{2^-{2'}^-}+ \delta  \rho_{1^+{1'}^+}  ~,
\ee
\be
\rho_{2^+{2'}^+} ={\gamma'}^*\left(\gamma  + {\delta}{\beta}\right) \rho_{2^-{2'}^-} +\delta\alpha {\gamma'}^*\rho_{1^-{2'}^-} + {\delta'}^*\rho_{2^+{1'}^+} =\\
\gamma\left({\gamma'}^*  + {\delta'}^*{\beta'}^*\right) \rho_{2^-{2'}^-} +\gamma{\delta'}^*{\alpha'}^*\rho_{2^-{1'}^-} + \delta\rho_{1^+{2'}^+} 
 ~. \label{BK4}
\ee
\end{definition}
Def. \ref{bra-ket} extends the single-cone pure-state time-propagation scheme in Def. \ref{purest} to the one for mixed-states.  For pure states, the latter reduces to the former. 

The following lemmas and propositions pertaining to the bra-ket scheme are a direct consequence of the properties of the pure state scheme Ref. \cite{hammer3d} and the observation that the density operator at initial time, without loss of generality, 
can be written as as sum of pure-state contributions of the form Eq. \ref{pures}.

\begin{lem}\label{L2}
The bra-ket scheme conserves positivity and Hermiticity of the density matrix.
\end{lem}

\begin{proof}
Both properties are a direct result of the bra--ket scheme:  Let, at initial time the density matrix $\rho_o=\sum_k \gamma_k \mid\Psi_k\rangle  \langle \Psi_k \mid$ be positive definite, i.e., 

$$
\langle \Phi\mid \rho_o \mid \Phi\rangle =\sum_{\i, \j} \langle \Phi\mid i \rangle \rho_o(i;j) \langle j \mid \Phi\rangle \ge 0 
$$
for an arbitrary element $\mid \Phi\rangle$ of the Hilbert space. 
Abbreviating the one-step time evolution of a pure state Eqs. \eqref{p1} and \eqref{p2} by $\Psi^+=K \Psi^-$, 
\bea
K=\left( \begin{array}{cc} \alpha & \beta \nonumber\\ \delta \alpha  & (\gamma +\delta\beta) \end{array}\right)
\eea
the bra-ket scheme propagates 
$$
\langle \Phi\mid \rho_o \mid \Phi\rangle \rightarrow  \langle \Phi\mid K\rho_o K^\dagger \mid \Phi\rangle = \langle \Phi'\mid \rho_o \mid \Phi'\rangle \ge 0 , ~ \mid \Phi'\rangle=K^\dagger \mid \Phi\rangle  
$$
since positivity holds for $\rho_o$.   Hence under the bra-ket scheme positivity is preserved  for arbitrary time step.   Preservation of Hermiticity is seen from  $\left(K\rho_o K^\dagger\right)^\dagger= K\rho_o K^\dagger$. 
\end{proof}

Note that the standard norm is not conserved within the pure-state scheme since  $K$ is not unitary.\cite{hammer3d}  Hence, although the time-evolution superoperator can be cast in Kraus form with a single Kraus operator (living on two adjacent time-sheets), we do not have $K^\dagger K=1$.\cite{choi,kraus}

\begin{lem}\label{L1}
Under periodic, respectively, zero boundary conditions, the functional
\be
E_\mathbf{r}^0 = \Tr\{ \rho_{11}({\cal G}^{(j_o-1/2)}_1;{\cal G}^{(j_o-1/2)}_1) + \rho_{22}({\cal G}_2^{(j_o)};{\cal G}^{(j_o)}_2)+
\Delta_o  \Re\big[D_- 
\rho_{12}({\cal {\bar G}}^{(j_o-1/2)}_2;{\cal G}^{(j_o)}_2)\big] \} , \label{tr2} 
\ee
with ${\bf r}=(r_x=\Delta_o/\Delta_x, r_y=\Delta_o/\Delta_y)$,  is conserved under the bra-ket time propagation, Eqs.\ref{BK1} to \ref{BK4}.
\end{lem}

\begin{proof} 
The proof of this lemma as well as the following two propositions rest on the fact that the density matrix at initial time may be cast into diagonal form and then, under the scheme, is placed onto the grid according to

\be
\rho= \sum_k \gamma_k \mid\Psi ^{(k)}\rangle \langle \Psi^{(k)}\mid =   
\sum_k \gamma_k \left(\begin{array}{cc} \psi_1^{(k)} {\psi_1^{(k)}}^* & \psi_1^{(k)} {\psi_2^{(k)}}^* \\ \psi_2^{(k)} {\psi_1^{(k)}}^* & \psi_2^{(k)} {\psi_2^{(k)}}^*\end{array}\right) \rightarrow \sum_k \gamma_k \left(\begin{array}{cc} \psi_1^{(k)}({\cal  G}_1){\psi_1^{(k)}({\cal  G}_1)}^* & \psi_1^{(k)}({\cal  G}_1) {\psi_2^{(k)}({\cal  G}_2})^* \\ \psi_2^{(k)}({\cal  G}_2) {\psi_1^{(k)}({\cal  G}_1})^* & \psi_2^{(k)}({\cal  G}_2) {\psi_2^{(k)}({\cal  G}_2})^*\end{array}\right) , \label{mixeds}
\ee
where $\mid \Psi^{(k)}\rangle $ are normalized to one, $0\leq \gamma_k \leq 1$ real, and $\sum_k \gamma_k =1$.  Since each term in this sum is propagated individually within this scheme, all the stability properties of the pure-state dynamics apply.  The  
conserved functional for the individual pure-state contribution,  ${E_{\mathbf{r}}^0}^{(k)}$,  and the definition of the trace on the grid  in Eq. \eqref{gridtrace} as a sum over all lattice sites of the respective time-sheet  render the conserved functional for the density operator  $\sum_k \gamma_k {E_{\mathbf{r}}^0}^{(k)}=E_\mathbf{r}^0$ .
\end{proof}

\begin{prop}\label{P1}
Let $r_x=\frac{\Delta_o}{\Delta_x}$, $r_y=\frac{\Delta_o}{\Delta_y}$ with $r_x+r_y< 1$ (e.g.~using $r_x=r_y <  1/2$). Then, the bra-ket scheme Eqs.\ref{BK1} to \ref{BK4} is stable and satisfies the estimate
\be
\Tr\{ \rho_{11}({\cal G}^{(j_o-1/2)}_1;{\cal G}^{(j_o-1/2)}_1) + \rho_{22}({\cal G}_2^{(j_o)};{\cal G}^{(j_o)}_2)\} \leq \frac{E^0_\mathbf{r}}{1- r_x - r_y} 
\label{stab1}
\ee
for all time.
\end{prop}
\begin{proof}
Using the decomposition Eq. \eqref{mixeds} stability can be shown term by term in the sum.  
\end{proof}
Comment: Conservation of $E_\mathbf{r}^0$ on the grid corresponds to the conservation of  the trace of $\rho$ in the continuum limit.  At the same time the conservation of the former within this scheme implies non-conservation  of the trace when the latter is defined  on the grid according to Eq. \eqref{gridtrace}.   This non-conservation is of order $\Delta_o/\Delta$ and thus can be adjusted by  this ratio.

From the analysis of the pure-state scheme we know that Prop.  \ref{P1} is too restrictive.  In fact, the stability condition  $r_x^2+r_y^2 \leq 1$ holds for constant mass and potential.  For the special case $r=r_x=r_y$ this  less restrictive stability condition reads $r \leq 1/\sqrt{2}$. This stability condition also holds for an arbitrary space- and time-dependent mass and potential.

\begin{definition}
We define the functionals\footnote{For compactness of notation  the diagonal time-index is placed onto $\rho$ as a superscript.}
\bea
{\tilde \Tr}(\rho_{11}^{j_o-\frac{1}{2}})& = & \frac{1}{2}\sum_{ j_x,j_y} \left[ \rho_{11}^{j_o-\frac{1}{2}}( j_x,j_y; j_x,j_y) + \frac{1}{2}\left[ \Re\bigg\{ \rho_{11}^{j_o-\frac{1}{2}}( j_x,j_y+1; j_x,j_y) \bigg\}  +  \Re\bigg\{ \rho_{11}^{j_o-\frac{1}{2}}( j_x+\frac{1}{2},j_y+\frac{1}{2}; j_x-\frac{1}{2},j_y+\frac{1}{2}) \bigg\} \right. \right.  \nonumber  \\
 &   +  &  \left. \left.
 \Im\bigg\{ \rho_{11}^{j_o-\frac{1}{2}}( j_x+\frac{1}{2},j_y+\frac{1}{2}; j_x,j_y) \bigg\} + \Im\bigg\{ \rho_{11}^{j_o-\frac{1}{2}}( j_x+\frac{1}{2},j_y+\frac{1}{2}; j_x,j_y+1) \bigg\}  \right. \right. \\ \nonumber 
  &   +  &  \left. \left. \Im\bigg\{ \rho_{11}^{j_o-\frac{1}{2}}( j_x-\frac{1}{2},j_y+\frac{1}{2}; j_x,j_y) \bigg\}+  \Im\bigg\{ \rho_{11}^{j_o-\frac{1}{2}}( j_x-\frac{1}{2},j_y+\frac{1}{2}; j_x,j_y+1) \bigg\}  \right] \right]  \label{TRU}
\eea
\bea
{\tilde \Tr}(\rho_{22}^{j_o})& = & \frac{1}{2}\sum_{ j_x,j_y} \left[ \rho_{22}^{j_o}( j_x,j_y-\frac{1}{2}; j_x,j_y-\frac{1}{2}) + \frac{1}{2}\left[ \Re\bigg\{ \rho_{22}^{j_o}( j_x,j_y+\frac{1}{2}; j_x,j_y-\frac{1}{2}) \bigg\}  +  \Re\bigg\{ \rho_{22}^{j_o}( j_x+\frac{1}{2},j_y; j_x-\frac{1}{2},j_y) \bigg\} \right. \right.  \nonumber  \\
 &   +  &  \left. \left.
 \Im\bigg\{ \rho_{22}^{j_o}( j_x+\frac{1}{2},j_y; j_x,j_y-\frac{1}{2}) \bigg\} + \Im\bigg\{ \rho_{22}^{j_o}( j_x+\frac{1}{2},j_y; j_x,j_y+\frac{1}{2}) \bigg\}  \right. \right. \\ \nonumber 
  &   +  &  \left. \left. \Im\bigg\{ \rho_{22}^{j_o}( j_x-\frac{1}{2},j_y; j_x,j_y-\frac{1}{2}) \bigg\}+  \Im\bigg\{ \rho_{22}^{j_o}( j_x-\frac{1}{2},j_y; j_x,j_y+\frac{1}{2}) \bigg\}  \right] \right]   ~.\label{TRV}
\eea

Again, the summation runs over all lattice sites of  time sheet ${\cal G}_l^{j_l}$, i.e., $j_x,j_y \in \mathbb{Z}^2 \cup (\mathbb{Z}+\frac{1}{2})^2$, with $j_1=j_0-\frac{1}{2}$ and $j_2=j_0$ and $j_o \in \mathbb{Z}$.
\end{definition}

\begin{prop}\label{P2}
Let $r=r_x=r_y=1/\sqrt{2}$ hold in \eqref{BK1}-\eqref{BK4}. Then this scheme is stable and for all time satisfies 
\be
  {\tilde \Tr}(\rho_{11}^{j_o-1/2}) + {\tilde \Tr}(\rho_{22}^{j_o})  \le 2E^0~, \mbox{ with } E^0= E^0_{1/\sqrt{2},1/\sqrt{2}}~.\label{BKL3}
\ee
\end{prop}
\begin{proof}
This inequality has been shown to hold for pure states based on the definition of an averaged norm.\cite{hammer3d}  In the present representation of the Dirac Hamiltonian Eq. \eqref{H} one defines for  upper  and lower spinor component $\psi_1$ and $\psi_2$, respectively, 
\be
\left\|\tilde{\psi_1^{j_o-\frac{1}{2}}}\right\|^2 := \sum_{(j_x,j_y)} \bigg|\frac{{\psi_1^{j_o-\frac{1}{2}}}_{j_x+\frac{1}{2} , j_y+\frac{1}{2}}-i {\psi_1^{j_o-\frac{1}{2}}}_{j_x+1,j_y}}{2\sqrt{2}} + \frac{{\psi_1^{j_o-\frac{1}{2}}}_{j_x+\frac{1}{2},j_y-\frac{1}{2}}-i {\psi_1^{j_o-\frac{1}{2}}}_{j_x ,j_y}}{2\sqrt{2}}\bigg|^2 ~,
\ee
and 

\be
\left\|\tilde{\psi_2^{j_o}}\right\|^2 := \sum_{(j_x,j_y)} \bigg|\frac{{\psi_2^{j_o}}_{j_x , j_y+\frac{1}{2}}-i {\psi_2^{j_o}}_{j_x+\frac{1}{2},j_y}}{2\sqrt{2}} + \frac{{\psi_2^{j_o}}_{j_x,j_y-\frac{1}{2}}-i {\psi_2^{j_o}}_{j_x-\frac{1}{2} ,j_y}}{2\sqrt{2}}\bigg|^2 ~,
\ee
with  $j_x,j_y \in \mathbb{Z}^2 \cup (\mathbb{Z}+1/2)^2$ .
For a pure state the averaged norm  
can be recast into the expression \eqref{TRU} for $\rho_{11}$ on subgrid ${\cal G}_1^{j_o-\frac{1}{2}}$.   Similarly, using the averaged norm for components $\psi_2$ one arrives at \eqref{TRV}  for subgrid ${\cal G}_2^{j_o}$.
 Using the pure-state decomposition  Eq. \eqref{mixeds} , adding the $\gamma_j$ weighted contributions for individual pure states  verifies inequality  \eqref{BKL3} for a general density matrix.
\end{proof}

\section{Direct scheme}\label{directs}

An alternative, direct  scheme  using the same placement of density matrix elements on the lattice formulated is formulated here.  It allows centered difference quotients over one lattice spacing for an approximation of all first-order partial derivatives thus eliminating the source for fermion doubling.    Rather than  following the separate bra-ket application of the pure-state scheme, however, it is based on a simultaneous time propagation of the bra and ket indices thereby reducing the number of operations required per time step.

The compact global form oft his numerical scheme may be formulated  as an initial value problem: at initial time the four density  matrix elements are stored on the two adjacent time sheets, ${\cal G}_1^{(j_o-1/2)}$ and ${\cal G}_2^{(j_o)}$,  associated, respectively, with time  $j_o-1/2$ for grid ${\cal G}_1$ and time $j_o$ for grid ${\cal G}_2$.  Replacing partial derivatives in Eq. \ref{DVN} by difference quotients, as detailed in Sect \ref{placement}, on obtains for this direct scheme:
\begin{definition}\label{defdi}
\begin{eqnarray}
 D_o \rho_{11}({\cal {\bar G}}_1^{(j_o)};{\cal {\bar G}}_1^{(j_o)'})  & = & \left[M_+({\cal {\bar G}}_1^{(j_o)})-M_+({\cal {\bar G}}_1^{(j_o)'}) \right]{\cal T} \rho_{11}({\cal {\bar G}}_1^{(j_o)};{\cal {\bar G}}_1^{(j_o)'}) \nonumber \\
 & -  &D_+\rho_{21}({\cal {\bar G}}_1^{(j_o)};{\cal G}_1^{(j_o-1/2)'}) 
-  D_-'\rho_{12}({\cal  G}_1^{(j_o-1/2)};{\cal  {\bar G}}_1^{(j_o)'}) ~,
\label{11}
\end{eqnarray}
\begin{eqnarray}
D_o  \rho_{12}({\cal {\bar G}}_1^{(j_o)};{\cal {\bar G}}_2^{(j_o+1/2)'}) & =  & \left[M_+({\cal {\bar G}}_1^{(j_o)})-M_- ({\cal {\bar G}}_2^{(j_o+1/2)'}\right]{\cal T}\rho_{12}({\cal {\bar G}}_1^{(j_o)};{\cal {\bar G}}_2^{(j_o+1/2)'}) \nonumber \\
& -  & D_+\rho_{22}({\cal {\bar G}}_1^{(j_o)};{\cal G}_2^{(j_o)'}) 
- D_+' \rho_{11}({\cal  G}_1^{(j_o+1/2)};{\cal {\bar G}}_2^{(j_o+1/2)'}) ~,
\label{12}
\end{eqnarray}
\begin{eqnarray}
D_o  \rho_{21}({\cal {\bar G}}_2^{(j_o+1/2)};{\cal {\bar G}}_1^{(j_o)'}) & =  &  \left[M_-({\cal {\bar G}}_2^{(j_o+1/2)})-M_+({\cal {\bar G}}_1^{(j_o)'})\right]{\cal T}\rho_{21}({\cal {\bar G}}_2^{(j_o+1/2)};{\cal {\bar G}}_1^{(j_o)'})   \nonumber  \\ 
&- &  D_- \rho_{11}({\cal  {\bar G}}^{(j_o+1/2)}_2;{\cal G}_1^{(j_o+1/2)'}) 
- D_-' \rho_{22}({\cal  G}_2^{(j_o)};{\cal {\bar G}}_1^{(j_o)'})  ~,
\label{21}
\end{eqnarray}
and
\begin{eqnarray}
D_o  \rho_{22}({\cal {\bar G}}^{(j_o+1/2)}_2;{\cal {\bar G}}^{(j_o+1/2)'}_2) & =  &  \left[M_-({\cal {\bar G}}_2^{(j_o+1/2)})-M_-({\cal {\bar G}}_2^{(j_o+1/2)'})\right]{\cal T}\rho_{22}({\cal {\bar G}}_2^{(j_o+1/2)};{\cal {\bar G}}_2^{(j_o+1/2)'})  \nonumber \\
&   - &  D_- \rho_{12}({\cal {\bar G}}_2^{(j_o+1/2)};{\cal  G}_2^{(j_o+1)'})   - D_+'\rho_{21}({\cal  G}_2^{(j_o+1)};{\cal {\bar G}}_2^{(j_o+1/2)'}) ~.
\label{22}
\end{eqnarray}
\end{definition}
Eqs. (\ref{12}) and (\ref{21}) are equivalent.    The explicit form of the scheme for specific lattice sites is given in the Appendix.  

The key, as for the pure state case, is a combined staggering in space and time which allows the calculation of symmetric difference quotients for each of the first-order partial derivatives.  
Note  that the first order-difference operators, as well as the time-averaging operator, which formally are defined on grids $ {\cal  {\bar G}}_i$ place the density matrix 
elements on the grids $ {\cal  G}_i$ (rather than $ {\cal  {\bar G}}_i$),so that the latter must be stored as matrices  $\rho_{ij}({\cal G}_i;{\cal G}_j), ~ i,j=1,2$. 
In particular, spatial derivatives are performed on the time-sheet which is sandwiched between the  initial and final time-sheet of the associated row or column which is propagated: spatial derivatives map ${\cal {\bar G}}_2^{(j_o+1/2)}$ onto  ${\cal G}_1^{(j_o+1/2)}$; likewise   ${\cal {\bar G}}_1^{(j_o)}$ is mapped onto ${\cal G}_2^{(j_o)}$.  The numerical procedure follows a local time-sheet by time sheet  propagation which allows full parallelization in its execution as an initial-value problem.  For a single time step this explicit scheme scales like $N^2$, with $N$ denoting the number of grid points on a single time sheet.  

\subsection{Hermiticity under the direct scheme}

\begin{lem}\label{L3}
The direct scheme Eqs. \eqref{11} to \eqref{22} conserves Hermiticity of the density matrix.
\end{lem}

\begin{proof} 
It is sufficient to show conservation for one time step under the assumption of Hermiticity of the initial density matrix placed on time sheets $j_o-1/2$ and $j_o$.  
This is done in straight-forward way by showing that, within the scheme, complex conjugation is equivalent to a transposition of the density matrix on time-sheets $j_o+1/2$ and $j_o+1$.  

Using $M_\pm^*=-M_\pm$ and relations, such as 
\begin{eqnarray}
 \left\{D'_+\rho_{21}({\cal {\bar G}}_1^{(j_o)'};{\cal G}_1^{(j_o-1/2)})\right\}^* &=&\left\{\frac{{\rho_{21}}_{j_{x'},j_{y'} + \frac{1}{2};j_x,j_y}^{j_o' ; j_{o}- \frac{1}{2}} -{\rho_{21}}_{j_{x'},j_{y'} - \frac{1}{2};j_x,j_y}^{j_o' ; j_{o}- \frac{1}{2}}} {\Delta_y}+
  i\frac{{\rho_{21}}_{j_x'+ \frac{1}{2},j_{y'} ;j_{x},j_{y}}^{j_o' ; j_{o}- \frac{1}{2}} -{\rho_{21}}_{j_x'- \frac{1}{2},j_{y'} ;j_{x},j_{y}}^{j_o' ; j_{o}- \frac{1}{2} }}{\Delta_x}\right\}^*\nonumber \\
&=&\frac{{\rho_{12}}_{j_x,j_y;j_{x'},j_{y'} + \frac{1}{2}}^{  j_{o}- \frac{1}{2};j_o'} -{\rho_{12}}_{j_x,j_y;j_{x'},j_{y'} - \frac{1}{2}}^{  j_{o}- \frac{1}{2};j_o'}} {\Delta_y}-
  i\frac{{\rho_{12}}_{j_{x},j_{y};j_x'+ \frac{1}{2},j_{y'} }^{  j_{o}- \frac{1}{2};j_o'} -{\rho_{12}}_{j_{x},j_{y};j_x'- \frac{1}{2},j_{y'} }^{ j_{o}- \frac{1}{2};j_o' }}{\Delta_x}\nonumber \\
&=& D_-'\rho_{12}({\cal G}_1^{(j_o-1/2)};{\cal {\bar G}}_1^{(j_o)'})~,
\end{eqnarray}
as well as the short-hand grid notation introduced above one may write 
\begin{eqnarray}
 {\rho_{11}({\cal { G}}_1^{(j_o+\frac{1}{2})'};{\cal {G}}_1^{(j_o+\frac{1}{2})})}^* &=& \frac{1}{\frac{1}{\Delta_o} - \frac{({M_+({\cal {\bar G}}_1^{(j_o)'})}^*-{ M_+({\cal {\bar G}}_1^{(j_o)}) }^*)}{2}}
 \Bigl[   \frac{1}{\Delta_o} + \frac{({ M_+({\cal {\bar G}}_1^{(j_o)'}) }-{ M_+({\cal {\bar G}}_1^{(j_o)}) })}{2}{\rho_{11}({\cal { G}}_1^{(j_o-\frac{1}{2})'};{\cal {G}}_1^{(j_o-\frac{1}{2})})} \nonumber \\
   &-& {D_+'\rho_{21}({\cal {\bar G}}_1^{(j_o)'};{\cal G}_1^{(j_o-\frac{1}{2})})} - {D_-\rho_{12}({\cal  G}_1^{(j_o-\frac{1}{2})'};{\cal   {\bar G}}_1^{(j_o)})}\Bigr]^*\nonumber  \\
 & =&  \frac{1}{\frac{1}{\Delta_o} - \frac{({-M_+({\cal {\bar G}}_1^{(j_o)'})}+{ M_+({\cal {\bar G}}_1^{(j_o)}) })}{2}} \Bigl[   \frac{1}{\Delta_o} + \frac{( -M_+({\cal {\bar G}}_1^{(j_o)'}) +{ M_+({\cal {\bar G}}_1^{(j_o)}) })}{2}{\rho_{11}({\cal { G}}_1^{(j_o-\frac{1}{2})};{\cal {G}}_1^{(j_o-\frac{1}{2})'})} \nonumber \\
    &-&  {D_-'\rho_{12}({\cal G}_1^{(j_o-\frac{1}{2})};{\cal  {\bar G}}_1^{(j_o)'})} - {D_+\rho_{21}({\cal  {\bar G}}_1^{(j_o)};{\cal   G}_1^{(j_o-\frac{1}{2})'})}\Bigr] \nonumber \\
 & =&  \rho_{11}({\cal { G}}_1^{(j_o+\frac{1}{2})};{\cal {G}}_1^{(j_o+\frac{1}{2})'}) ~,
\end{eqnarray}
\begin{eqnarray}
  {\rho_{12}({\cal { G}}_1^{(j_o+\frac{1}{2})'};{\cal {G}}_2^{(j_o+1)})}^*  &=&  \frac{1}{\frac{1}{\Delta_o} - \frac{({M_+({\cal {\bar G}}_1^{(j_o)'})}^*-{ M_-({\cal {\bar G}}_2^{(j_o+ \frac{1}{2})}) }^*)}{2}}
 \Bigl[   \frac{1}{\Delta_o} + \frac{({ M_+({\cal {\bar G}}_1^{(j_o)'}) }-{ M_-({\cal {\bar G}}_2^{(j_o + \frac{1}{2})}) })}{2}{\rho_{12}({\cal { G}}_1^{(j_o-\frac{1}{2})'};{\cal {G}}_2^{(j_o)})} \nonumber \\
   &-&  {D_+'\rho_{22}({\cal {\bar G}}_1^{(j_o)'};{\cal G}_2^{(j_o)})} - {D_+\rho_{11}({\cal  G}_1^{(j_o+\frac{1}{2})'};{\cal   {\bar G}}_2^{(j_o+\frac{1}{2})})}\Bigr]^*\nonumber  \\
 &=& \frac{1}{\frac{1}{\Delta_o} - \frac{({-M_+({\cal {\bar G}}_1^{(j_o)'})}+{ M_-({\cal {\bar G}}_2^{(j_o+\frac{1}{2})}) })}{2}} \Bigl[  \frac{1}{\Delta_o} + \frac{( -M_+({\cal {\bar G}}_1^{(j_o)'}) +{ M_-({\cal {\bar G}}_2^{(j_o+\frac{1}{2})}) })}{2}{\rho_{12}({\cal { G}}_1^{(j_o-\frac{1}{2})};{\cal {G}}_2^{(j_o)'})} \nonumber \\
  &-&  {D_-'\rho_{22}({\cal G}_1^{(j_o)};{\cal  {\bar G}}_1^{(j_o)'})} - {D_-\rho_{11}({\cal  {\bar G}}_2^{(j_o+\frac{1}{2})};{\cal   G}_1^{(j_o+\frac{1}{2})'})}\Bigr] \nonumber \\
 &=&  \rho_{21}({\cal { G}}_2^{(j_o+1)};{\cal {G}}_1^{(j_o+\frac{1}{2})'}) ~,
\end{eqnarray}
\begin{eqnarray}
 {\rho_{21}({\cal { G}}_2^{(j_o+1)'};{\cal {G}}_1^{(j_o+\frac{1}{2})})}^* &=& \frac{1}{\frac{1}{\Delta_o} - \frac{{M_-({\cal {\bar G}}_2^{(j_o+\frac{1}{2})})'}^*-{M_+({\cal {\bar G}}_1^{(j_o)})}^*}{2}}
   \Bigl[  \frac{1}{\Delta_o} + \frac{{M_-({\cal {\bar G}}_2^{(j_o+\frac{1}{2})})'}-{M_+({\cal {\bar G}}_1^{(j_o)})}}{2}
  {\rho_{21}({\cal { G}}_2^{(j_o)'};{\cal  { G}}_1^{(j_o-\frac{1}{2})}) }  \nonumber  \\ 
  &-& { D_- '\rho_{11}({\cal  { \bar G}}^{(j_o+\frac{1}{2})'}_2;{\cal G}_1^{(j_o+\frac{1}{2})})}
  - {D_- \rho_{22}({\cal  G}_2^{(j_o)'};{\cal  {\bar G}}_1^{(j_o)})}\Bigr]^* \nonumber \\
&=&     \frac{1}{\frac{1}{\Delta_o} - \frac{-{M_-({\cal {\bar G}}_2^{(j_o+\frac{1}{2})'})}+{M_+({\cal {\bar G}}_1^{(j_o)})}}{2}}
\Bigl[  \frac{1}{\Delta_o} + \frac{-{M_-({\cal {\bar G}}_2^{(j_o+\frac{1}{2})})'}+{M_+({\cal {\bar G}}_1^{(j_o)})}}{2}
  {\rho_{12}({\cal { G}}_1^{(j_o-\frac{1}{2})};{\cal  { G}}_2^{(j_o)'}) }  \nonumber  \\ 
 & -& { D_+' \rho_{11}({\cal G}_1^{(j_o+\frac{1}{2})};{\cal   { \bar G}}^{(j_o+\frac{1}{2})'}_2)} 
  - {D_+ \rho_{22}({\cal {\bar G}}_1^{(j_o)};{\cal   G}_2^{(j_o)'})} \Bigr] \nonumber \\
&=& \rho_{12}({\cal { G}}_1^{(j_o+\frac{1}{2})};{\cal  { G}}_2^{(j_o+1)'}) ~,
\end{eqnarray}
and 
\begin{eqnarray}
 {\rho_{22}({\cal { G}}^{(j_o+1)'}_2;{\cal  { G}}^{(j_o+1)}_2)}^*  &= &  \frac{1}{\frac{1}{\Delta_o} - \frac{{M_-({\cal {\bar G}}_2^{(j_o+\frac{1}{2})'})}^*-{M_-({\cal {\bar G}}_2^{(j_o+\frac{1}{2})})}^*}{2}}
  \Bigl[  \frac{1}{\Delta_o} + \frac{{M_-({\cal {\bar G}}_2^{(j_o+\frac{1}{2})'})}-{M_-({\cal {\bar G}}_2^{(j_o+\frac{1}{2})})}}{2}
 {\rho_{22}({\cal { G}}_2^{(j_o)'};{\cal  { G}}_2^{(j_o)}) } \nonumber \\
  &-&   {D_-' \rho_{12}({\cal {\bar G}}_2^{(j_o+\frac{1}{2})'};{\cal   G}_2^{(j_o+1)})}   - 
  {D_+\rho_{21}({\cal  G}_2^{(j_o+1)'};{\cal  {\bar G}}_2^{(j_o+\frac{1}{2})})}\Bigr]^* \nonumber \\
&=&    \frac{1}{\frac{1}{\Delta_o} - \frac{-{M_-({\cal {\bar G}}_2^{(j_o+\frac{1}{2})'})}+{M_-({\cal {\bar G}}_2^{(j_o+\frac{1}{2})})}}{2}}
\Bigl[  \frac{1}{\Delta_o} + \frac{-{M_-({\cal {\bar G}}_2^{(j_o+\frac{1}{2})'})}+{M_-({\cal {\bar G}}_2^{(j_o+\frac{1}{2})})}}{2}
  {\rho_{22}({\cal { G}}_2^{(j_o)};{\cal  { G}}_2^{(j_o)'}) } \nonumber \\
  &-&  {D_+' \rho_{21}({\cal  G}_2^{(j_o+1)};{\cal  {\bar G}}_2^{(j_o+\frac{1}{2})'})}   - 
  {D_-\rho_{12}({\cal {\bar G}}_2^{(j_o+\frac{1}{2})};{\cal   G}_2^{(j_o+1)'})}\Bigr]\nonumber \\
&=& \rho_{22}({\cal { G}}^{(j_o+1)}_2;{\cal  { G}}^{(j_o+1)'}_2) ~.
\end{eqnarray}
\end{proof}

\subsection{Stability of the direct scheme}

We first consider mass and potential to be position independent.  Time-dependence is permitted.  

\begin{prop}\label{PVN}
For constant-in-position potential $V$ and mass $m$, the direct scheme Def. \ref{defdi} is stable for
\be
r_x^2+r_y^2\leq \frac{1}{4}
\ee
\end{prop}

\begin{proof}
The proof is given by von Neumann analysis setting, for time-step $\Delta_t$, 
\be
\rho^+_{ij}=e^{i\omega \Delta_t}  e^{i {\bf k\cdot r}} e^{i {\bf k'\cdot r'}} \rho^-_{ij}, ~ i,j=1,2
\label{rhoc}
\ee
and inserting this expression into the scheme Def. \ref{defdi}.  The resulting  four linear  homogeneous equations in $\rho^-_{ij}$ have a characteristic determinant which has the solutions 
\be
\sin^2(\frac{\omega \Delta_t}{2})|_\pm= \frac{p+p'+ \frac{\beta^2}{2}}{1+\beta^2} \pm \frac{\sqrt{ (p+p'+ \frac{\beta^2}{2})^2  -(1+\beta^2)(p-p')^2}}{1+\beta^2}
\label{sin2}
\ee
with
$$
p^{(')}=r_y^2\sin^2{(\frac{k_y^{(')} \Delta_y}{2})}+r_x^2\sin^2{(\frac{k_x^{(')} \Delta_x}{2})}
$$
and $\beta=\frac{m\Delta_o}{\hbar v}$.  Noting that $\beta,p,p'\geq 0$, the maximum of the rhs of Eq. \eqref{sin2} is obtained for $p=p'=p_{max}=r_x^2+r_y^2$ and the positive sign for the square root.  This leads to the 
condition $4p_{max}<1$ rendering $\sin^2(\frac{\omega \Delta_t}{2})|_\pm\leq 1$.  Finally, in the domain  $0\leq p^{(')}\leq \frac{1}{4}$, the radicant in Eq. \eqref{sin2} remains positive definite for $\beta\geq0$.
\end{proof}

\begin{figure*}[h]
\centering
(a)  \includegraphics[width=6cm]{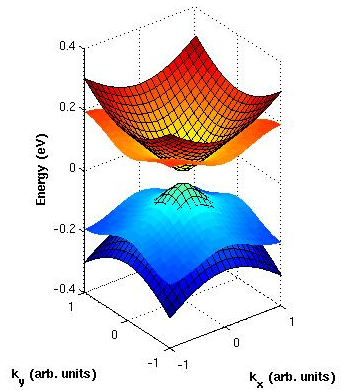} 
(b)  \includegraphics[width=6cm]{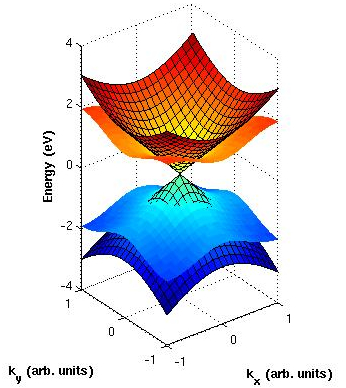}
\caption{ Exact energy-momentum dispersion (meshed-surface) versus the direct-scheme dispersion: mass $m=0.02$eV, $v=6.2\times 10^6$m/s, $\Delta_t=\frac{\Delta}{2\sqrt{2}}$, $\Delta=\Delta_x=\Delta_y$, (a)  $\Delta=6$nm ($150\times 150$ grid points), $\Delta_t=3.4\times 10^{-15}$s (b) $\Delta=0.6$nm ($1500\times 1500$ grid points), $\Delta_t=3.4\times 10^{-16}$s.}
\label{Fig0}
\end{figure*}

Comments: 

(i) Von Neumann analysis reveals the single-cone energy-momentum dispersion relation of the scheme.  For Hermitian $\rho$, one has ${\bf k'}= -{\bf k}$ and $p=p'$.  Hence for $r_x^2+r_y^2\leq \frac{1}{4}$ the
four real solutions for $\omega$  are $\omega_1=\omega_2=0$ and 
\be
\omega_3=-\omega_4=\frac{2}{\Delta_t} \arcsin{\left\{     \sqrt{\frac{4p + \beta^2}{1+\beta^2}}   \right\}}~.
\ee
The origin of these four solutions can be interpreted as follows.  For given $k$ vector there is a positive- (particle) and a negative-energy (antiparticle) solution.  Both of them are needed for the construction of a complete basis set.  Hence there are four types of matrix elements in Eq. \eqref{rhoc}: particle-particle ($\omega_1=0$) and antiparticle-antiparticle  ($\omega_2=0$), particle-antiparticle ($\omega_4$) and antiparticle-particle ($\omega_3$).  Thus the energy-momentum dispersion of the direct scheme is given by

\be
E_{k}=\pm \frac{\hbar}{\Delta_t} \arcsin{\left\{    \frac{\Delta_t }{\hbar}  \sqrt{\frac{\left(v\frac{2\hbar}{\Delta_x}\sin{(\frac{k_x \Delta_x}{2})}\right)^2+\left(v\frac{2\hbar}{\Delta_y}\sin{(\frac{k_y\Delta_y}{2})}\right)^2 + m^2}{1+\left(\frac{m\Delta_t}{\hbar}\right)^2}}   \right\}} ~.\label{EKDIR}
\ee
A comparison between the exact dispersion to the one for the direct scheme is given in Fig. \ref{Fig0} for a specific numerical example with $m=0.02$eV.  Note that for the case of 1500 grid points, the finite mass gap cannot be resolved for the scale used to allow the display of the entire range of k-vector space.   

(ii) Under the conditions of Prop. \ref{PVN}, preservation of positivity is ensured.  Starting from a Hermitian positive density matrix $\rho^-_{ij}$, the latter may be expanded in the stationary orthonormal plane-wave eigenkets of the scheme according to Eq. \eqref{mixeds} which  evolve into $\rho^+_{ij}$ according to Eq. \eqref{rhoc} and ${\bf k'}= -{\bf k}$ .

We now turn to the case of time- and position- dependent  mass and potential.
The direct scheme cannot be expected to conserve positivity in general since it is based on a finite difference approximation of the  linearized differential equation 
$$
\rho^+ =(1-\frac{i\Delta_t H}{\hbar}) \rho^- (1+\frac{i\Delta_t H}{\hbar}  )=  \rho^-  -\frac{i\Delta_t}{\hbar}\left[ H,\rho^-\right] +\left(\frac{\Delta_t}{\hbar}\right)^2 H\rho^-H \approx  \rho^-  -\frac{i\Delta_t}{\hbar}\left[ H,\rho^-\right] 
$$
Clearly, the original scheme is of Kraus form (with non-unitary Kraus operator), however, the linearized form is not.   The maximum violation of positivity of the latter  is of order $\left(\Delta_t\right)^2$ ($\left(\frac{\Delta_o}{\Delta}\right)^2$ for $m_\pm=0$) per time step.  For states $|\Phi\rangle$ for which $\langle \Phi |\rho^-|\Phi\rangle=0$ it is given by    $-\left(\frac{\Delta_t}{\hbar}\right)^2 \langle \Phi| H\rho^-H|\Phi\rangle$ ($\leq 0$ for positive $\rho^-$).    So starting from a positive definite initial density operator, it is conceivable that with time one looses positivity.  Nevertheless, the direct scheme supports a conserved functional which replaces conservation of the trace for the continuum Dirac equation and which is closely related to the  conserved functional of Lemma \ref{L1} for the bra-ket scheme.

\begin{lem}\label{LP3}
Under periodic, respectively, zero boundary conditions,
\begin{eqnarray}
T_{{\bf r}}^0 (\rho)= &&\Tr\{ \rho_{11}({\cal G}^{(j_o-1/2)}_1;{\cal  G}^{(j_o-1/2)}_1) + \rho_{22}({\cal G}_2^{(j_o)};{\cal  G}^{(j_o)}_2)-~~~~~~~~~~~~~~~~~~~~~~~~~~~~~~\nonumber \\
&& \Delta_oD_+\rho_{21}({\cal {\bar G}}^{(j_o)}_1;{\cal  G}^{(j_o-1/2)}_1)+
\Delta_o D_- \rho_{12}({\cal {\bar G}}^{(j_o-1/2)}_2;{\cal  G}^{(j_o)}_2)\}, {\bf r}=(r_x=\Delta_o/\Delta_x, r_y=\Delta_o/\Delta_y),  \label{tr3}
\end{eqnarray}
is a conserved quantity under the direct scheme.
\end{lem}
Here trace $\Tr$, consistent with Definition \ref{def1},  implies the sum over $(j_x,j_y) \in \mathbb{Z}^2 \cup (\mathbb{Z}+\frac{1}{2})^2$ performed on the respective time-sheets.\footnote{Note that $\rho$ also is a matrix in spin space indicated by the subscripts $(i,j)$ in $\rho_{ij}$.}
This conservation law corresponds to the  conservation of the trace of $\rho$ in the continuum limit, in which all spatial derivatives are carried out on a ``single" time-sheet.
\begin{proof}
Making use of the Cor. \ref{cor1}  (integration by parts on the lattice) one may write
\begin{eqnarray}
& &\Tr\left\{\rho_{11}({\cal G}^{(j_o+1/2)}_1;{\cal  G}^{(j_o+1/2)}_1) + \rho_{22}({\cal G}_2^{(j_o+1)};{\cal  G}^{(j_o+1)}_2)\right\}\nonumber \\
& = &\sum_{j_x,j_y }\bigg[ \rho_{11}({\cal G}^{(j_o+1/2)}_1;{\cal  G}^{(j_o+1/2)}_1)+\rho_{22}({\cal G}_2^{(j_o+1)};{\cal  G}^{(j_o+1)}_2)\bigg] \nonumber \\
& = & \sum_{j_x,j_y }\bigg[ \rho_{11}({\cal G}^{(j_o-1/2)}_1;{\cal  G}^{(j_o-1/2)}_1) - \Delta_o D_+\rho_{21}({\cal {\bar G}}^{(j_o)}_1;{\cal  G}^{(j_o-1/2)}_1)\nonumber \\
&-&  \Delta_o D_-'\rho_{12}({\cal {G}}^{(j_o-1/2)}_1;{\cal \bar {G}}^{(j_o)}_1)\bigg]\nonumber \\
& + & \sum_{j_x,j_y }\bigg[ \rho_{22}({\cal G}_2^{(j_o)};{\cal  G}^{(j_o)}_2) - \Delta_o D_-\rho_{12}({\cal {\bar G}}^{(j_o+1/2)}_2;{\cal  G}^{(j_o+1)}_2)\nonumber \\
&-& \Delta_o D_+'\rho_{21}({\cal { G}}^{(j_o+1)}_2;{\cal {\bar G}}^{(j_o+1/2)}_2)\bigg] \nonumber \\
& = & \sum_{j_x,j_y }\bigg[ \rho_{11}({\cal G}^{(j_o-1/2)}_1;{\cal  G}^{(j_o-1/2)}_1) - \Delta_o D_+\rho_{21}({\cal {\bar G}}^{(j_o)}_1;{\cal  G}^{(j_o-1/2)}_1)\nonumber \\
&+&  \Delta_o D_-\rho_{12}({\cal {\bar G}}^{(j_o-1/2)}_2;{ \cal G}^{(j_o)}_2)\bigg] \nonumber \\
& + & \sum_{j_x,j_y }\bigg[ \rho_{22}({\cal G}_2^{(j_o)};{\cal  G}^{(j_o)}_2) - \Delta_o D_-\rho_{12}({\cal {\bar G}}^{(j_o+1/2)}_2;{\cal  G}^{(j_o+1)}_2)\nonumber \\ 
&+&  \Delta_o D_+\rho_{21}({\cal {\bar G}}^{(j_o+1)}_1;{\cal  G}^{(j_o+1/2)}_1)\bigg] \nonumber \\ 
& = & \sum_{j_x,j_y }\bigg[ \rho_{11}({\cal G}^{(j_o-1/2)}_1;{\cal  G}^{(j_o-1/2)}_1) + {\rho_{22}({\cal G}^{(j_o)}_2;{\cal  G}^{(j_o)}_2})\nonumber \\
&-& \Delta_o D_+\rho_{21}({\cal {\bar G}}^{(j_o)}_1;{\cal  G}^{(j_o-1/2)}_1) + \Delta_o D_-\rho_{12}({\cal {\bar G}}^{(j_o-1/2)}_2;{\cal  G}^{(j_o)}_2) \nonumber \\
& + & \Delta_o D_+\rho_{21}({\cal {\bar G}}^{(j_o+1)}_1;{\cal  G}^{(j_o+1/2)}_1) - \Delta_o D_-\rho_{12}({\cal {\bar G}}^{(j_o+1/2)}_2;{\cal  G}^{(j_o+1)}_2)\bigg]~. \\ \nonumber
\end{eqnarray}

\end{proof}

In the time regime of positivity sufficient convergence conditions can be established from the conservation law Eq. (\ref{tr3}).
\begin{prop}\label{P5}
For $ r_x + r_y < 1/2$, with $r_\nu = \Delta_o/\Delta_\nu$, $\nu \in \{x,y\}$, and the assumption of positivity for $\rho$ the scheme satisfies the estimate:
\begin{eqnarray}
Tr\left\{ \rho_{11} + \rho_{22}\right\} \leq \frac{T_r^o}{1 - 2(r_x + r_y)}~.\label{stab2}
\end{eqnarray}
\end{prop}
\begin{proof}
Since Hermiticity is already shown, we can write (here complex conjugation is denoted by the bar symbol $\bar{}$ for compactness of notation)
$$\rho = \left(\begin{array}{cc}
\sum_i p_i u_i{\bar{u_i}} & \sum_i p_i u_i{\bar{v_i}}\\ \sum_i p_i v_i{\bar{u_i}} & \sum_i p_i v_i{\bar{v_i}}
\end{array}\right)$$
for each time step $j_t \in (\mathbb Z \cup \mathbb Z +\frac 1 2)$. Under the assumption of positivity of $\rho$  one has $0 < p_i \leq 1$, real. We use the short-hand notation $\Tr\{|M|\}=\sum_{(j_x,j_y)} |M|$, Cor. \ref{cor1}, as well as $2|\Re(u*v)| \leq \|u\|^2 + \|v\|^2$, to estimate  \\
\begin{eqnarray}\label{cvdir}
 &&Tr\left\{|- \Delta_o D_+\rho_{21} + \Delta_o D_-\rho_{12}|\right\} \nonumber \\
 & =&  Tr\left\{|-\Delta_oD_+\rho_{21} + \Delta_o \left\{D_+'\rho_{21}\right\}^*|\right\} ~\nonumber\\ 
 & =& Tr\left\{|-\Delta_oD_+\rho_{21} - \Delta_o \left\{D_+\rho_{21}\right\}^*|\right\}~\nonumber\\
& = &Tr\left\{|2\Re\left(-\Delta_oD_+\rho_{21}\right)|\right\} ~ \nonumber\\
 &=&2Tr\left\{ |r_y\Re\left[ \sum_i p_i v_i^{y_+} \bar{u_i} - p_i v_i^{y_-} \bar{u_i}\right] +  r_x \Re\left[ i \sum_i p_i v_i^{x_+} \bar{u_i} - p_i v_i^{x_-} \bar{u_i}\right]|\right\} ~\nonumber\\ 
 &\leq& 2 Tr\left\{ r_y|\Re\left[\sum_i p_i v_i^{y_+} \bar{u_i}\right]| + r_y|\Re\left[\sum_i p_i v_i^{y_-} \bar{u_i}\right]|\right\} ~\nonumber\\
 &+& 2Tr\left\{ r_x  |\Im\left[ \sum_i p_i v_i^{x_+} \bar{u_i}\right]| +  r_x|\Im\left[\sum_i p_i v_i^{x_-} \bar{u_i}\right] |\right\} ~\nonumber\\ 
&\leq& 2 (r_y + r_x) Tr\left\{\sum_i p_i u_i \bar{u_i} +\sum_i p_i v_i \bar{v_i}\right\} ~\nonumber\\
&\leq& 2 (r_y + r_x) Tr\left\{\rho_{11} + \rho_{22}\right\},~
\end{eqnarray}
followed by
\begin{eqnarray}
 T_r^o & = & Tr\left\{ \rho_{11} + \rho_{22} - \Delta_oD_+\rho_{21} + \Delta_o D_-\rho_{12} \right\} ~ \nonumber \\
 & \geq & Tr\left\{ \rho_{11} + \rho_{22} -|\Delta_oD_+\rho_{21} - \Delta_o D_-\rho_{12}|\right\} ~ \nonumber \\
 & \geq & Tr\left\{ \rho_{11} + \rho_{22}\right\} - 2(r_x + r_y) Tr\left\{\rho_{11} + \rho_{22}\right\} ~.
  \end{eqnarray}
  In Eq. \eqref{cvdir} we denote a single-site shift of a spinor component by a superscript.  For example, $v_i^{y_\pm}$ denotes a shift of $v_i$ by $\pm \Delta_y$.  
\end{proof}
  
In fact, a comparison of the conserved functional for the bra-ket and the direct scheme immediately allows a proof of Prop. \ref{P5} via Prop. \ref{P1}.  One observes that the difference between $\Tr\{ \rho\} $ (Def. \ref{def1})   and $T_{\bf{r}}^0$ is twice  the difference between $\Tr\{ \rho\} $  and  $E_\mathbf{r}^0 $
\bea
T_{\bf{r}}^0&=&\Tr\{ \rho_{11}({\cal G}^{(j_o-1/2)}_1;{\cal G}^{(j_o-1/2)}_1) + \rho_{22}({\cal G}_2^{(j_o)};{\cal G}^{(j_o)}_2)+
2 \Delta_o  \Re\big[D_- 
\rho_{12}({\cal {\bar G}}^{(j_o-1/2)}_2;{\cal G}^{(j_o)}_2)\big] \} \\ \nonumber 
&= &E_\mathbf{r}^0 + \Delta_o  \Re\big[D_- 
\rho_{12}({\cal {\bar G}}^{(j_o-1/2)}_2;{\cal G}^{(j_o)}_2)\big]~. \label{tr4} 
\eea
This accounts for the factor of $2$ in the denominator of the upper bound estimate in Eq. \eqref{stab2} as compared to Eq. \eqref{stab1}.   This observation immediately leads to an improved stability condition for the direct scheme via Prop. \ref{P2}.
\begin{prop}\label{P6}
Let $r=r_x=r_y=1/(2\sqrt{2})$ hold in \eqref{11}-\eqref{22}. Then, for all times  of positive definite $\rho$ this scheme is stable and satisfies 
\be
 {\tilde \Tr}(\rho_{11}^{j_o-1/2}) + {\tilde \Tr}(\rho_{22}^{j_o})  \le 2T^0~, \mbox{ with } T^0= T^0_{1/(2\sqrt{2}),1/(2\sqrt{2})}~.\label{BKL2}
\ee
\end{prop}
\begin{proof}
Observe that  $T^0_{1/(2\sqrt{2}),1/(2\sqrt{2})}= E^0_{1/\sqrt{2},1/\sqrt{2}}$ and use Prop. \ref{P2}.
\end{proof}
Note that  this estimate agrees with Prop. \ref{PVN} for position-independent mass and potential term.  Numerical simulations have confirmed this estimate, in particular, exceeding $r=r_x=r_y=1/(2\sqrt{2})$ has led to instability in specific simulations.

\section{Numerical simulations}\label{NumSim}

(2+1)D Dirac fermions have recently been realized as low-energy excitations in graphene and topological insulators.  In generic single-layer graphene there are four degenerate cones, while topological insulator states have an odd number of Dirac cones for a given surface.  Although our simulations are rather generic, we use parameters typical for TIs, such as Bi$_2$Te$_3$, using a Fermi velocity $v= 6.2\times 10^5$m/s.  
Simulation regions typically are ~ 1000x1000nm using at least ~ 100x100 grid points per time sheet.  At the rectangular boundary an absorbing layer is created by using a small imaginary contribution to the scalar potential which vanishes exponentially into the simulation region.  The schemes were implemented in a  MATLAB code using periodic boundary conditions allowing efficient global execution of the finite difference operations.  

\begin{figure}[h]
\centering
(a)  \includegraphics[width=6cm]{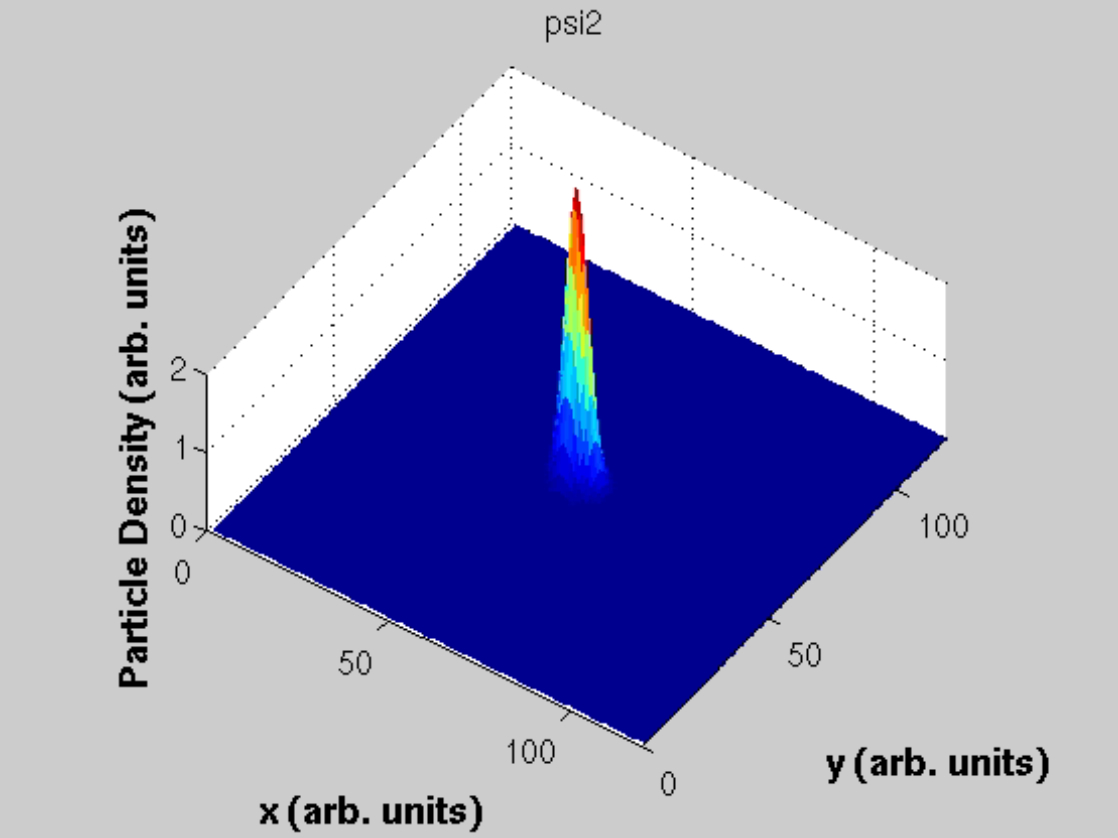}  (b) \includegraphics[width=6cm]{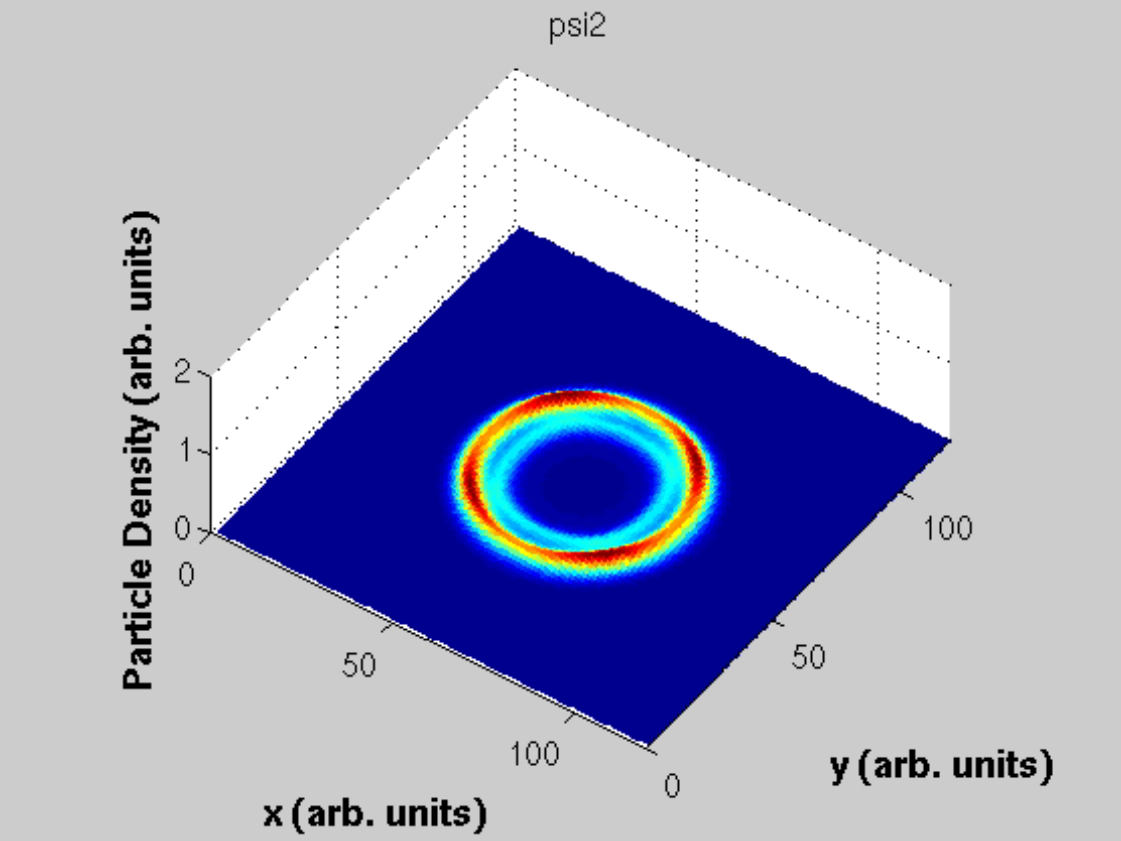}\\
~\\
(c) \includegraphics[width=6cm]{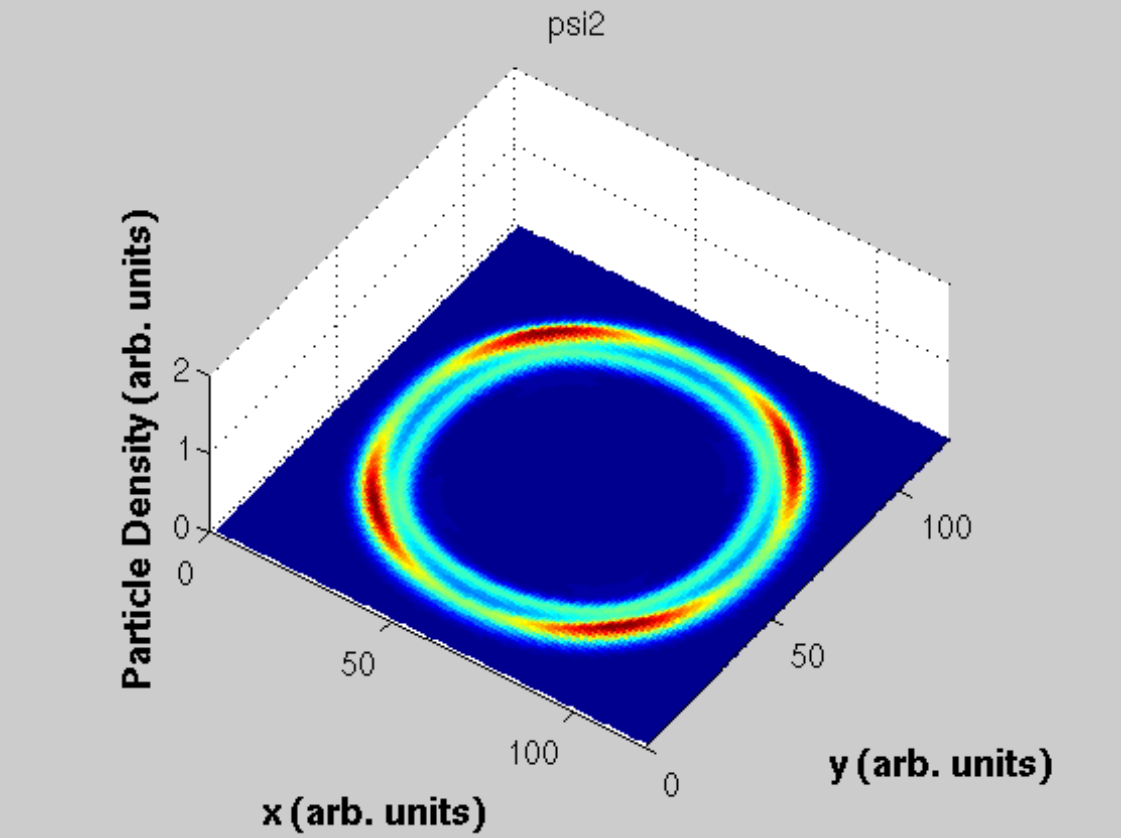}  (d) \includegraphics[width=6cm]{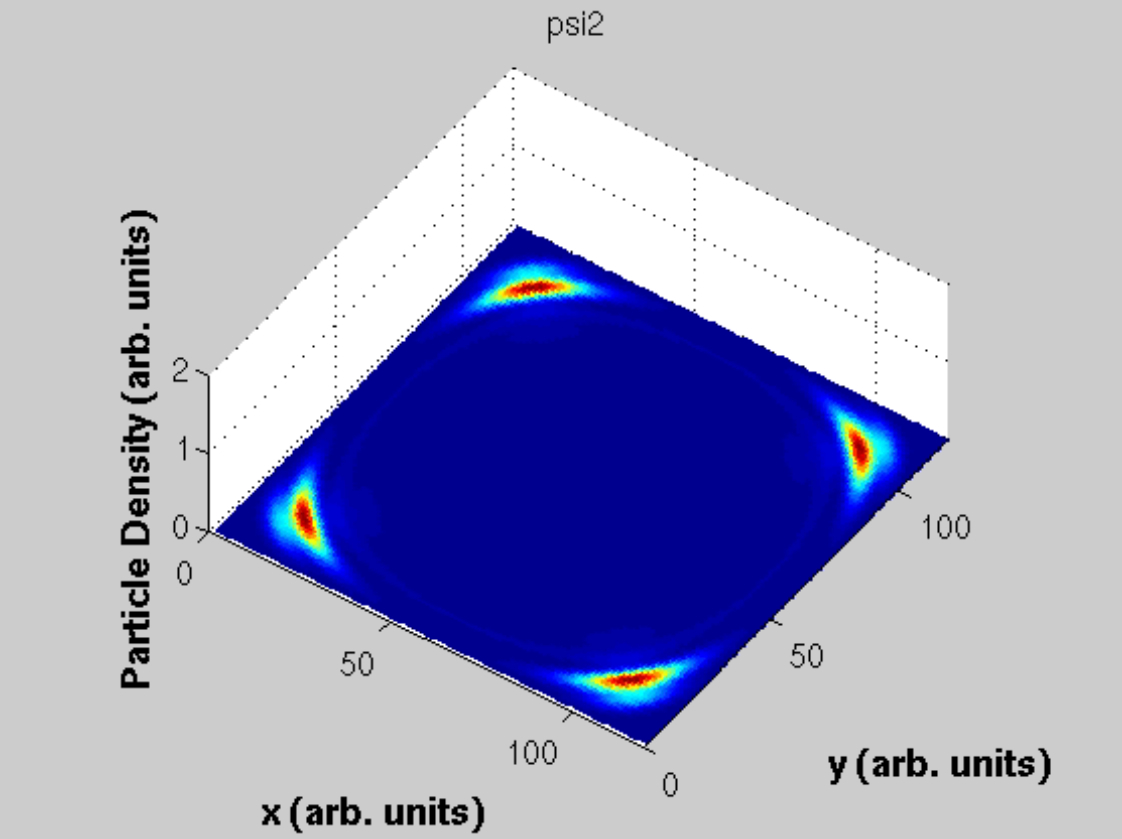}\\
~\\
\caption{Time-evolution of the single-particle density of an incoherent superposition of four Gaussian wave packets.}
\label{Fig1}
\end{figure}

\subsection{Gaussian superposition states}

We consider a mixed state of a zero-mass Dirac fermion  consisting of four  Gaussian wave packets centered, respectively,  at ${\bf k}_i= (\pm k_x,\pm k_y)$ with $E_x=\hbar |k_x| v=E_y=\hbar |k_y| v=0.01eV$.  All initial Gaussians share the same spectral width, average energy, and  center position, located at the center of the simulation region.  For this simulation the bra-ket scheme is used to study the propagation under free-particle conditions ($V=0$ in the simulation region).   Absorbing boundary conditions are implemented by means of an absorption layer via an imaginary potential contribution.\cite{MA-SL}  

Figs. 1(a)-1(d)  show snapshots of the single-particle density of states at four different times.  Compared to a Gaussian wave packet evolution under the Schr\"{o}dinger equation,  zero--mass Dirac fermions move dispersion-free.  More precisely, due to the linear dispersion, plane-wave contributions move with $\pm v\frac{{\bf k}}{k}$.  While the present lattice schemes maintain a monotonic energy dispersion in form of a single cone on the lattice, the linear dispersion regime is limited to small values of $k$ and $\Delta_t$, see Eq. \eqref{EKDIR} for the direct scheme.  
The lattice cone is represented by two energy bands of finite width which touch at $k=0$ and display particle-hole symmetry.   The analytic form for the dispersion underlying the bra-ket scheme was given in Ref. \cite{hammer3d}.    

 On a TI surface the number of flavors (Dirac cones) is odd and  Dirac fermions are helical. Hence charge transport is linked to spin transport, with  the particle spin being locked perpendicular to momentum.  For the present simulation the spin current density is shown in several snapshots (corresponding to Fig. \ref{Fig1}) of the expanding wave packets given in  
Fig. \ref{Fig2}.   For each snapshot the length of the "spin vector" is taken relative to the maximum value present at this particular time.  For the superposition of four Gaussians, the initial state is spin-unpolarized by symmetry.  However as the density current density evolves local spin density builds up and propagates along with the wave packets.  These results have been obtained within  the direct scheme as well.

\begin{figure}[h]
\centering
(a)  \includegraphics[width=6cm]{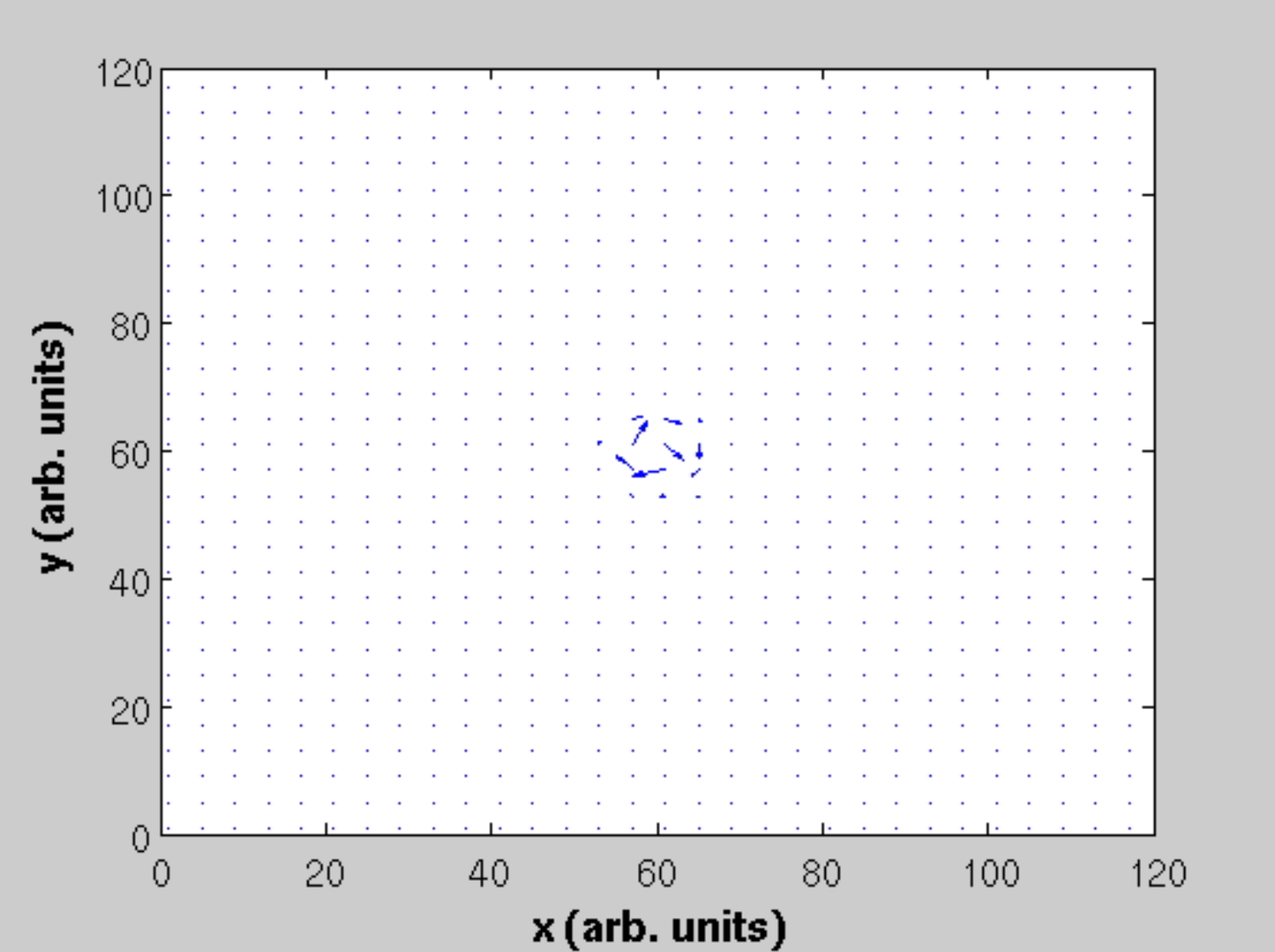}  (b) \includegraphics[width=6cm]{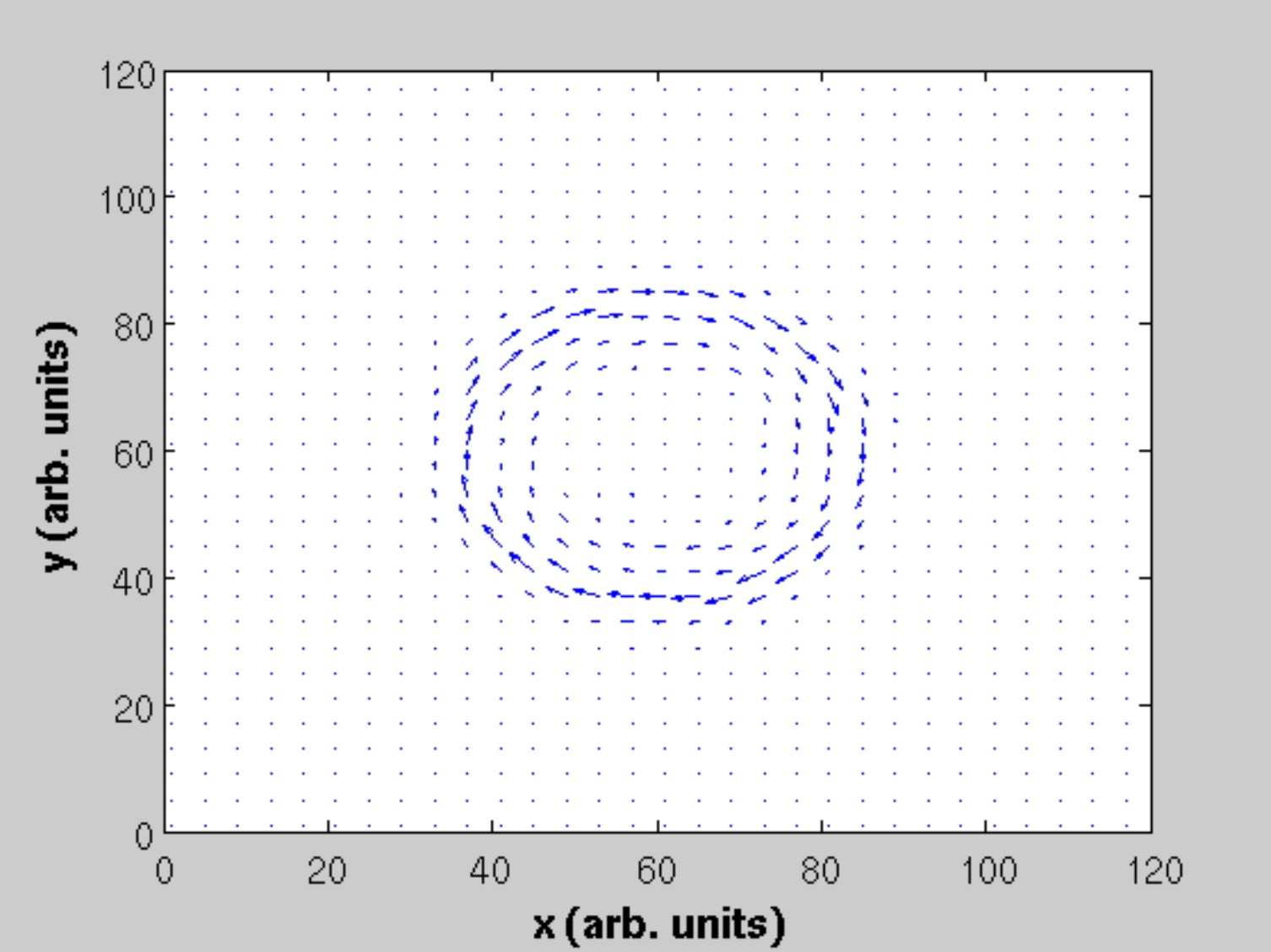}\\
~\\
(c) \includegraphics[width=6cm]{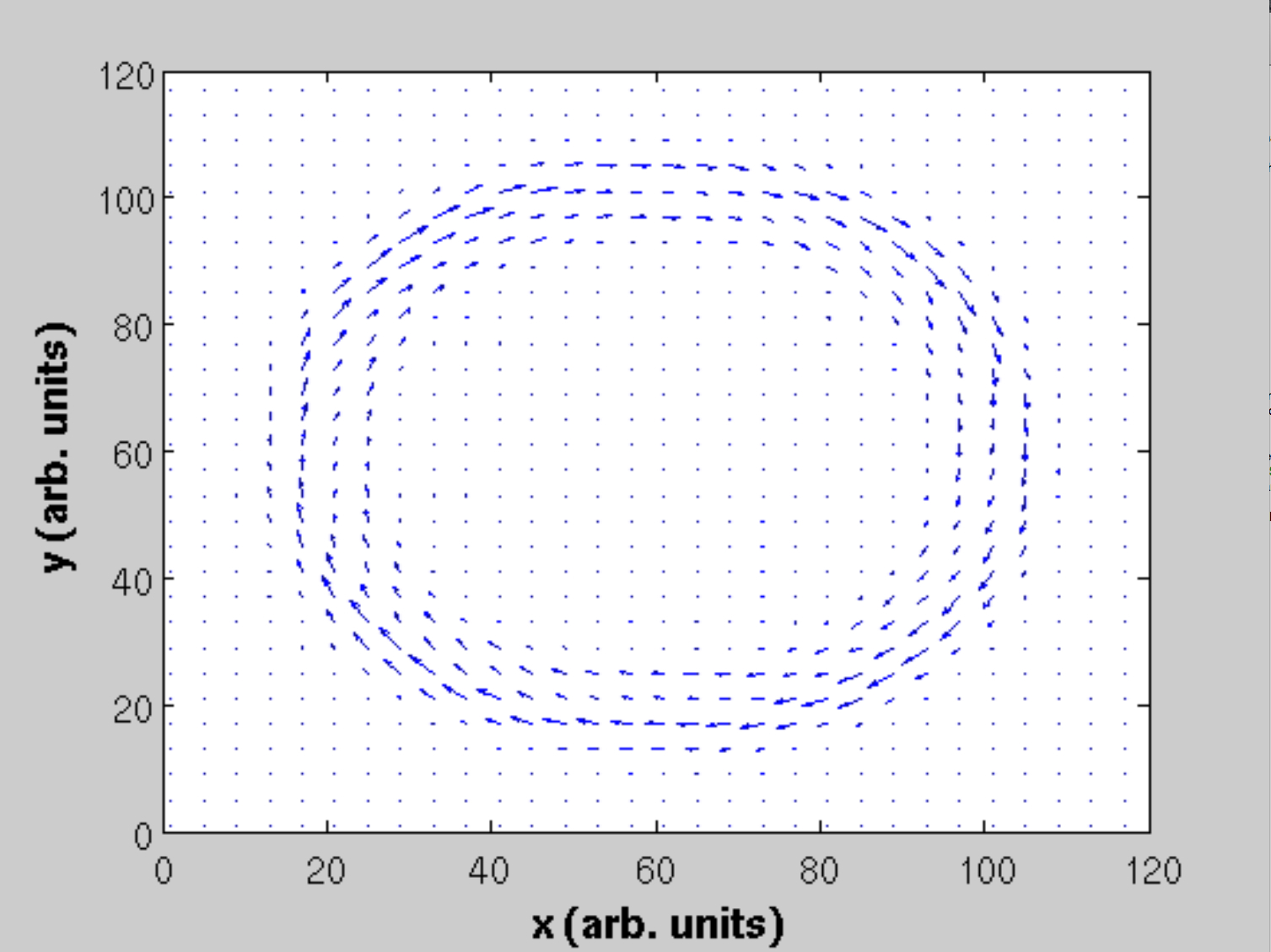}  (d) \includegraphics[width=6cm]{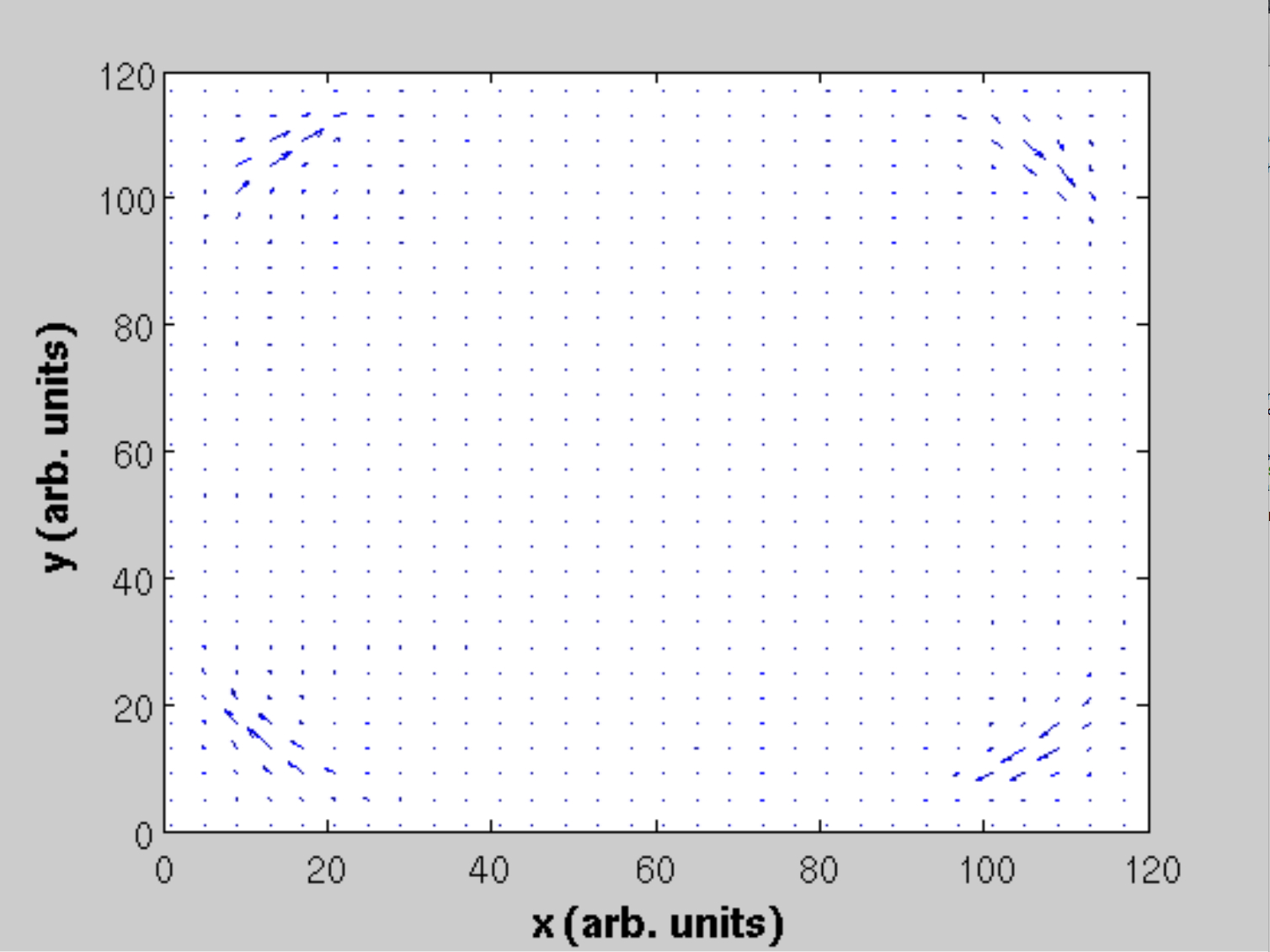}\\
~\\
\caption{Time-evolution of the spin current density of an incoherent superposition of four Gaussian wave packets.}
\label{Fig2}
\end{figure}

\subsection{Particle injection from an electric  contact}

For a second example we consider particle injection from an electric contact which is attached to one end of the rectangular simulation region.  Next to an absorbing layer placed around the simulation region to absorb out-flowing contributions, particle injection from the contact is simulated using the linearity of the Dirac equation which is preserved by the difference schemes: we cut the density matrix into two orthogonal pieces according to position: the density matrix on the grid within  the simulation region and the density matrix on the grid in the contact region.  We first construct the density matrix from plane-wave solutions of the free-particle Dirac equation with occupation probabilities given by the Fermi-Dirac distribution function at temperature $T=0K$.  We consider a filled Fermi sea of negative-energy states $E<0$ and study electron injection between in the energy interval $[0,E_F]$.   Since the simulation is based on periodic boundary conditions it is somewhat advantageous to use injection based on a one-sided Fermi-Dirac distribution, i.e. to omit outgoing states from the start rather than damping them out in the absorption layer.  In a first step we compute one time-step for the evolution of the truncated density matrix in the contact and record the contribution to the density matrix which arises on the grid of the simulation region.  From then on this contribution is added step-by-step to the calculation of the time evolution of the density matrix on the grid of the simulation region.  The basic assumption, common in transport simulations,  is that particle injection provided by the contact is independent of the outflux of particles.  

The buildup of particle density due to the presence of a contact at the left end of the simulation region is shown in Fig. \ref{Fig3}.  For this simulation 120x120 mesh points, including 2 absorbing layers, adjacent and opposite to the contact as shown in Fig. \ref{Fig4}(a), were used.  A Fermi energy  $E_F=0.01$eV and 50 randomly generated incident plane waves were used to construct the equilibrium density matrix. 
The simulation was performed over 500 time steps (third images in Fig. \ref{Fig3}).  At final simulation time, steady-state is not quite reached.  Also noticeable are two slight ridges in particle density propagating into the system. These can be attributed to the fact, that the construction of the equilibrium density matrix is based on plane waves for an infinite system, while the simulation is based on periodic boundary conditions 
(in y-direction).   They are not resolved in the spin texture.   A more detailed study of contact simulations based on this finite-difference approach will be given elsewhere.\cite{Poellau}

\begin{figure}[h]
\centering
(a)  \includegraphics[width=5cm]{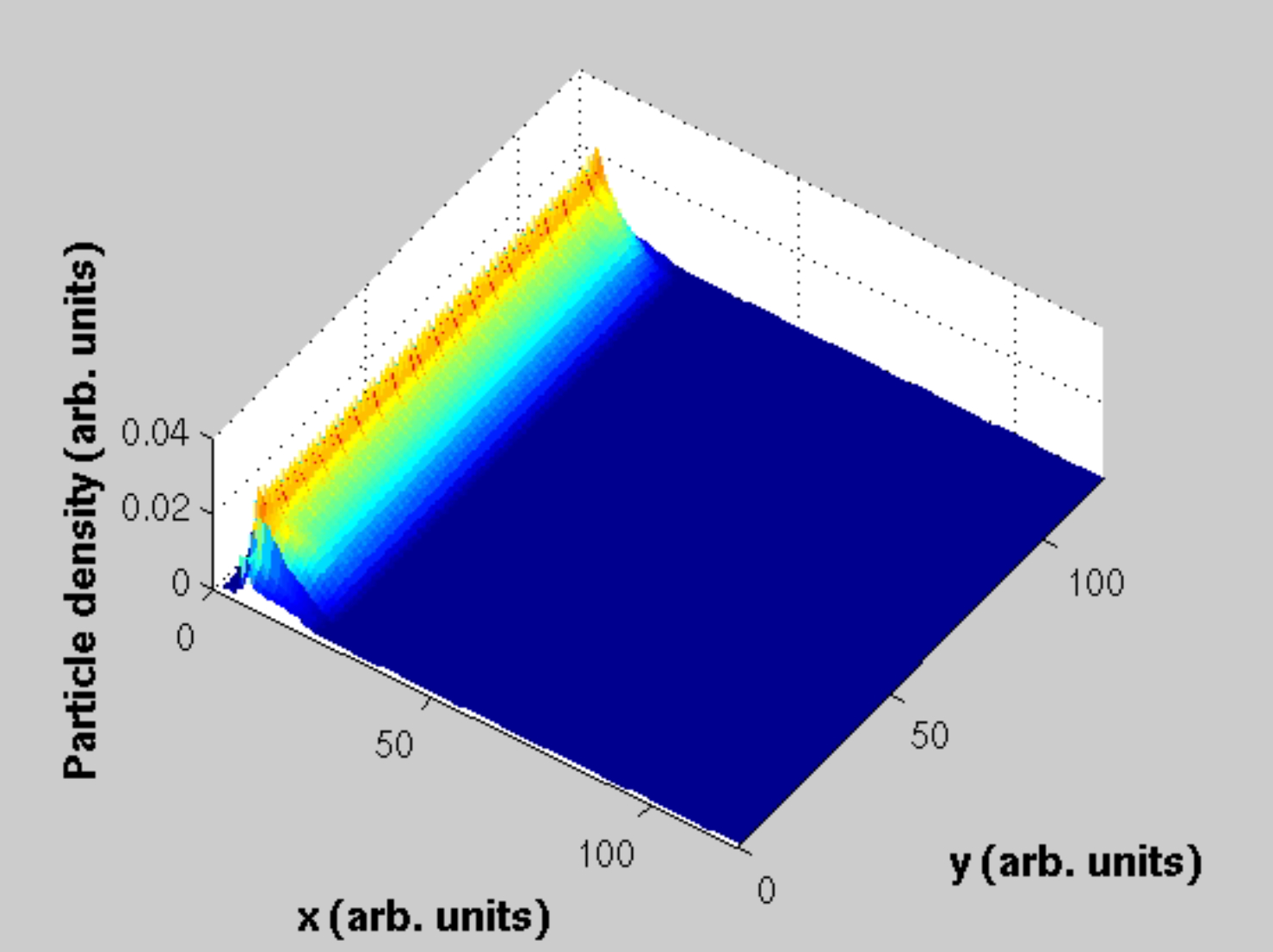}   \includegraphics[width=5cm]{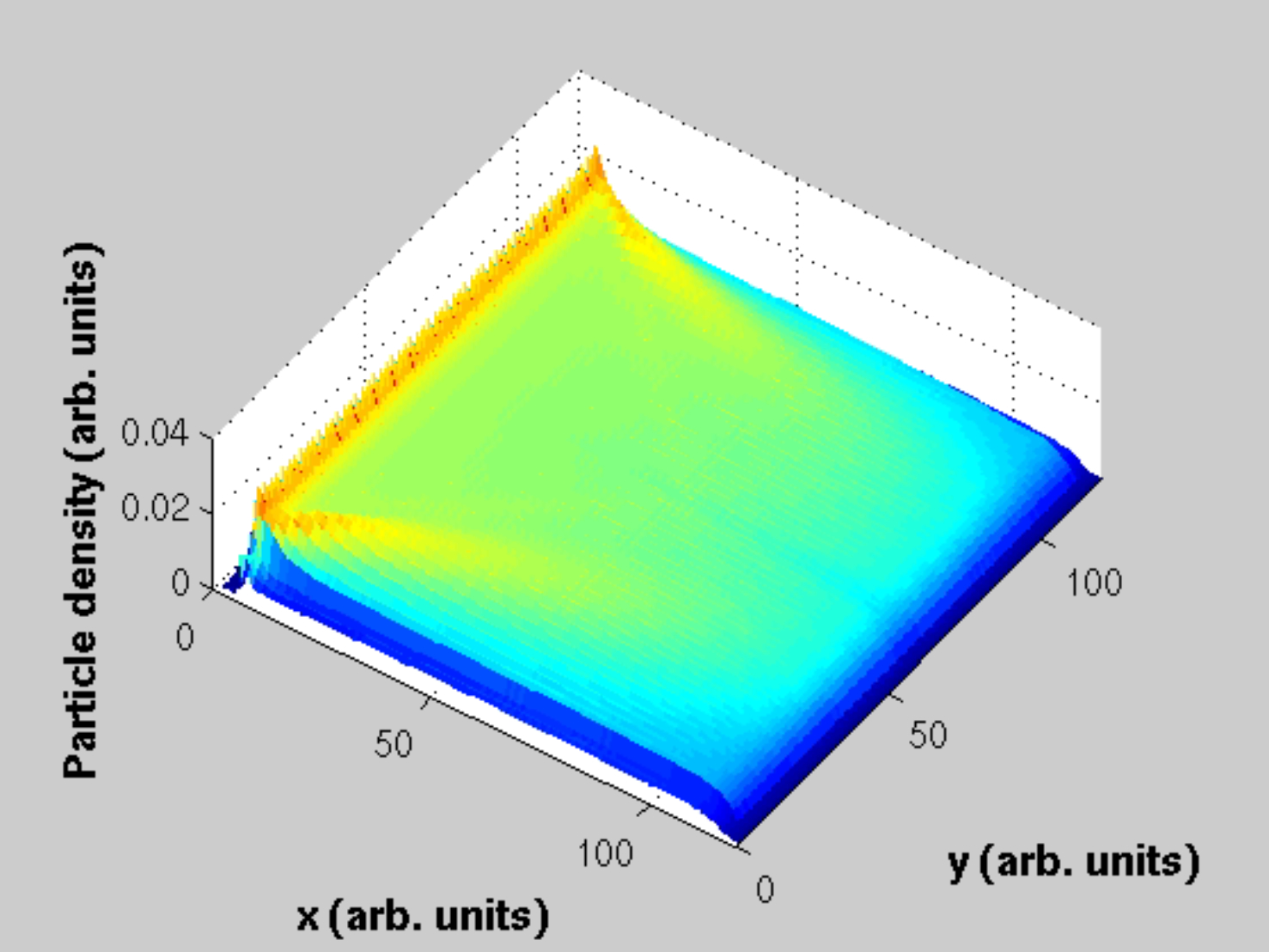}   \includegraphics[width=5cm]{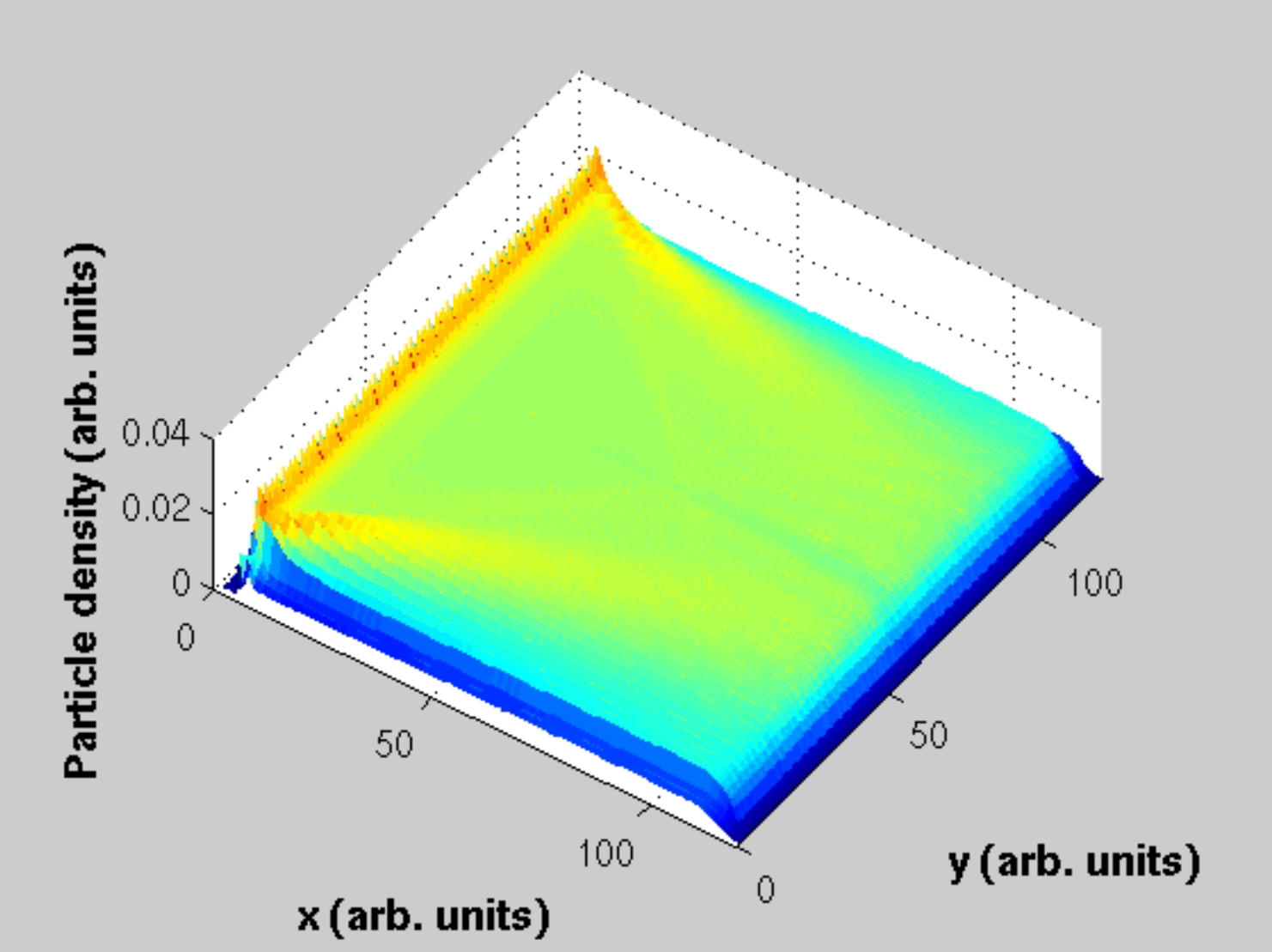}\\
~\\
(b)  \includegraphics[width=5cm]{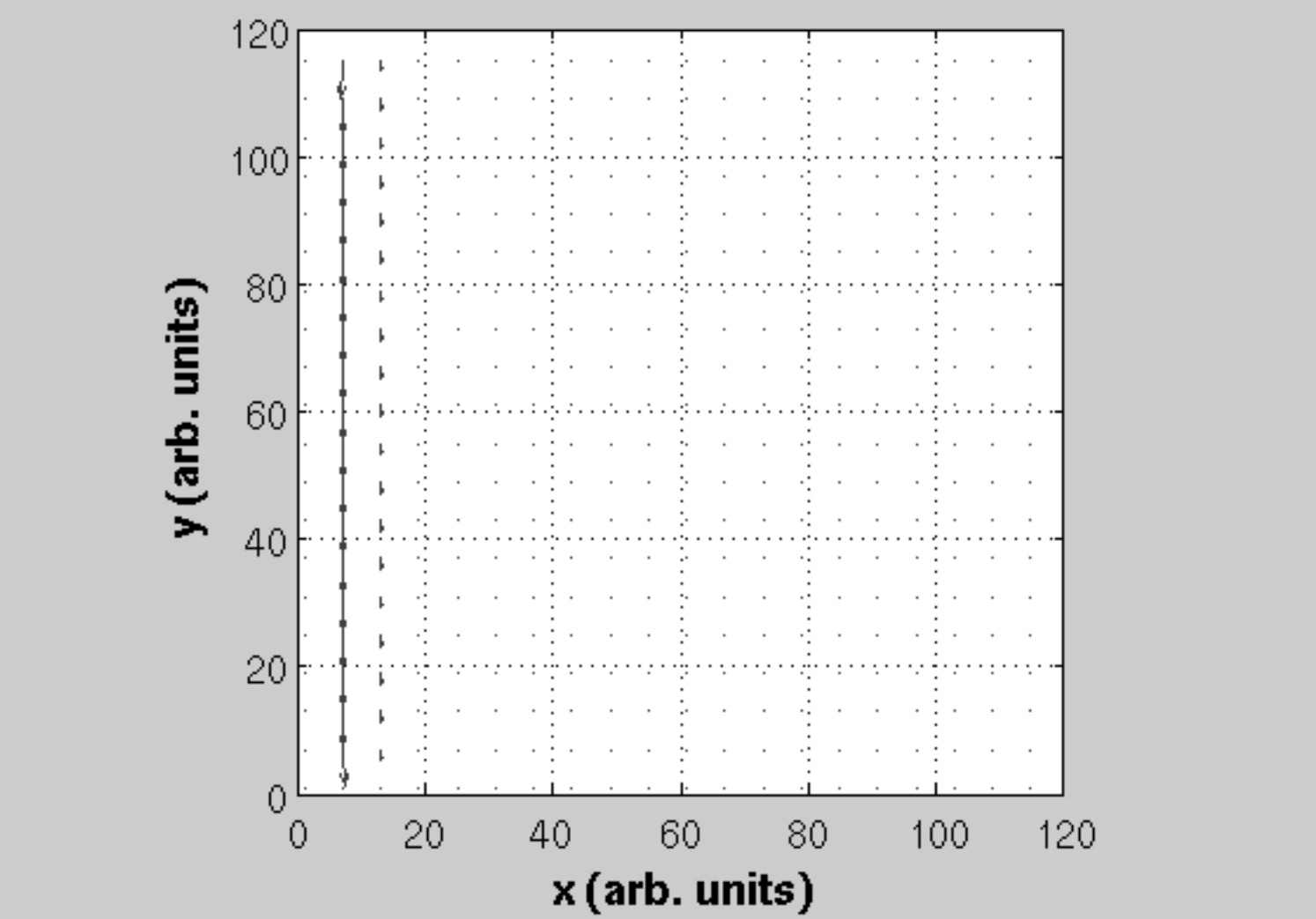}   \includegraphics[width=5cm]{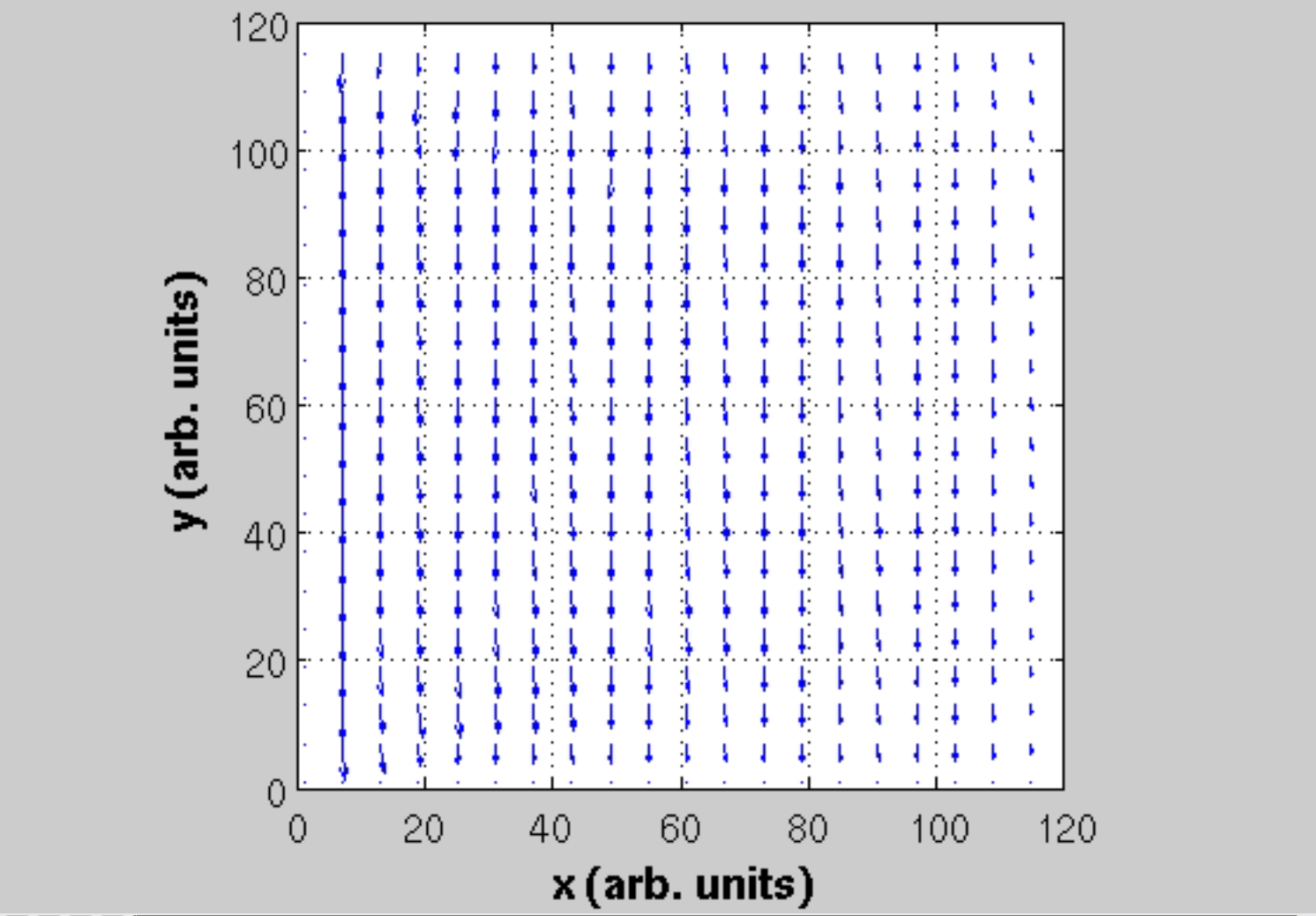}   \includegraphics[width=5cm]{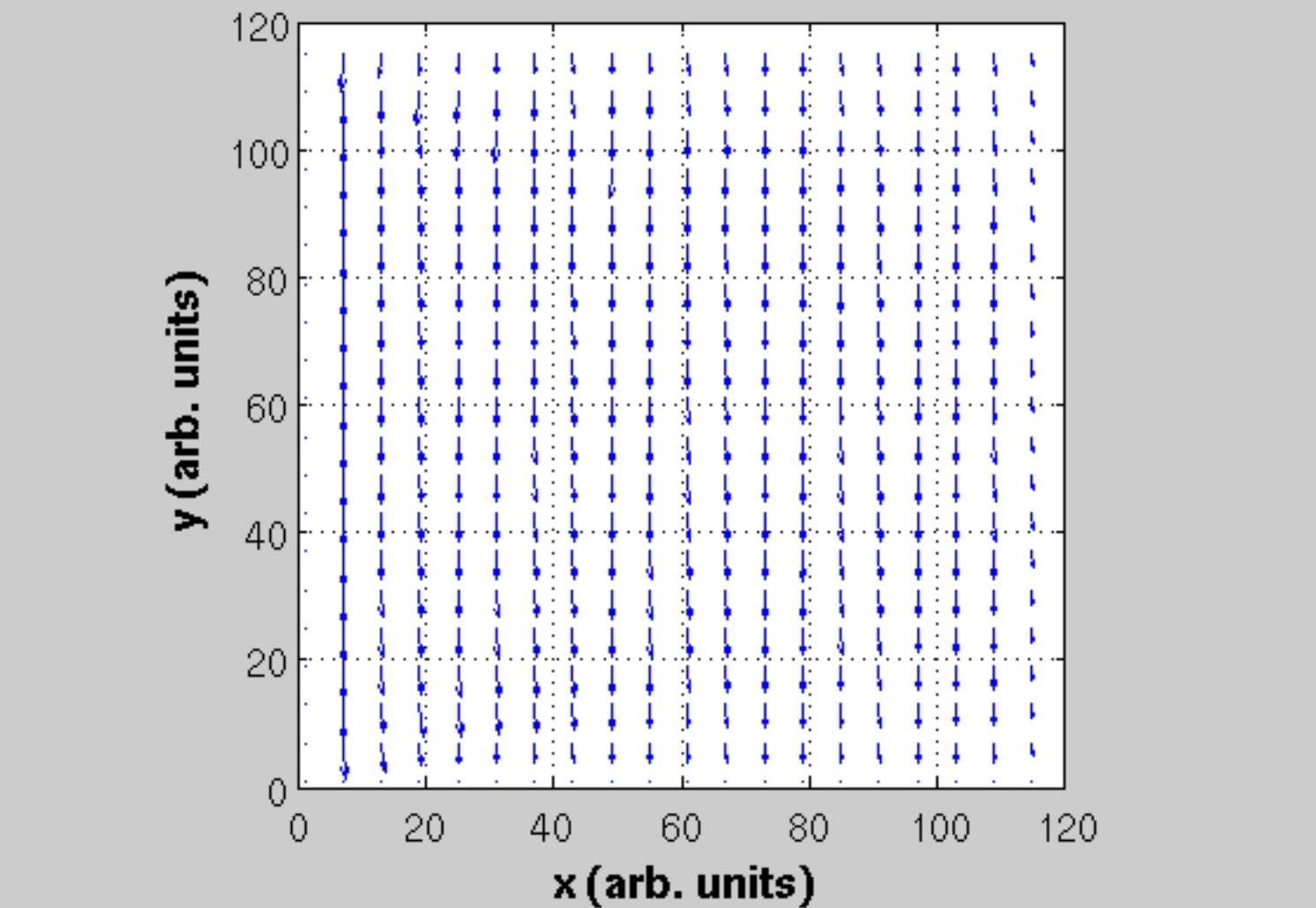}\\
~\\
\caption{(a) Charge buildup due to particle injection from a left-sided contact (particle reservoir) in thermal equilibrium at $T=0$.  The simulation region of 900nm$\times$900nm is discretized by 120$\times$120 mesh points. The corresponding applied bias is $0.02$eV. (b) Evolution of spin texture due to particle injection.}
\label{Fig3}
\end{figure}

Examples for imaginary potential contributions for the modeling of absorbing boundary conditions are given in Fig. \ref{Fig4}.  

\begin{figure}[h]
\centering
(a)  \includegraphics[width=5cm]{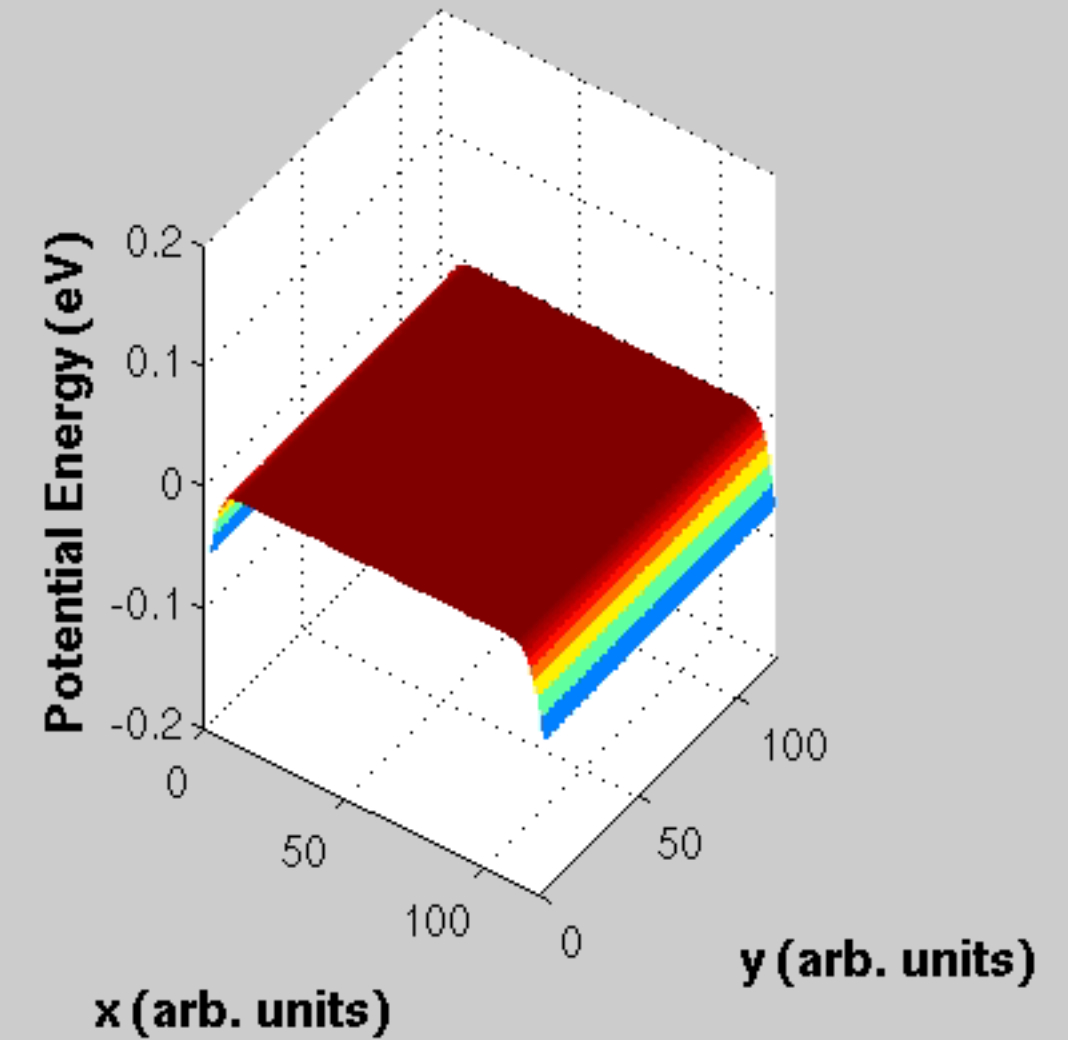}  (b)  \includegraphics[width=5cm]{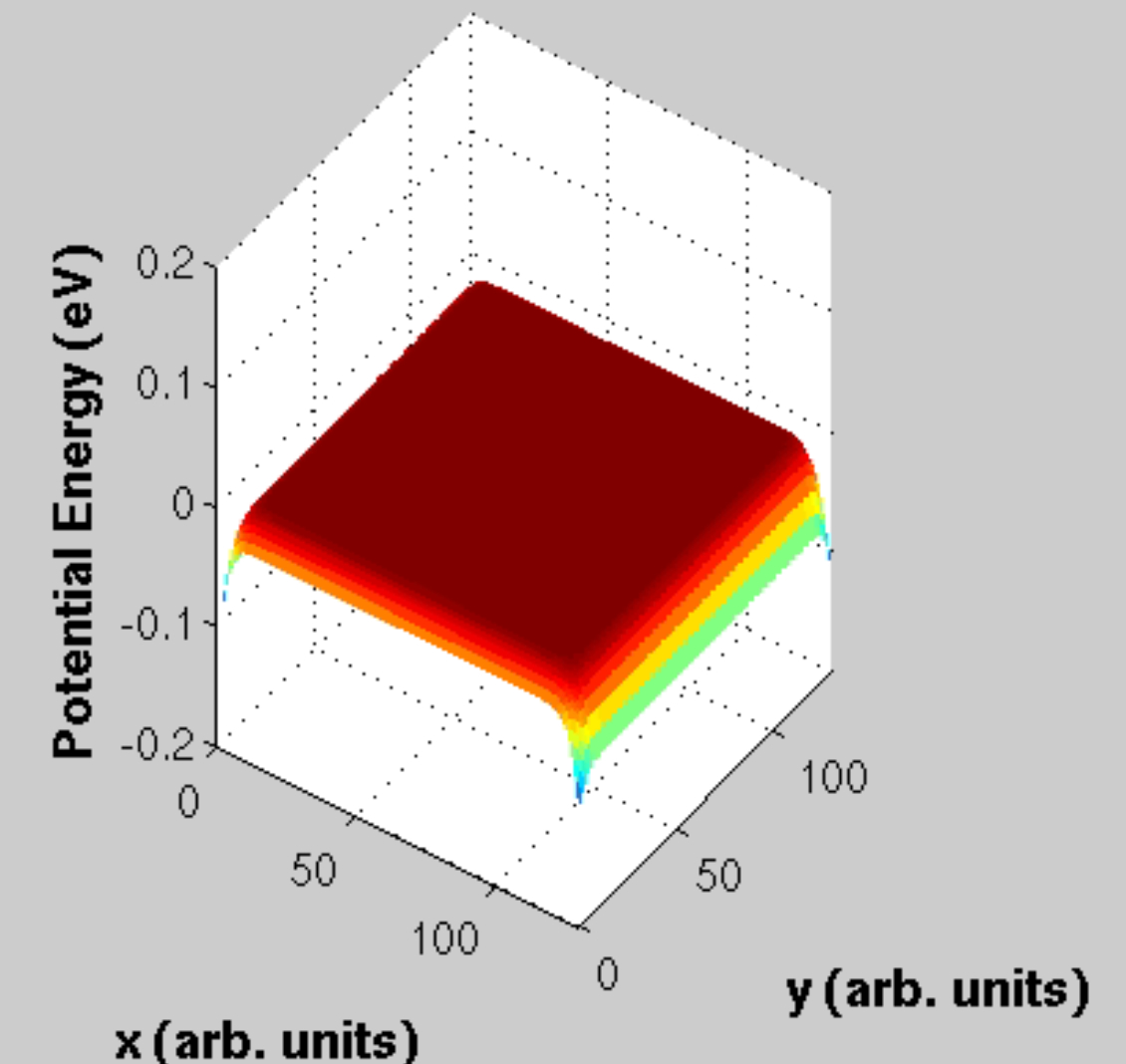} \\
\caption{(a) Imaginary potential contribution used to model absorbing boundary conditions: (a) absorbing layer at the contact end  and opposite side only, (b) absorbing layer surrounding the simulation region.}
\label{Fig4}
\end{figure}

\begin{figure}[h]
\centering
(a)  \includegraphics[width=6cm]{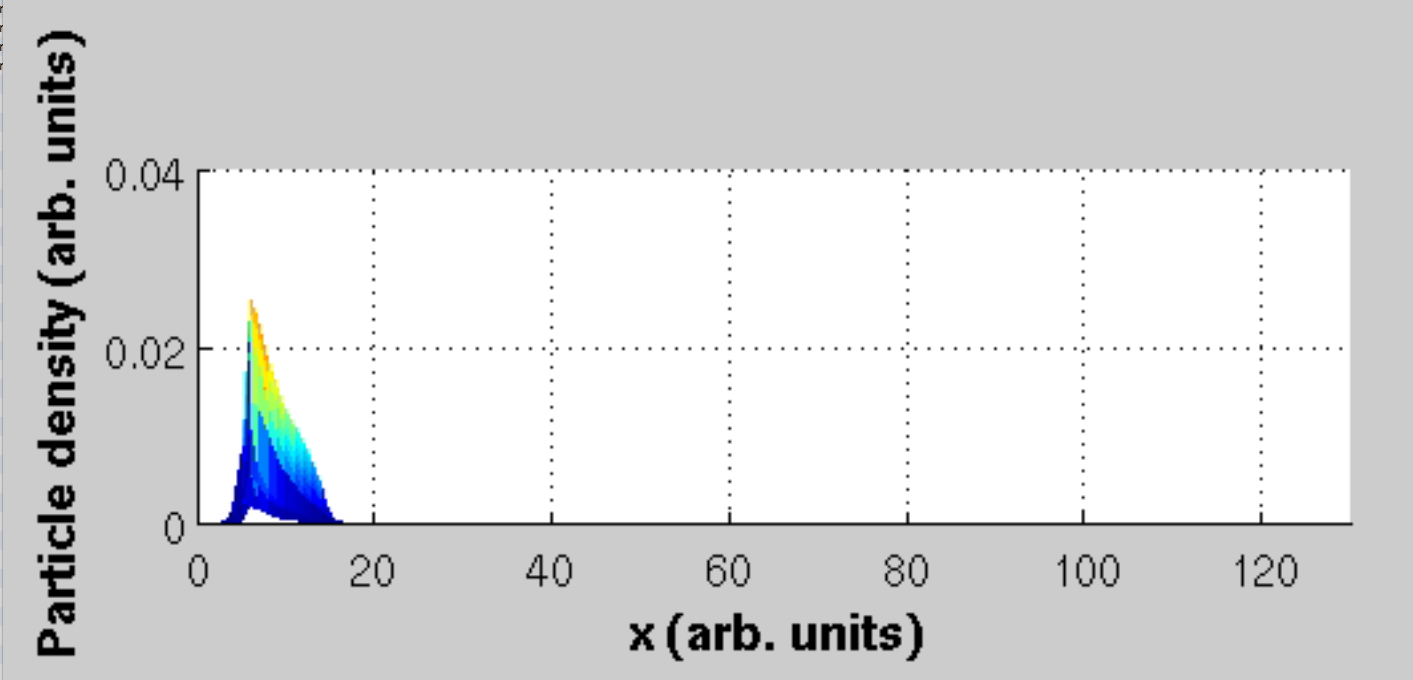}   (b) \includegraphics[width=6cm]{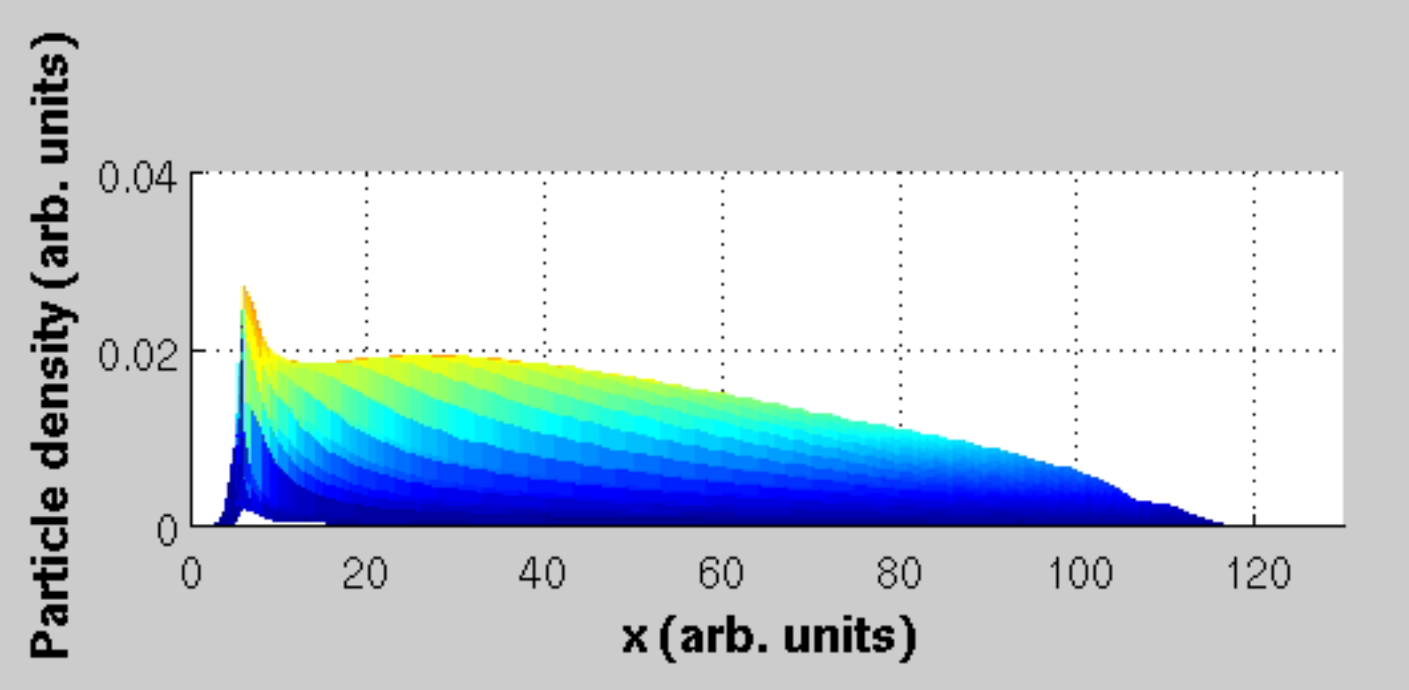} \\
(c)  \includegraphics[width=6cm]{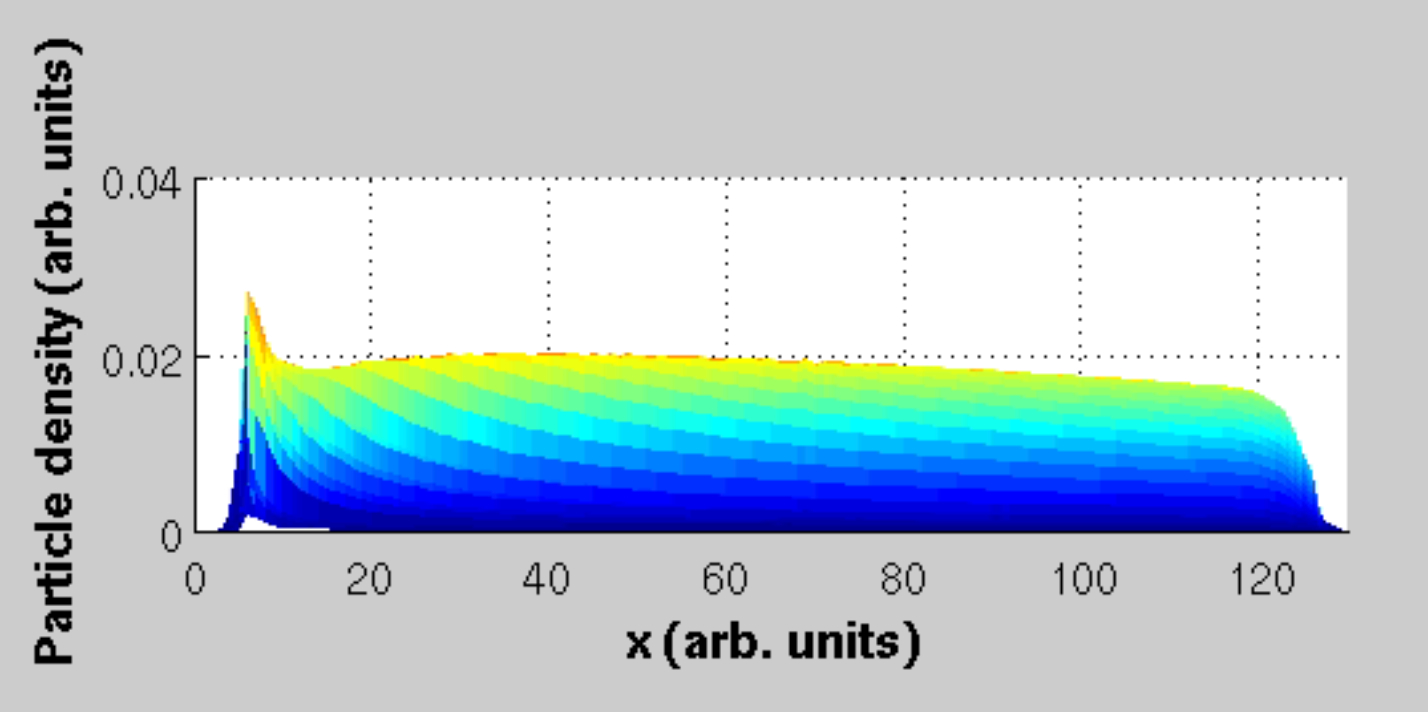}   (d)  \includegraphics[width=6cm]{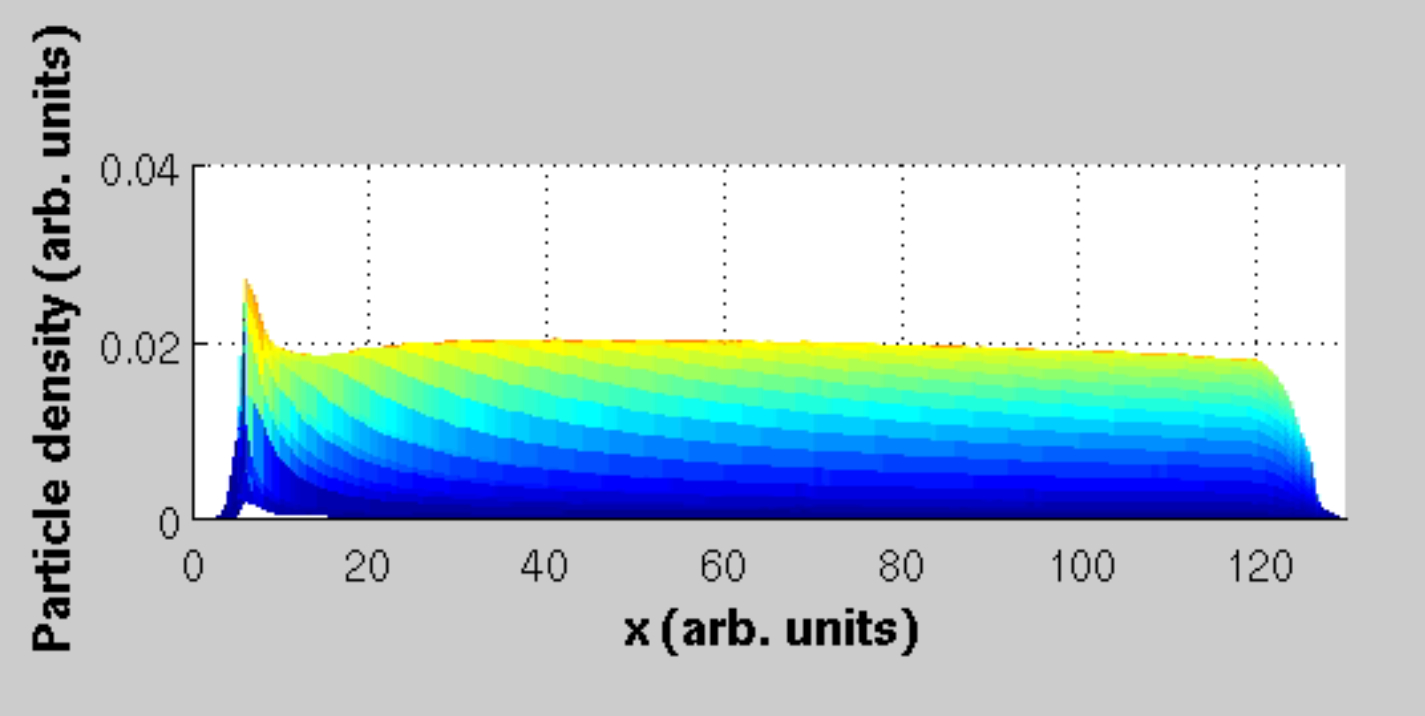} 
~\\
\caption{ Side view of the charge buildup due to particle injection from a left-sided contact (particle reservoir) in thermal equilibrium at $T=0$.  The simulation region of $L_x\times L_y$=900nm$\times$900nm is discretized by 130$\times$130 mesh points.  The applied bias is $0.02$eV. (a) 50 time steps, (b) 200 time steps, (c) 400 time steps, (d) 600 time steps. }
\label{Fig5}
\end{figure}

In Fig. \ref{Fig5} we show a side view of the charge build-up due to the action of an electric contact.  Input data are as above, however, a grid size of 130$\times$130 and the imaginary potential displayed in Fig. \ref{Fig4}(b) was used.  The time step $\Delta_t$ is 3.9$\times10^{-15}$s. Towards the central corridor of the simulation region fairly homogeneous charge distribution is reached.  Due to the additional absorption regions (compared to Fig. \ref{Fig4}(a)) near $y=0$ and $y=L_y$ parallel to the average propagation direction $x$ some loss of particle density is obtained at the y-boundaries.  Steady state is reached at around 600 time steps.

\section{Extensions of the schemes}\label{EXT}

\subsection{External magnetic fields}

The numerical schemes presented here can directly be used to model the surface-state contributions to electron transport in topological insulators.  These contributions inherently provide spin-polarized electric currents which may find application in spintronic devices.   Disorder effects within this approach can be treated at a microscopic level via an explicit  account of potential fluctuations.  For this purpose it will be useful to implement general electromagnetic potentials into the scheme.
In presence of the  vector potential the kinetic momentum entering the Hamiltonian Eq. \eqref{H} is modified from ${\bf p}$ to  $\left({\bf p}-\frac{q}{c} {\bf {\cal A}}(x,y,t)\right)$\footnote{$q$ is the fermion electric charge.} and is incorporated via the  Peierls substitution for the spinor components $\psi$ on the lattice \cite{peierls,hammer1d,hammer2D2cone}.   For the present case it takes the form  
\begin{eqnarray}
{\rho_{ij}}_{x_i,y_i; x_j,y_j}^{t_i ;t_j}  \rightarrow { \hat{\rho_{ij}}}_{x_i,y_i; x_j,y_j}^{t_i ;t_j} := \exp\{ -i ( a^{t_i}_{x_i,y_i}-a^{t_j}_{x_j,y_j})\} {\rho_{ij}}_{x_i,y_i; x_j,y_j}^{t_i ;t_j}~.
\label{peierlss}
\end{eqnarray}
The 
real phase contribution $a^{t_i}_{x_i,y_i}$ is defined as the line integral of the  vector potential ${\bf {\cal A}}$, starting at an arbitrary, but fixed position $(x_o,y_o)$ on the grid and ending on the lattice point 
$(x_i,y_i)$,
$$
a^{t_i}_{x_i,y_i}= \frac{q}{\hbar c}\int_{(x_o,y_o)}^{(x,y)} d{\bf s}\cdot {\bf {\cal A}}({\bf s},t)\mid_{x=x_i, y=y_i, t=t_i}~.
$$
Replacing ${\rho_{ij}}$ by the Peierls-transformed ${\hat\rho_{ij}}$ in Def.  \ref{bra-ket} and Def.  \ref{defdi} gives the respective numerical scheme in presence of an electromagnetic vector potential.  Therefore,   the gauge-invariant introduction of the electromagnetic vector potential via the Peierls substitution on the grid leaves the stability properties of the respective scheme intact.

\subsection{Lindblad equation and Green's function approach}

The two schemes presented here can be extended to a numerical treatment of the Lindblad equation for the Dirac Hamiltonian
\begin{equation}\label{Eqrho1}
\dot{\rho}(t)=\frac{1}{i\hbar} \left[H(t),\rho(t)\right]+ D\{\rho\},
\end{equation}
with the dissipator 
\begin{equation}\label{D}
D\{\rho\} =\sum_\mu D_\mu\{\rho\}  =\sum_\mu\left[L_\mu\rho(t)L_\mu^{+}-
\frac{1}{2}\{L_\mu^{+}L_\mu,\rho(t)\}\right]~,
\end{equation}
and  initial condition $\rho(0)=\rho_0$,
the dissipation term  containing Lindblad operators $L_\mu, \mu=1,...$ with 
the curly bracket $\{...,...\}$ denoting the anti--commutator.   Discretization of the dissipator follows the placement of the density matrix elements on the staggered grid,  as discussed in Sect. \ref{placement}.  The Lindblad extension allows a treatment of Dirac fermions as an open quantum system, i.e. an account of dissipative effects (within a Markov approximation) in the system arising, for example, from electric contacts, phonons, nuclear magnetic moments, and particle transfer between surface and bulk states. 

Next to the density matrix approach, the Green's function approach is the most versatile approach used in non-equilibrium electron  transport theory.\cite{wien2}  Since the (single-particle) Green's function fulfills the Dirac equation, in a real-space-time formulation of the latter a single-cone  implementation on a lattice is ensured following the procedure outlined in this paper.   The main difference is that the (single-particle) Green's function inherently uses two time indices, while the density matrix only needs one (in the continuum limit).   Hence, on the lattice introduced above, a single-particle Green's function matrix element needs up to four time sheets.

\section{Summary and Outlook}

In summary we have proposed two closely related single-cone finite-difference schemes for the numerical treatment of the von Neumann equation for (2+1)D Dirac fermions.  Main mathematical properties, such as stability the underlying energy--momentum dispersion relation, have been explored.  Elementary numerical examples for Dirac fermion propagation characterized by particle density and spin texture have been given.   

Numerical solution of the time--dependent Dirac equation is quite memory intensive.  Comparison of the two schemes shows that the advantage of the direct scheme regarding the number of matrix multiplications per time step is compensated by the reduced radius of convergence when compared to the bra--ket scheme.  In MATLAB which we used for numerical analysis global matrix operations as compared to do--loops are most efficient.  Therefore it is important to keep the number of full-sized matrices to a minimum.   Our schemes are direct and are highly efficient since they scale proportional to the square of the number of spatial grid points.  Nevertheless, 
the computation of the coherent part of the time-evolution is redundant since it is performed on left- and right-sided matrix indices independently.  As evident from the bra-ket scheme, one may proceed as follows: decompose the initial density matrix
in its orthonormal  basis according to Eq. \eqref{mixeds} and propagate the ket, transpose the ket to get the bra's time-evolution, and construct the density matrix.\cite{hammer3d}  Depending on the number of basis functions needed to represent the problem this may be significantly less memory expensive.   Extension to a treatment of the Lindblad equation along these lines, termed the quantum--jump approach,  has been explored for the case of the Schr\"{o}dinger equation.\cite{Gardiner-Zoller} Conceptual advantages of a density matrix or Green's function formulation arise when dissipative effects and interaction processes, in general,  need to be included, in particular for systems with a continuous spectrum  typical to quantum transport.\cite{Datta} 

Open questions arise regarding the treatment of simulation boundaries.  In the present work a combination of periodic and absorbing boundary conditions has been used to suppress communication between opposite simulation boundaries.  We are currently exploring the potential  of 
Lindblad operators to imprint the signature of  electric contacts onto the simulation region.\cite{Poellau}  Finally there are open questions regarding conservation of positivity of the direct scheme for the case of general position- and time-dependent electromagnetic potentials.


\newpage
%

\begin{appendix}
\section{Explicit  form of the direct scheme \label{A}}

Here we give the explicit form of  (\ref{11}) - (\ref{22}) for the direct scheme.  Pairs of subscripts denote spatial grid position, the superscript denotes the time sheet.  For clarity, we use double time superscripts throughout.\\

For density matrix elements initially placed on ${\cal G}_1 (j_x, j_y, j_0-\frac{1}{2})$ and  ${\cal G}_2 (j_x + \frac{1}{2}, j_y, j_0)$, potential and mass terms on $\bar{\cal G}_1(j_x,j_y,j_0)$ and $\bar{\cal G}_2 (j_x+\frac{1}{2}, j_y, j_0 +\frac{1}{2})$ one has
\begin{align}
\begin{array}{l}
\frac{{{\rho_{11}}}_{j_x,j_y;j_{x'},j_{y'}}^{j_o + \frac{1}{2}; j_{o} + \frac{1}{2}} - {\rho_{11}}_{j_x,j_y;j_{x'},j_{y'}}^{j_o - \frac{1}{2}; j_{o} - \frac{1}{2}} }{\Delta_t} =
\frac{1}{i \hbar}({m_+}_{j_x, j_y}^{j_o} - {m_+}_{ j_{x'}, j_{y'}}^{ j_{o}}) \frac{{\rho_{11}}_{j_x,j_y;j_{x'},j_{y'}}^{j_o + \frac{1}{2}; j_{o} + \frac{1}{2}} + {\rho_{11}}_{j_x,j_y;j_{x'},j_{y'}}^{j_o - \frac{1}{2}, j_{o} - \frac{1}{2}}}{2}\nonumber \\ -
v\frac{{\rho_{21}}_{j_x,j_y + \frac{1}{2};j_{x'},j_{y'}}^{j_o ; j_{o}- \frac{1}{2}} -{\rho_{21}}_{j_x,j_y - \frac{1}{2};j_{x'},j_{y'}}^{j_o ; j_{o}- \frac{1}{2}} }{\Delta_y}-
iv\frac{{\rho_{21}}_{j_x+ \frac{1}{2},j_y ;j_{x'},j_{y'}}^{j_o ; j_{o}- \frac{1}{2}} -{\rho_{21}}_{j_x- \frac{1}{2},j_y ;j_{x'},j_{y'}}^{j_o ; j_{o}- \frac{1}{2} }}{\Delta_x}\nonumber \\ -
v\frac{{\rho_{12}}_{j_x,j_y ;j_{x'},j_{y'}+ \frac{1}{2}}^{j_o-\frac{1}{2} ; j_{o} } -{\rho_{12}}_{j_x,j_y ;j_{x'},j_{y'}- \frac{1}{2}}^{j_o- \frac{1}{2} ; j_{o}} }{\Delta_{y'}} +
iv\frac{{\rho_{12}}_{j_x,j_y ;j_{x'}+ \frac{1}{2},j_{y'}}^{j_o-\frac{1}{2} ; j_{o} } -{\rho_{12}}_{j_x,j_y ;j_{x'}- \frac{1}{2},j_{y'}}^{j_o- \frac{1}{2} ; j_{o}} }{\Delta_{x'}}~,
\end{array}
\end{align}

\begin{align}
\begin{array}{l}
\frac{{\rho_{12}}_{j_x,j_y;j_{x'}+\frac{1}{2},j_{y'}}^{j_o + \frac{1}{2}; j_{o} + 1} - {\rho_{12}}_{j_x,j_y;j_{x'}+\frac{1}{2},j_{y'}}^{j_o - \frac{1}{2}; j_{o} } }{\Delta_t} =
\frac{1}{i \hbar}({m_+}_{j_x, j_y}^{j_o} - {m_-}_{ j_{x'}+\frac{1}{2}, j_{y'}}^{ j_{o}+\frac{1}{2}}) \frac{{\rho_{12}}_{j_x,j_y;j_{x'}+\frac{1}{2},j_{y'}}^{j_o + \frac{1}{2}; j_{o}+1} + {\rho_{12}}_{j_x,j_y;j_{x'}+\frac{1}{2},j_{y'}}^{j_o - \frac{1}{2}; j_{o} }}{2}\nonumber \\ -
v\frac{{\rho_{22}}_{j_x,j_y + \frac{1}{2};j_{x'}+\frac{1}{2},j_{y'}}^{j_o ; j_{o}} -{\rho_{22}}_{j_x,j_y - \frac{1}{2};j_{x'}+\frac{1}{2},j_{y'}}^{j_o ; j_{o}} }{\Delta_y}-
iv\frac{{\rho_{22}}_{j_x+ \frac{1}{2},j_y ;j_{x'}+\frac{1}{2},j_{y'}}^{j_o ; j_{o}} -{\rho_{22}}_{j_x- \frac{1}{2},j_y ;j_{x'}+\frac{1}{2},j_{y'}}^{j_o; j_{o} }}{\Delta_x}\nonumber \\ -
v\frac{{\rho_{11}}_{j_x,j_y ;j_{x'}+\frac{1}{2},j_{y'}+ \frac{1}{2}}^{j_o+\frac{1}{2} ; j_{o}+\frac{1}{2} } -{\rho_{11}}_{j_x,j_y ;j_{x'}+\frac{1}{2} ,j_{y'}- \frac{1}{2}}^{j_o+ \frac{1}{2} ; j_{o}+\frac{1}{2} } }{\Delta_{y'}} -
iv\frac{{\rho_{11}}_{j_x,j_y ;j_{x'}+1,j_{y'}}^{j_o+\frac{1}{2} ; j_{o}+\frac{1}{2}  } -{\rho_{11}}_{j_x,j_y ;j_{x'},j_{y'}}^{j_o+ \frac{1}{2} ; j_{o}+\frac{1}{2} } }{\Delta_{x'}}~,
\end{array}
\end{align}

\begin{align}
\begin{array}{l}
\frac{{\rho_{21}}_{j_x+\frac{1}{2},j_y;j_{x'},j_{y'}}^{j_o + 1; j_{o} + \frac{1}{2}} - {\rho_{21}}_{j_x+\frac{1}{2},j_y;j_{x'},j_{y'}}^{j_o ; j_{o}-\frac{1}{2} } }{\Delta_t} =\nonumber \\
\frac{1}{i \hbar}({m_-}_{j_x+\frac{1}{2}, j_y}^{j_o+\frac{1}{2}} - {m_+}_{ j_{x'}, j_{y'}}^{ j_{o}}) \frac{{\rho_{21}}_{j_x+\frac{1}{2},j_y;j_{x'},j_{y'}}^{j_o + 1; j_{o}+\frac{1}{2}} + {\rho_{21}}_{j_x+\frac{1}{2},j_y;j_{x'},j_{y'}}^{j_o ; j_{o} -\frac{1}{2} }}{2}\nonumber \\-
v\frac{{\rho_{11}}_{j_x+\frac{1}{2},j_y + \frac{1}{2};j_{x'},j_{y'}}^{j_o+\frac{1}{2} ; j_{o}+\frac{1}{2}} -{\rho_{11}}_{j_x+\frac{1}{2},j_y - \frac{1}{2};j_{x'},j_{y'}}^{j_o+\frac{1}{2} ; j_{o}+\frac{1}{2}} }{\Delta_y}+
iv\frac{{\rho_{11}}_{j_x+1,j_y ;j_{x'},j_{y'}}^{j_o+\frac{1}{2} ; j_{o}+\frac{1}{2}} -{\rho_{11}}_{j_x,j_y ;j_{x'},j_{y'}}^{j_o+\frac{1}{2} ; j_{o}+\frac{1}{2} }}{\Delta_x}\nonumber \\-
v\frac{{\rho_{22}}_{j_x+\frac{1}{2},j_y ;j_{x'},j_{y'}+ \frac{1}{2}}^{j_o ; j_{o}} -{\rho_{22}}_{j_x+\frac{1}{2},j_y ;j_{x'},j_{y'}- \frac{1}{2}}^{j_o ; j_{o} } }{\Delta_{y'}} +
iv\frac{{\rho_{22}}_{j_x+\frac{1}{2},j_y ;j_{x'}+\frac{1}{2},j_{y'}}^{j_o ; j_{o}  } -{\rho_{22}}_{j_x+\frac{1}{2},j_y ;j_{x'}- \frac{1}{2},j_{y'}}^{j_o ; j_{o} } }{\Delta_{x'}}~,
\end{array}
\end{align}

\begin{align}
\begin{array}{l}
\frac{{\rho_{22}}_{j_x+\frac{1}{2},j_y;j_{x'}+\frac{1}{2},j_{y'}}^{j_o + 1; j_{o} + 1} - {\rho_{22}}_{j_x+\frac{1}{2},j_y;j_{x'}+\frac{1}{2},j_{y'}}^{j_o ; j_{o} } }{\Delta_t} = \nonumber \\
\frac{1}{i \hbar}({m_-}_{j_x+\frac{1}{2}, j_y}^{j_o+\frac{1}{2}} - {m_-}_{ j_{x'}+\frac{1}{2}, j_{y'}}^{ j_{o}+\frac{1}{2}}) \frac{{\rho_{22}}_{j_x+\frac{1}{2},j_y;j_{x'}+\frac{1}{2},j_{y'}}^{j_o + 1; j_{o}+1} + {\rho_{22}}_{j_x+\frac{1}{2},j_y;j_{x'}+\frac{1}{2},j_{y'}}^{j_o ; j_{o} }}{2}\nonumber \\ -
v\frac{{\rho_{12}}_{j_x+\frac{1}{2},j_y + \frac{1}{2};j_{x'}+\frac{1}{2},j_{y'}}^{j_o+\frac{1}{2} ; j_{o}+1} -{\rho_{12}}_{j_x+\frac{1}{2},j_y - \frac{1}{2};j_{x'}+\frac{1}{2},j_{y'}}^{j_o+\frac{1}{2} ; j_{o}+1} }{\Delta_y}+
iv\frac{{\rho_{12}}_{j_x+1,j_y ;j_{x'}+\frac{1}{2},j_{y'}}^{j_o+\frac{1}{2} ; j_{o}+1} -{\rho_{12}}_{j_x,j_y ;j_{x'}+\frac{1}{2},j_{y'}}^{j_o+\frac{1}{2} ; j_{o}+1 }}{\Delta_x}\nonumber \\ -
v\frac{{\rho_{21}}_{j_x+\frac{1}{2},j_y ;j_{x'}+\frac{1}{2},j_{y'}+ \frac{1}{2}}^{j_o+1 ; j_{o}+\frac{1}{2}} -{\rho_{21}}_{j_x+\frac{1}{2},j_y ;j_{x'}+\frac{1}{2},j_{y'}- \frac{1}{2}}^{j_o+1 ; j_{o}+\frac{1}{2} } }{\Delta_{y'}} -
iv\frac{{\rho_{21}}_{j_x+\frac{1}{2},j_y ;j_{x'}+1,j_{y'}}^{j_o +1; j_{o}+\frac{1}{2}  } -{\rho_{21}}_{j_x+\frac{1}{2},j_y ;j_{x'},j_{y'}}^{j_o+1 ; j_{o}+\frac{1}{2} } }{\Delta_{x'}}~,
\end{array}
\end{align}

For density matrix elements on ${\cal G}_1 (j_x, j_y, j_0-\frac{1}{2})$ and $ {\cal G}_2 (j_x , j_y+ \frac{1}{2}, j_0)$,  potential and mass terms on $\bar{{\cal G}_1} (j_x,j_y,j_0), \bar{{\cal G}_2} (j_x, j_y+\frac{1}{2}, j_0 +\frac{1}{2}):$\\

\begin{align}
\begin{array}{l}
\frac{{{\rho_{11}}}_{j_x,j_y;j_{x'},j_{y'}}^{j_o + \frac{1}{2}; j_{o} + \frac{1}{2}} - {\rho_{11}}_{j_x,j_y;j_{x'},j_{y'}}^{j_o - \frac{1}{2}; j_{o} - \frac{1}{2}} }{\Delta_t} =
\frac{1}{i \hbar}({m_+}_{j_x, j_y}^{j_o} - {m_+}_{ j_{x'}, j_{y'}}^{ j_{o}}) \frac{{\rho_{11}}_{j_x,j_y;j_{x'},j_{y'}}^{j_o + \frac{1}{2}; j_{o} + \frac{1}{2}} + {\rho_{11}}_{j_x,j_y;j_{x'},j_{y'}}^{j_o - \frac{1}{2}, j_{o} - \frac{1}{2}}}{2}\nonumber \\ -
v\frac{{\rho_{21}}_{j_x,j_y + \frac{1}{2};j_{x'},j_{y'}}^{j_o ; j_{o}- \frac{1}{2}} -{\rho_{21}}_{j_x,j_y - \frac{1}{2};j_{x'},j_{y'}}^{j_o ; j_{o}- \frac{1}{2}} }{\Delta_y}-
iv\frac{{\rho_{21}}_{j_x+ \frac{1}{2},j_y ;j_{x'},j_{y'}}^{j_o ; j_{o}- \frac{1}{2}} -{\rho_{21}}_{j_x- \frac{1}{2},j_y ;j_{x'},j_{y'}}^{j_o ; j_{o}- \frac{1}{2} }}{\Delta_x}\nonumber \\ -
v\frac{{\rho_{12}}_{j_x,j_y ;j_{x'},j_{y'}+ \frac{1}{2}}^{j_o-\frac{1}{2} ; j_{o} } -{\rho_{12}}_{j_x,j_y ;j_{x'},j_{y'}- \frac{1}{2}}^{j_o- \frac{1}{2} ; j_{o}} }{\Delta_{y'}} +
iv\frac{{\rho_{12}}_{j_x,j_y ;j_{x'}+ \frac{1}{2},j_{y'}}^{j_o-\frac{1}{2} ; j_{o} } -{\rho_{12}}_{j_x,j_y ;j_{x'}- \frac{1}{2},j_{y'}}^{j_o- \frac{1}{2} ; j_{o}} }{\Delta_{x'}}~,
\end{array}
\end{align}

\begin{align}
\begin{array}{l}
\frac{{\rho_{12}}_{j_x,j_y;j_{x'},j_{y'}+\frac{1}{2}}^{j_o + \frac{1}{2}; j_{o} + 1} - {\rho_{12}}_{j_x,j_y;j_{x'},j_{y'}+\frac{1}{2}}^{j_o - \frac{1}{2}; j_{o} } }{\Delta_t} =
\frac{1}{i \hbar}({m_+}_{j_x, j_y}^{j_o} - {m_-}_{ j_{x'}, j_{y'}+\frac{1}{2}}^{ j_{o}+\frac{1}{2}}) \frac{{\rho_{12}}_{j_x,j_y;j_{x'},j_{y'}+\frac{1}{2}}^{j_o + \frac{1}{2}; j_{o}+1} + {\rho_{12}}_{j_x,j_y;j_{x'},j_{y'}+\frac{1}{2}}^{j_o - \frac{1}{2}; j_{o} }}{2}\nonumber \\ -
v\frac{{\rho_{22}}_{j_x,j_y + \frac{1}{2};j_{x'},j_{y'}+\frac{1}{2}}^{j_o ; j_{o}} -{\rho_{22}}_{j_x,j_y - \frac{1}{2};j_{x'},j_{y'}+\frac{1}{2}}^{j_o ; j_{o}} }{\Delta_y}-
iv\frac{{\rho_{22}}_{j_x+ \frac{1}{2},j_y ;j_{x'},j_{y'}+\frac{1}{2}}^{j_o ; j_{o}} -{\rho_{22}}_{j_x- \frac{1}{2},j_y ;j_{x'}+,j_{y'}\frac{1}{2}}^{j_o; j_{o} }}{\Delta_x}\nonumber \\ -
v\frac{{\rho_{11}}_{j_x,j_y ;j_{x'},j_{y'}+ 1}^{j_o+\frac{1}{2} ; j_{o}+\frac{1}{2} } -{\rho_{11}}_{j_x,j_y ;j_{x'} ,j_{y'}}^{j_o+ \frac{1}{2} ; j_{o}+\frac{1}{2} } }{\Delta_{y'}} -
iv\frac{{\rho_{11}}_{j_x,j_y ;j_{x'}+\frac{1}{2},j_{y'}+\frac{1}{2}}^{j_o+\frac{1}{2} ; j_{o}+\frac{1}{2}  } -{\rho_{11}}_{j_x,j_y ;j_{x'}-\frac{1}{2},j_{y'}+\frac{1}{2}}^{j_o+ \frac{1}{2} ; j_{o}+\frac{1}{2} } }{\Delta_{x'}}~,
\end{array}
\end{align}

\begin{align}
\begin{array}{l}
\frac{{\rho_{21}}_{j_x,j_y+\frac{1}{2};j_{x'},j_{y'}}^{j_o + 1; j_{o} + \frac{1}{2}} - {\rho_{21}}_{j_x,j_y+\frac{1}{2};j_{x'},j_{y'}}^{j_o ; j_{o}-\frac{1}{2} } }{\Delta_t} =\nonumber \\
\frac{1}{i \hbar}({m_-}_{j_x, j_y+\frac{1}{2}}^{j_o+\frac{1}{2}} - {m_+}_{ j_{x'}, j_{y'}}^{ j_{o}}) \frac{{\rho_{21}}_{j_x+\frac{1}{2},j_y;j_{x'},j_{y'}}^{j_o + 1; j_{o}+\frac{1}{2}} + {\rho_{21}}_{j_x,j_y+\frac{1}{2};j_{x'},j_{y'}}^{j_o ; j_{o} -\frac{1}{2} }}{2}\nonumber \\-
v\frac{{\rho_{11}}_{j_x,j_y + 1;j_{x'},j_{y'}}^{j_o+\frac{1}{2} ; j_{o}+\frac{1}{2}} -{\rho_{11}}_{j_x+,j_y;j_{x'},j_{y'}}^{j_o+\frac{1}{2} ; j_{o}+\frac{1}{2}} }{\Delta_y}+
iv\frac{{\rho_{11}}_{j_x+\frac{1}{2},j_y+\frac{1}{2} ;j_{x'},j_{y'}}^{j_o+\frac{1}{2} ; j_{o}+\frac{1}{2}} -{\rho_{11}}_{j_x-\frac{1}{2},j_y+\frac{1}{2} ;j_{x'},j_{y'}}^{j_o+\frac{1}{2} ; j_{o}+\frac{1}{2} }}{\Delta_x}\nonumber \\-
v\frac{{\rho_{22}}_{j_x,j_y+\frac{1}{2} ;j_{x'},j_{y'}+ \frac{1}{2}}^{j_o ; j_{o}} -{\rho_{22}}_{j_x,j_y+\frac{1}{2} ;j_{x'},j_{y'}- \frac{1}{2}}^{j_o ; j_{o} } }{\Delta_{y'}} +
iv\frac{{\rho_{22}}_{j_x,j_y+\frac{1}{2} ;j_{x'}+\frac{1}{2},j_{y'}}^{j_o ; j_{o}  } -{\rho_{22}}_{j_x,j_y+\frac{1}{2} ;j_{x'}- \frac{1}{2},j_{y'}}^{j_o ; j_{o} } }{\Delta_{x'}}~,
\end{array}
\end{align}

\begin{align}
\begin{array}{l}
\frac{{\rho_{22}}_{j_x,j_y+\frac{1}{2};j_{x'},j_{y'}+\frac{1}{2}}^{j_o + 1; j_{o} + 1} - {\rho_{22}}_{j_x,j_y+\frac{1}{2};j_{x'},j_{y'}+\frac{1}{2}}^{j_o ; j_{o} } }{\Delta_t} = \nonumber \\
\frac{1}{i \hbar}({m_-}_{j_x, j_y+\frac{1}{2}}^{j_o+\frac{1}{2}} - {m_-}_{ j_{x'}, j_{y'}+\frac{1}{2}}^{ j_{o}+\frac{1}{2}}) \frac{{\rho_{22}}_{j_x,j_y+\frac{1}{2};j_{x'},j_{y'}+\frac{1}{2}}^{j_o + 1; j_{o}+1} + {\rho_{22}}_{j_x,j_y+\frac{1}{2};j_{x'},j_{y'}+\frac{1}{2}}^{j_o ; j_{o} }}{2}\nonumber \\ -
v\frac{{\rho_{12}}_{j_x,j_y + 1;j_{x'},j_{y'}+\frac{1}{2}}^{j_o+\frac{1}{2} ; j_{o}+1} -{\rho_{12}}_{j_x,j_y ;j_{x'},j_{y'}+\frac{1}{2}}^{j_o+\frac{1}{2} ; j_{o}+1} }{\Delta_y}+
iv\frac{{\rho_{12}}_{j_x+\frac{1}{2},j_y+\frac{1}{2} ;j_{x'},j_{y'}+\frac{1}{2}}^{j_o+\frac{1}{2} ; j_{o}+1} -{\rho_{12}}_{j_x-\frac{1}{2},j_y +\frac{1}{2};j_{x'},j_{y'}+\frac{1}{2}}^{j_o+\frac{1}{2} ; j_{o}+1 }}{\Delta_x}\nonumber \\ -
v\frac{{\rho_{21}}_{j_x,j_y+\frac{1}{2} ;j_{x'},j_{y'}+ 1}^{j_o+1 ; j_{o}+\frac{1}{2}} -{\rho_{21}}_{j_x,j_y+\frac{1}{2} ;j_{x'},j_{y'}}^{j_o+1 ; j_{o}+\frac{1}{2} } }{\Delta_{y'}} -
iv\frac{{\rho_{21}}_{j_x,j_y+\frac{1}{2} ;j_{x'}+\frac{1}{2},j_{y'}+\frac{1}{2}}^{j_o +1; j_{o}+\frac{1}{2}  } -{\rho_{21}}_{j_x,j_y+\frac{1}{2} ;j_{x'}- \frac{1}{2},j_{y'}+ \frac{1}{2}}^{j_o+1 ; j_{o}+\frac{1}{2} } }{\Delta_{x'}}~,
\end{array}
\end{align}

For density matrix elements placed on ${\cal G}_1 (j_x+\frac{1}{2}, j_y+\frac{1}{2}, j_o-\frac{1}{2})$ and $ {\cal G}_2 (j_x+ \frac{1}{2} , j_y, j_o)$;  potential and mass terms on $\bar{{\cal G}_1} (j_x+\frac{1}{2},j_y+\frac{1}{2},j_o),\bar{{\cal G}_2} (j_x+\frac{1}{2}, j_y, j_o +\frac{1}{2})$\\

\begin{align}
\begin{array}{l}
\frac{{\rho_{11}}_{j_x+\frac{1}{2},j_y+\frac{1}{2};j_{x'}+\frac{1}{2},j_{y'}+\frac{1}{2}}^{j_o + \frac{1}{2}; j_{o} + \frac{1}{2}} - {\rho_{11}}_{j_x+\frac{1}{2},j_y+\frac{1}{2};j_{x'}+\frac{1}{2},j_{y'}+\frac{1}{2}}^{j_o - \frac{1}{2}; j_{o} - \frac{1}{2}} }{\Delta_o} =\nonumber \\
\frac{1}{i \hbar}({m_+}_{j_x+\frac{1}{2}, j_y+\frac{1}{2}}^{j_o} - {m_+}_{ j_{x'}+\frac{1}{2}, j_{y'}+\frac{1}{2}}^{ j_{o}}) \frac{{\rho_{11}}_{j_x+\frac{1}{2},j_y+\frac{1}{2};j_{x'}+\frac{1}{2},j_{y'}+\frac{1}{2}}^{j_o + \frac{1}{2}; j_{o} + \frac{1}{2}} + {\rho_{11}}_{j_x+\frac{1}{2},j_y+\frac{1}{2};j_{x'}+\frac{1}{2},j_{y'}+\frac{1}{2}}^{j_o - \frac{1}{2}, j_{o} - \frac{1}{2}}}{2}\nonumber \\ -
v\frac{{\rho_{21}}_{j_x+\frac{1}{2},j_y + 1;j_{x'},j_{y'}}^{j_o ; j_{o}- \frac{1}{2}} -{\rho_{21}}_{j_x+\frac{1}{2},j_y ;j_{x'},j_{y'}}^{j_o ; j_{o}- \frac{1}{2}} }{\Delta_y}-
iv\frac{{\rho_{21}}_{j_x+ 1,j_y+\frac{1}{2} ;j_{x'},j_{y'}}^{j_o ; j_{o}- \frac{1}{2}} -{\rho_{21}}_{{j_x},j_y+\frac{1}{2} ;j_{x'},j_{y'}}^{j_o ; j_{o}- \frac{1}{2} }}{\Delta_x}\nonumber \\ -
v\frac{{\rho_{12}}_{j_x+\frac{1}{2},j_y+\frac{1}{2} ;j_{x'},j_{y'}+ \frac{1}{2}}^{j_o-\frac{1}{2} ; j_{o} } -{\rho_{12}}_{j_x+\frac{1}{2},j_y+\frac{1}{2} ;j_{x'},j_{y'}- \frac{1}{2}}^{j_o- \frac{1}{2} ; j_{o}} }{\Delta_{y'}}+
iv\frac{{\rho_{12}}_{j_x+\frac{1}{2},j_y+\frac{1}{2} ;j_{x'}+ \frac{1}{2},j_{y'}}^{j_o-\frac{1}{2} ; j_{o} } -{\rho_{12}}_{j_x+\frac{1}{2},j_y+\frac{1}{2} ;j_{x'}- \frac{1}{2},j_{y'}}^{j_o- \frac{1}{2} ; j_{o}} }{\Delta_{x'}}~,
\end{array}
\end{align}

\begin{align}
\begin{array}{l}
\frac{{\rho_{12}}_{j_x+\frac{1}{2},j_y+\frac{1}{2};j_{x'}+\frac{1}{2},j_{y'}}^{j_o + \frac{1}{2}; j_{o} + 1} - {\rho_{12}}_{j_x+\frac{1}{2},j_y+\frac{1}{2};j_{x'}+\frac{1}{2},j_{y'}}^{j_o - \frac{1}{2}; j_{o} } }{\Delta_o} =\nonumber \\
\frac{1}{i \hbar}({m_+}_{j_x+\frac{1}{2}, j_y+\frac{1}{2}}^{j_o} - {m_-}_{ j_{x'}+\frac{1}{2}, j_{y'}}^{ j_{o}+\frac{1}{2}}) \frac{{\rho_{12}}_{j_x+\frac{1}{2},j_y+\frac{1}{2};j_{x'}+\frac{1}{2},j_{y'}}^{j_o + \frac{1}{2}; j_{o}+1} + {\rho_{12}}_{j_x+\frac{1}{2},j_y+\frac{1}{2};j_{x'}+\frac{1}{2},j_{y'}}^{j_o - \frac{1}{2}; j_{o} }}{2}\nonumber \\ -
v\frac{{\rho_{22}}_{j_x+\frac{1}{2},j_y +1;j_{x'}+\frac{1}{2},j_{y'}}^{j_o ; j_{o}} -{\rho_{22}}_{j_x+\frac{1}{2},j_y;j_{x'}+\frac{1}{2},j_{y'}}^{j_o ; j_{o}} }{\Delta_y}-
iv\frac{{\rho_{22}}_{j_x+1,j_y+\frac{1}{2};j_{x'}+\frac{1}{2},j_{y'}}^{j_o ; j_{o}} -{\rho_{22}}_{j_x,j_y+\frac{1}{2} ;j_{x'}+\frac{1}{2},j_{y'}}^{j_o; j_{o} }}{\Delta_x}\nonumber \\ -
v\frac{{\rho_{11}}_{j_x+\frac{1}{2},j_y+\frac{1}{2} ;j_{x'}+\frac{1}{2},j_{y'}+\frac{1}{2}}^{j_o+\frac{1}{2} ; j_{o}+\frac{1}{2} } -{\rho_{11}}_{j_x+\frac{1}{2},j_y+\frac{1}{2} ;j_{x'}+\frac{1}{2} ,j_{y'}-\frac{1}{2}}^{j_o+ \frac{1}{2} ; j_{o}+\frac{1}{2} } }{\Delta_{y'}} -
iv\frac{{\rho_{11}}_{j_x+\frac{1}{2},j_y+\frac{1}{2} ;j_{x'}+1,j_{y'}}^{j_o+\frac{1}{2} ; j_{o}+\frac{1}{2}  } -{\rho_{11}}_{j_x+\frac{1}{2},j_y+\frac{1}{2} ;j_{x'},j_{y'}}^{j_o+ \frac{1}{2} ; j_{o}+\frac{1}{2} } }{\Delta_{x'}}~,
\end{array}
\end{align}

\begin{align}
\begin{array}{l}
\frac{{\rho_{21}}_{j_x+\frac{1}{2},j_y;j_{x'}+\frac{1}{2},j_{y'}+\frac{1}{2}}^{j_o + 1; j_{o} + \frac{1}{2}} - {\rho_{21}}_{j_x+\frac{1}{2},j_y;j_{x'}+\frac{1}{2},j_{y'}+\frac{1}{2}}^{j_o ; j_{o}-\frac{1}{2} } }{\Delta_o} =\nonumber \\
\frac{1}{i \hbar}({m_-}_{j_x+\frac{1}{2}, j_y}^{j_o+\frac{1}{2}} - {m_+}_{ j_{x'}+\frac{1}{2}, j_{y'}+\frac{1}{2}}^{ j_{o}}) \frac{{\rho_{21}}_{j_x+\frac{1}{2},j_y;j_{x'}+\frac{1}{2},j_{y'}+\frac{1}{2}}^{j_o + 1; j_{o}+\frac{1}{2}} + {\rho_{21}}_{j_x+\frac{1}{2},j_y;j_{x'}+\frac{1}{2},j_{y'}+\frac{1}{2}}^{j_o ; j_{o} -\frac{1}{2} }}{2}\nonumber \\ -
v\frac{{\rho_{11}}_{j_x+\frac{1}{2},j_y+\frac{1}{2};j_{x'}+\frac{1}{2},j_{y'}+\frac{1}{2}}^{j_o+\frac{1}{2} ; j_{o}+\frac{1}{2}} -{\rho_{11}}_{j_x+\frac{1}{2},j_y-\frac{1}{2};j_{x'}+\frac{1}{2},j_{y'}+\frac{1}{2}}^{j_o+\frac{1}{2} ; j_{o}+\frac{1}{2}} }{\Delta_y}+
iv\frac{{\rho_{11}}_{j_x+1,j_y;j_{x'}+\frac{1}{2},j_{y'}+\frac{1}{2}}^{j_o+\frac{1}{2} ; j_{o}+\frac{1}{2}} -{\rho_{11}}_{j_x,j_y ;j_{x'}+\frac{1}{2},j_{y'}+\frac{1}{2}}^{j_o+\frac{1}{2} ; j_{o}+\frac{1}{2} }}{\Delta_x}\nonumber \\ -
v\frac{{\rho_{22}}_{j_x+\frac{1}{2},j_y ;j_{x'}+\frac{1}{2},j_{y'}+ 1}^{j_o ; j_{o}} -{\rho_{22}}_{j_x+\frac{1}{2,j_y} ;j_{x'}+\frac{1}{2},j_{y'}}^{j_o ; j_{o} } }{\Delta_{y'}}+
iv\frac{{\rho_{22}}_{j_x+\frac{1}{2},j_y ;j_{x'}+1,j_{y'}+\frac{1}{2}}^{j_o ; j_{o}  } -{\rho_{22}}_{j_x+\frac{1}{2},j_y ;j_{x'},j_{y'}+\frac{1}{2}}^{j_o ; j_{o} } }{\Delta_{x'}}~,
\end{array}
\end{align}

\begin{align}
\begin{array}{l}
\frac{{\rho_{22}}_{j_x+\frac{1}{2},j_y;j_{x'}+\frac{1}{2},j_{y'}}^{j_o + 1; j_{o} + 1} - {\rho_{22}}_{j_x+\frac{1}{2},j_y;j_{x'}+\frac{1}{2},j_{y'}}^{j_o ; j_{o} } }{\Delta_o} = \nonumber \\
\frac{1}{i \hbar}({m_-}_{j_x+\frac{1}{2}, j_y}^{j_o+\frac{1}{2}} - {m_-}_{ j_{x'}+\frac{1}{2}, j_{y'}}^{ j_{o}+\frac{1}{2}}) \frac{{\rho_{22}}_{j_x+\frac{1}{2},j_y;j_{x'}+\frac{1}{2},j_{y'}}^{j_o + 1; j_{o}+1} + {\rho_{22}}_{j_x+\frac{1}{2},j_y;j_{x'}+\frac{1}{2},j_{y'}}^{j_o ; j_{o} }}{2}\nonumber \\ -
v\frac{{\rho_{12}}_{j_x+\frac{1}{2},j_y+\frac{1}{2};j_{x'}+\frac{1}{2},j_{y'}}^{j_o+\frac{1}{2} ; j_{o}} -{\rho_{12}}_{j_x+\frac{1}{2},j_y-\frac{1}{2} ;j_{x'}+\frac{1}{2},j_{y'}}^{j_o+\frac{1}{2} ; j_{o}} }{\Delta_y}+
iv\frac{{\rho_{12}}_{j_x,j_y+1 ;j_{x'}+\frac{1}{2},j_{y'}}^{j_o+\frac{1}{2} ; j_{o}+1} -{\rho_{12}}_{j_x,j_y;j_{x'}+\frac{1}{2},j_{y'}}^{j_o+\frac{1}{2} ; j_{o}+1 }}{\Delta_x}\nonumber \\-
v\frac{{\rho_{21}}_{j_x+\frac{1}{2},j_y ;j_{x'}+\frac{1}{2},j_{y'}+\frac{1}{2}}^{j_o +1; j_{o}+\frac{1}{2}} -{\rho_{21}}_{j_x+\frac{1}{2},j_y ;j_{x'}+\frac{1}{2},j_{y'}-\frac{1}{2}}^{j_o+1 ; j_{o}+\frac{1}{2} } }{\Delta_{y'}}-
iv\frac{{\rho_{21}}_{j_x+\frac{1}{2},j_y ;j_{x'}+1,j_{y'}}^{j_o +1; j_{o}+\frac{1}{2}  } -{\rho_{21}}_{j_x+\frac{1}{2},j_y ;j_{x'},j_{y'}}^{j_o +1; j_{o}+\frac{1}{2} } }{\Delta_{x'}}~,
\end{array}
\end{align}

Finally, for density matrix elements initially placed on ${\cal G}_1 (j_x+\frac{1}{2}, j_y+\frac{1}{2}, j_o-\frac{1}{2})$ and ${\cal G}_2 (j_x , j_y+ \frac{1}{2}, j_o)$,  $ \bar{{\cal G}_1} (j_x+\frac{1}{2},j_y+\frac{1}{2},j_o),\bar{{\cal G}_2} (j_x, j_y+\frac{1}{2}, j_o +\frac{1}{2}):$\\
\begin{align}
\begin{array}{l}
\frac{{\rho_{11}}_{j_x+\frac{1}{2},j_y+\frac{1}{2};j_{x'}+\frac{1}{2},j_{y'}+\frac{1}{2}}^{j_o + \frac{1}{2}; j_{o} + \frac{1}{2}} - {\rho_{11}}_{j_x+\frac{1}{2},j_y+\frac{1}{2};j_{x'}+\frac{1}{2},j_{y'}+\frac{1}{2}}^{j_o - \frac{1}{2}; j_{o} - \frac{1}{2}} }{\Delta_o} =\nonumber \\
\frac{1}{i \hbar}({m_+}_{j_x+\frac{1}{2}, j_y+\frac{1}{2}}^{j_o} - {m_+}_{ j_{x'}+\frac{1}{2}, j_{y'}+\frac{1}{2}}^{ j_{o}}) \frac{{\rho_{11}}_{j_x+\frac{1}{2},j_y+\frac{1}{2};j_{x'}+\frac{1}{2},j_{y'}+\frac{1}{2}}^{j_o + \frac{1}{2}; j_{o} + \frac{1}{2}} + {\rho_{11}}_{j_x+\frac{1}{2},j_y+\frac{1}{2};j_{x'}+\frac{1}{2},j_{y'}+\frac{1}{2}}^{j_o - \frac{1}{2}, j_{o} - \frac{1}{2}}}{2}\nonumber \\ -
v\frac{{\rho_{21}}_{j_x+\frac{1}{2},j_y + 1;j_{x'},j_{y'}}^{j_o ; j_{o}- \frac{1}{2}} -{\rho_{21}}_{j_x+\frac{1}{2},j_y ;j_{x'},j_{y'}}^{j_o ; j_{o}- \frac{1}{2}} }{\Delta_y}-
iv\frac{{\rho_{21}}_{j_x+ 1,j_y+\frac{1}{2} ;j_{x'},j_{y'}}^{j_o ; j_{o}- \frac{1}{2}} -{\rho_{21}}_{{j_x},j_y+\frac{1}{2} ;j_{x'},j_{y'}}^{j_o ; j_{o}- \frac{1}{2} }}{\Delta_x}\nonumber \\ -
v\frac{{\rho_{12}}_{j_x+\frac{1}{2},j_y+\frac{1}{2} ;j_{x'},j_{y'}+ \frac{1}{2}}^{j_o-\frac{1}{2} ; j_{o} } -{\rho_{12}}_{j_x+\frac{1}{2},j_y+\frac{1}{2} ;j_{x'},j_{y'}- \frac{1}{2}}^{j_o- \frac{1}{2} ; j_{o}} }{\Delta_{y'}}+
iv\frac{{\rho_{12}}_{j_x+\frac{1}{2},j_y+\frac{1}{2} ;j_{x'}+ \frac{1}{2},j_{y'}}^{j_o-\frac{1}{2} ; j_{o} } -{\rho_{12}}_{j_x+\frac{1}{2},j_y+\frac{1}{2} ;j_{x'}- \frac{1}{2},j_{y'}}^{j_o- \frac{1}{2} ; j_{o}} }{\Delta_{x'}}~,
\end{array}
\end{align}

\begin{align}
\begin{array}{l}
\frac{{\rho_{12}}_{j_x+\frac{1}{2},j_y+\frac{1}{2};j_{x'},j_{y'}+\frac{1}{2}}^{j_o + \frac{1}{2}; j_{o} + 1} - {\rho_{12}}_{j_x+\frac{1}{2},j_y+\frac{1}{2};j_{x'},j_{y'}+\frac{1}{2}}^{j_o - \frac{1}{2}; j_{o} } }{\Delta_o} =\nonumber \\
\frac{1}{i \hbar}({m_+}_{j_x+\frac{1}{2}, j_y+\frac{1}{2}}^{j_o} - {m_-}_{ j_{x'}, j_{y'}+\frac{1}{2}}^{ j_{o}+\frac{1}{2}}) \frac{{\rho_{12}}_{j_x+\frac{1}{2},j_y+\frac{1}{2};j_{x'},j_{y'}+\frac{1}{2}}^{j_o + \frac{1}{2}; j_{o}+1} + {\rho_{12}}_{j_x+\frac{1}{2},j_y+\frac{1}{2};j_{x'},j_{y'}+\frac{1}{2}}^{j_o - \frac{1}{2}; j_{o} }}{2}\nonumber \\ -
v\frac{{\rho_{22}}_{j_x+\frac{1}{2},j_y +1;j_{x'},j_{y'}+\frac{1}{2}}^{j_o ; j_{o}} -{\rho_{22}}_{j_x+\frac{1}{2},j_y;j_{x'},j_{y'}+\frac{1}{2}}^{j_o ; j_{o}} }{\Delta_y}-
iv\frac{{\rho_{22}}_{j_x+1,j_y+\frac{1}{2};j_{x'},j_{y'}+\frac{1}{2}}^{j_o ; j_{o}} -{\rho_{22}}_{j_x,j_y+\frac{1}{2} ;j_{x'},j_{y'}+\frac{1}{2}}^{j_o; j_{o} }}{\Delta_x}\nonumber \\ -
v\frac{{\rho_{11}}_{j_x+\frac{1}{2},j_y+\frac{1}{2} ;j_{x'},j_{y'}+ 1}^{j_o+\frac{1}{2} ; j_{o}+\frac{1}{2} } -{\rho_{11}}_{j_x+\frac{1}{2},j_y+\frac{1}{2} ;j_{x'} ,j_{y'}}^{j_o+ \frac{1}{2} ; j_{o}+\frac{1}{2} } }{\Delta_{y'}} -
iv\frac{{\rho_{11}}_{j_x+\frac{1}{2},j_y+\frac{1}{2} ;j_{x'}+\frac{1}{2},j_{y'}+\frac{1}{2}}^{j_o+\frac{1}{2} ; j_{o}+\frac{1}{2}  } -{\rho_{11}}_{j_x+\frac{1}{2},j_y+\frac{1}{2} ;j_{x'}-\frac{1}{2},j_{y'}+\frac{1}{2}}^{j_o+ \frac{1}{2} ; j_{o}+\frac{1}{2} } }{\Delta_{x'}}~,
\end{array}
\end{align}

\begin{align}
\begin{array}{l}
\frac{{\rho_{21}}_{j_x,j_y+\frac{1}{2};j_{x'}+\frac{1}{2},j_{y'}+\frac{1}{2}}^{j_o + 1; j_{o} + \frac{1}{2}} - {\rho_{21}}_{j_x,j_y+\frac{1}{2};j_{x'}+\frac{1}{2},j_{y'}+\frac{1}{2}}^{j_o ; j_{o}-\frac{1}{2} } }{\Delta_o} =\nonumber \\
\frac{1}{i \hbar}({m_-}_{j_x, j_y+\frac{1}{2}}^{j_o+\frac{1}{2}} - {m_+}_{ j_{x'}+\frac{1}{2}, j_{y'}+\frac{1}{2}}^{ j_{o}}) \frac{{\rho_{21}}_{j_x,j_y+\frac{1}{2};j_{x'}+\frac{1}{2},j_{y'}+\frac{1}{2}}^{j_o + 1; j_{o}+\frac{1}{2}} + {\rho_{21}}_{j_x,j_y+\frac{1}{2};j_{x'}+\frac{1}{2},j_{y'}+\frac{1}{2}}^{j_o ; j_{o} -\frac{1}{2} }}{2}\nonumber \\ -
v\frac{{\rho_{11}}_{j_x,j_y + 1;j_{x'}+\frac{1}{2},j_{y'}+\frac{1}{2}}^{j_o+\frac{1}{2} ; j_{o}+\frac{1}{2}} -{\rho_{11}}_{j_x,j_y;j_{x'}+\frac{1}{2},j_{y'}+\frac{1}{2}}^{j_o+\frac{1}{2} ; j_{o}+\frac{1}{2}} }{\Delta_y}+
iv\frac{{\rho_{11}}_{j_x+\frac{1}{2},j_y+\frac{1}{2} ;j_{x'}+\frac{1}{2},j_{y'}+\frac{1}{2}}^{j_o+\frac{1}{2} ; j_{o}+\frac{1}{2}} -{\rho_{11}}_{j_x-\frac{1}{2},j_y+\frac{1}{2} ;j_{x'}+\frac{1}{2},j_{y'}+\frac{1}{2}}^{j_o+\frac{1}{2} ; j_{o}+\frac{1}{2} }}{\Delta_x}\nonumber \\ -
v\frac{{\rho_{22}}_{j_x,j_y+\frac{1}{2} ;j_{x'}+\frac{1}{2},j_{y'}+ 1}^{j_o ; j_{o}} -{\rho_{22}}_{j_x,j_y+\frac{1}{2} ;j_{x'}+\frac{1}{2},j_{y'}}^{j_o ; j_{o} } }{\Delta_{y'}}+
iv\frac{{\rho_{22}}_{j_x,j_y +\frac{1}{2};j_{x'}+1,j_{y'}+\frac{1}{2}}^{j_o ; j_{o}  } -{\rho_{22}}_{j_x,j_y +\frac{1}{2};j_{x'},j_{y'}+\frac{1}{2}}^{j_o ; j_{o} } }{\Delta_{x'}}~,
\end{array}
\end{align}

\begin{align}
\begin{array}{l}
\frac{{\rho_{22}}_{j_x,j_y+\frac{1}{2};j_{x'},j_{y'}+\frac{1}{2}}^{j_o + 1; j_{o} + 1} - {\rho_{22}}_{j_x,j_y+\frac{1}{2};j_{x'},j_{y'}+\frac{1}{2}}^{j_o ; j_{o} } }{\Delta_o} = \nonumber \\
\frac{1}{i \hbar}({m_-}_{j_x, j_y+\frac{1}{2}}^{j_o+\frac{1}{2}} - {m_-}_{ j_{x'}, j_{y'}+\frac{1}{2}}^{ j_{o}+\frac{1}{2}}) \frac{{\rho_{22}}_{j_x,j_y+\frac{1}{2};j_{x'},j_{y'}+\frac{1}{2}}^{j_o + 1; j_{o}+1} + {\rho_{22}}_{j_x,j_y+\frac{1}{2};j_{x'},j_{y'}+\frac{1}{2}}^{j_o ; j_{o} }}{2}\nonumber \\ -
v\frac{{\rho_{12}}_{j_x,j_y + 1;j_{x'},j_{y'}+\frac{1}{2}}^{j_o+\frac{1}{2} ; j_{o}} -{\rho_{12}}_{j_x,j_y ;j_{x'},j_{y'}+\frac{1}{2}}^{j_o+\frac{1}{2} ; j_{o}} }{\Delta_y}+
iv\frac{{\rho_{12}}_{j_x+\frac{1}{2},j_y+\frac{1}{2} ;j_{x'},j_{y'}+\frac{1}{2}}^{j_o+\frac{1}{2} ; j_{o}+1} -{\rho_{12}}_{j_x-\frac{1}{2},j_y +\frac{1}{2};j_{x'},j_{y'}+\frac{1}{2}}^{j_o+\frac{1}{2} ; j_{o}+1 }}{\Delta_x}\nonumber \\-
v\frac{{\rho_{21}}_{j_x,j_y+\frac{1}{2} ;j_{x'},j_{y'}+1}^{j_o +1; j_{o}+\frac{1}{2}} -{\rho_{21}}_{j_x,j_y+\frac{1}{2} ;j_{x'},j_{y'}}^{j_o+1 ; j_{o}+\frac{1}{2} } }{\Delta_{y'}}-
iv\frac{{\rho_{21}}_{j_x,j_y+\frac{1}{2} ;j_{x'}+\frac{1}{2},j_{y'}+\frac{1}{2}}^{j_o +1; j_{o}+\frac{1}{2}  } -{\rho_{21}}_{j_x,j_y+\frac{1}{2} ;j_{x'}-\frac{1}{2},j_{y'}+\frac{1}{2}}^{j_o +1; j_{o}+\frac{1}{2} } }{\Delta_{x'}}~.
\end{array}
\end{align}

\end{appendix}

\end{document}